\documentclass[11pt, DIV10,a4paper]{article}
\usepackage{color}
\usepackage{float}
\usepackage[bf,small]{caption2}
\usepackage{comment}
\usepackage{amsmath}
\usepackage{bigints}
\usepackage{rotating, booktabs}
\usepackage{amsopn}
\usepackage{amsfonts}
\usepackage{amsthm}
\usepackage{amssymb}
\usepackage{amsbsy}
\usepackage{multirow}
\usepackage{slashbox}
\usepackage{natbib}
\usepackage{tocbibind}
\usepackage[a4paper,colorlinks,breaklinks,bookmarksopen,bookmarksnumbered]{hyperref}
\usepackage[refpage]{nomencl}
\usepackage{makeidx}
\usepackage{pst-all}
\usepackage{epsfig,psfrag}
\usepackage{graphicx}
\usepackage{amsmath}
\usepackage{epstopdf}
\usepackage{tabularx}
\usepackage{subfigure}
\usepackage{setspace}
\usepackage[mathscr]{euscript}
\usepackage[margin=1in]{geometry}
\usepackage{amsfonts} % Mathesymbole
\usepackage{calc}

\newcommand{\ueq}[1][]{%
  \if\relax\detokenize{#1}\relax
    \sbox0{$\underbrace{=}_{}$}%
    \mathrel{\mathmakebox[\wd0]{=}}
  \else
    \mathrel{\underbrace{=}_{\mathclap{#1}}}
  \fi}

\newcommand {\ctn}{\cite}

\usepackage{url}
\renewcommand{\d}{\ensuremath{\delta}}
\newcommand{\e}{\ensuremath{\epsilon}}

\newcommand{\boldm}{\mathbf m}

\newcommand{\be}{\pmb\e}

\newtheorem{theorem}{Theorem}

\newcommand{\topline}{\hrule height 1pt width \textwidth \vspace*{2pt}}
\newcommand{\botline}{\vspace*{2pt}\hrule height 1pt width \textwidth \vspace*{4pt}}
\newtheorem{algo}{Algorithm} 

\numberwithin{equation}{section}
\numberwithin{algo}{section}
\numberwithin{table}{section}
\numberwithin{figure}{section}

\newtheorem{lemma}{Lemma}[section]

\normalsize
\begin{document}

\title{\textbf{A Brief Review of Optimal Scaling of the Main MCMC Approaches and Optimal Scaling of Additive 
%Transformation Based Markov Chain Monte Carlo 
TMCMC Under Non-Regular Cases}}
\author{ Kushal Kr. Dey$^{\dag}$ , Sourabh Bhattacharya$^{\ddag, +}$ }
\date{}
\maketitle
\begin{center}
$^{\dag}$  University of Chicago \\
$^{\ddag}$ Indian Statistical Institute\\
$+$ Corresponding author:  \href{mailto: bhsourabh@gmail.com}{bhsourabh@gmail.com}
%\href{mailto: kshldey@gmail.com}{kshldey@gmail.com}
%,  \href{mailto: bhsourabh@gmail.com}{bhsourabh@gmail.com}\\
\end{center}

\begin{abstract}

Transformation based Markov Chain Monte Carlo (TMCMC) was proposed by 
\ctn{Dutta13} as an efficient alternative to the Metropolis-Hastings  algorithm, %Random Walk Metropolis (RWM) algorithm, 
especially in high dimensions.
The main advantage of this algorithm is that it simultaneously updates all components of a high dimensional 
parameter using appropriate move types defined by  deterministic transformation of a single random variable. This results in reduction in time complexity at each step of the chain and enhances the acceptance rate. \\

In this paper, we first provide a brief review of the optimal scaling theory for various existing MCMC approaches, 
comparing and contrasting them with the corresponding TMCMC approaches.The optimal scaling of 
the simplest form of TMCMC, namely \emph{additive TMCMC}, has been  studied extensively for
 the Gaussian proposal density in \ctn{Dey13}. Here, we discuss diffusion-based optimal scaling
 behavior of additive TMCMC for non-Gaussian proposal densities -- in particular, uniform, Student's $t$ 
 and Cauchy proposals. Although we could not formally prove our diffusion result for the Cauchy proposal, simulation based results lead us to {\it conjecture} that at least the recipe for obtaining general optimal scaling and optimal acceptance rate holds for the Cauchy case as well. We also consider diffusion based optimal scaling of TMCMC when the target density is discontinuous. Such non-regular situations have been studied in the case of Random Walk 
 Metropolis Hastings (RWMH) algorithm by \ctn{NealRoberts11} using expected squared jumping 
 distance (ESJD), but the diffusion theory based scaling has not been considered.  \\

%In the case of the Random Walk metropolis (RWM) algorithm these non-regular situations have been studied 
%
%Although we could not formally prove our diffusion result for the Cauchy proposal, simulation based results
%lead us to {\it conjecture} that at least the recipe for obtaining general optimal scaling and optimal acceptance rate 
%holds for the Cauchy case as well.
%that emerge 
%from our general diffusion result, which we prove for non-Cauchy proposals, holds for the Cauchy case as well. 

We compare our diffusion based optimally scaled TMCMC approach with the ESJD based optimally scaled RWM with 
simulation studies involving several target distributions and proposal distributions including the challenging Cauchy 
proposal case, showing that additive TMCMC outperforms RWMH in almost all cases considered.  \\

%with our conjectured optimal scale and acceptance rate.

\end{abstract}

\section{Introduction}

Markov Chain Monte Carlo (MCMC) techniques have revolutionized the statistical literature over the past two decades. 
It is extensively used today in Bayesian computation, systems biology, statistical physics, among many other fields. 
The simplest and the most popular MCMC technique in high dimensions is the Random Walk Metropolis Hastings (RWMH) algorithm. In this algorithm, at each iteration of the chain, a move is suggested based on a proposal density centered at the current position of the chain.

%However this 
%algorithm updates each co-ordinate separately and this leads to high computational burden and more importantly, 
%extremely poor acceptance rate in high dimensions $\--$ meaning the chain fails to move for long stretches of time.
%
%
In the RWMH algorithm, the most popular choice of proposal density is the Gaussian distribution. However the variance or the scaling factor of this Gaussian proposal density is of utmost importance. If the variance is small, the magnitude of jumps of the chain would be smaller and the chain converges slowly. If the variance is large, we end up rejecting too 
many proposed moves. Considering a diffusion based approach, \ctn{Roberts97a} proposed optimal scaling 
(variance) of the Gaussian proposal for target distributions with $iid$ components. Later, optimal scalings were derived for more general classes of target densities (see \ctn{Bedard2007}, \ctn{Pillai2011}, \ctn{BedRose}, \ctn{Bedard2009}). The optimal acceptance rate, corresponding to the optimal scaling, for most set-ups considered, is 0.234. 

In most high-dimensional and realistic scenarios, the RWM algorithm, as well as other Metropolis Hastings (M-H) algorithms exhibit relatively poor acceptance rates when all the variables are jointly updated at a time. Sequential updating  can maintain high acceptance rates, but can be computationally burdensome in the extreme. Moreover, such algorithms usually have poor mixing properties due to high posterior correlations between the parameters.
In order to counter these problems effectively, \ctn{Dutta13} introduced the general Transformation based Markov Chain Monte Carlo (TMCMC) algorithm. 
In a nutshell, TMCMC constructs appropriate ``move types", within which simple deterministic transformations of a single random variable is used to simultaneously update all the parameters. 

%a single random variable 
%having arbitrary distribution on some relevant support, is used to
%give simple deterministic transformations to all the parameters, which are then updated simultaneously in a single block.
This strategy has been shown to dramatically improve the acceptance rate and reduce computational burden. 
Properties like aperiodicity, Harris recurrence, irreducibility and geometric ergodicity of the additive TMCMC algorithm have already been studied in great detail; see \ctn{Dutta13}, \ctn{Dey13b}. All these studies show TMCMC to be a competent alternative to RWM, specially when the dimensionality is very high. 

We briefly describe TMCMC in the next section. 

\section{TMCMC and Optimal Scaling Theory}
\label{sec:definition}

Consider simulation from a $d$ dimensional distribution
%(usually $\mathbb{R}^{d}$, where $\mathbb R$ is the real line), 
and assume that we are currently at a point 
$x= (x_{1}, \ldots, x_{d})$.
Let us define the $d$-dimensional random vector $b=(b_{1}, \ldots, b_{d})$, such that, for $i=1,\ldots,d$, 
\begin{equation}
b_{i} =\left\{\begin{array}{ccc} +1 & \mbox{with probability} & p_i; \\
  0 & \mbox{with probability} & 1-p_i-q_i;\\
 -1 & \mbox{with probability} & q_i,
 \end{array}\right.
 \label{eq:b_tmcmc}
\end{equation}
where, for each $i$, $0<p_i,q_i<1$ such that $p_i+q_i\leq 1$.
Let $\epsilon\sim \varrho(\epsilon)=\tilde\varrho(\epsilon)I_{\mathbb S}(\epsilon)$, 
where $\tilde\varrho (\cdot)$ is any arbitrary density
supported on some suitable space $\mathbb S$; here $I_{\mathbb S}(\cdot)$ denotes the indicator function of $\mathbb S$.

TMCMC uses moves of the following type:
\begin{equation}
(x_{1}, \ldots, x_{d}) \rightarrow (T^{b_1}(x_{1},\epsilon), \ldots, T^{b_d}(x_{d},\epsilon)), 
\label{eq:tmcmc_move}
\end{equation}
where $T^{+1}(x_i,\epsilon)$, the forward transformation to coordinate $x_i$, and $T^{-1}(x_i,\epsilon)$, the backward
transformation to $x_i$, are bijective for fixed $\epsilon$ and injective  
for fixed $x_i$, satisfying 
\begin{equation}
T^{+1}(T^{-1}(x_i,\epsilon),\epsilon)=T^{-1}(T^{+1}(x_i,\epsilon),\epsilon)=x_i.
\label{eq:transformation1}
\end{equation}
The transformation  
\begin{equation}
T^{0}(x_i,\epsilon)\equiv x_i,~\forall\epsilon\in\mathbb S, 
\label{eq:transformation2}
\end{equation}
indicates no change to 
the coordinate $x_i$ while updating the vector $x=(x_{1}, \ldots, x_{d})$ to $x^*=\mathcal T_{b}(x,\epsilon)$,
where $\mathcal T_{b}(x,\epsilon)$ denotes the updated vector $(T^{b_1}(x_{1},\epsilon), \ldots, T^{b_d}(x_{d},\epsilon))$.
Assuming for simplicity of illustration that $p_i=q_i$ for $i=1,\ldots,d$, 
move (\ref{eq:tmcmc_move}) is to be accepted with probability
\begin{equation}
\alpha=\min\left\{1,\frac{\pi(x^*)}{\pi(x)}J^b(x,\epsilon)\right\},
\label{eq:acc_tmcmc_general}
\end{equation}
where $J^b(x,\epsilon)=\left|\frac{\partial(\mathcal T^b(x,\epsilon),\epsilon)}{\partial(x,\epsilon)}\right|$
is the Jacobian of the transformation associated with $\mathcal T^b$.
For general $(p_1,\ldots,p_d)$ and $(q_1,\ldots,q_d)$, the acceptance ratio depends upon 
these probabilities; see \ctn{Dutta11}.

For a wide range of target densities, \ctn{Dey13} derived the optimal scaling of the TMCMC algorithm with additive transformation.

\begin{equation}
T^{+1}(x_i,\epsilon) = x_i + \epsilon  \hspace{1 cm} T^{-1}(x_i,\epsilon) = x_i - \epsilon  \hspace{0.5 cm}  i=1,2,\cdots, d 
\label{eq:addtransformation}
\end{equation}

The optimal acceptance rate for the optimally scaled additive TMCMC algorithm was found to be 0.439, in contrast with 0.234, the optimal acceptance rate of the RWM algorithm.  Also the diffusion speed for TMCMC was found to be more robust to the choice of scaling, compared to RWM algorithm. Indeed,
even if the choice of the scale is suboptimal, the diffusion speed of TMCMC is not much affected, while, on the other hand,
that of RWM is significantly adversely affected by sub-optimal scalings. Since in complex, realistic problems, 
determination
of the exact optimal scaling can prove to be a difficult exercise, this robustness property of TMCMC is a strong advantage. \\

In all the above considerations, it was inherently assumed that the proposal distribution was Gaussian. 
A common criticism of the Gaussian proposal is that it is light-tailed and hence exploration of the state space 
would be slow. Starting from an initial point $x_{0}$, the chain would usually move to points
close to $x_{0}$, and in the rare cases when it makes a jump of large magnitude to some point $y$, distant 
from $x_{0}$, the acceptance rate $\min \left\{1,\frac{\pi(y)}{\pi(x_{0})} \right\}$ would usually turn out to be very small, 
and hence the probability of accepting such a jump would be very low. This is one of the prime reasons why the 
RWM or the TMCMC chain with the Gaussian proposal have slow convergence rate and also high autocorrelation time.
 %(see Figure \ref{fig:figex1}). 
%\begin{figure}%[htp]
%\centering
%\includegraphics[width=10cm,height=6cm]{acf}
%\caption{ACFs of RWM and additive TMCMC for dimension 5 and Gaussian proposal with proposal variance 0.01.
%The target density is $N(0,1)$, the standard normal distribution. The successive autocorrelations for both 
%RWM and TMCMC are quite high. In fact for TMCMC, the ACFs are found to be uniformly higher than RWM.}
%\label{fig:figex1}
%\end{figure}

%\subsection{Alternative optimal scaling approach for RWM using expected squared jumping distance}
%\label{subsec:esjd}
One way to resolve the aforementioned problem is to consider 
the uniform or heavy tailed proposal distributions like the Cauchy distribution instead of the light tailed Gaussian proposal. 
However, with the Cauchy proposal distribution, the moments are not defined and hence the Taylor's series expansions necessary for proving diffusion based optimal scaling
results are no longer valid.
This is the case even if the usual regularity conditions (see, for example, Theorem 4.1 of \ctn{Dey13}
in the TMCMC context and \ctn{Roberts97a} in the context of RWM)
are satisfied.

Additionally, if some of the regularity conditions are violated, for example, 
if the support of the target density is bounded (discontinuous target density on $\mathbb R^d$, where
$\mathbb R$ is the real line and $d$ is the dimensionality of the target distribution), the problem of optimal 
scaling poses further challenges. 

To avoid these technical difficulties associated with the traditional diffusion based approach, 
\ctn{NealRoberts11} obtained optimal scaling for RWM corresponding to several non-Gaussian proposal densities 
by maximizing the expected squared jumping distance (ESJD), defined by 
\begin{equation}
ESJD = E \left [ \sum_{i=1}^{d} ( X_{1i}-X_{0i})^{2} \right ]. 
\label{eq:esjd}
\end{equation}
In the Gaussian proposal case \ctn{NealRoberts11} show that their ESJD based approach
coincides with the diffusion based approach.

In this article, we extend the diffusion based approach to optimal scaling of additive TMCMC in 
situations where (a) all the regularity conditions of Theorem 4.1 of \ctn{Dey13} are satisfied
but the proposal distribution is non-Gaussian, and (b) the non-regular
cases consisting of target densities with bounded support, the proposal distribution being non-Gaussian. 
Before we formalize our approach, we first provide a brief review of optimal scaling theory
for various approaches of MCMC, including TMCMC, to acquaint the readers with the basic concepts. 
Thus, our contribution in this article is two-fold: reviewing and discussing the optimal scaling literature for varieties of
MH and TMCMC based methods, and developing a novel diffusion based approach to 
optimal scaling in non-regular cases for additive TMCMC.

\section{An overview of optimal scaling theory for various existing MCMC approaches}
\label{sec:overview}

\subsection{Optimal scaling for the RWM approach}
\label{subsec:MH}
Assume that $\pi:\mathbb R^d\mapsto\mathbb R_+$ is the target density, and $x_t=(x_{t,1},\ldots,x_{t,d})$ is the 
MCMC realization at the $t$-th iteration, and that at the next iteration, the value 
$y_{t+1}=(y_{t,1},\ldots,y_{t,d})$ is proposed from some
density $q(x_t,\cdot)$, where, for any $x,y$, $q(x,y)$ is the conditional density of $y$ given $x$.
The Metropolis Hastings (MH) approach either accepts $x_{t+1}=y_{t+1}$ with probability
\begin{equation}
\alpha(x_t,y_{t+1})=\min\left\{1,\frac{\pi(y_{t+1})q(y_{t+1},x_t)}{\pi(x_t)q(x_t,y_{t+1})}\right\},
\label{eq:MH_acceptance_prob}
\end{equation}
or remains at the current value with $x_{t+1}=x_t$. Note that if $q(x_t,y_{t+1})=q(y_{t+1},x_t)$, that is, if $q$ is symmetric, then the ratio $q(y_{t+1},x_t)/q(x_t,y_{t+1})$ cancels in the acceptance ratio, thus simplifying the proceedings. The random walk proposal of the form $q(x,y)\equiv q(|y-x|)$, where $q(\cdot)$ is symmetric about zero, 
is an example of such a symmetric proposal, and has become the default proposal mechanism for MCMC simulation, and is known as the RWM algorithm. Thus, in RWM, $y_{t+1}$ is of the form $y_{t+1}=x_t+\epsilon_{t+1}$, where
$\left\{\epsilon_t:t=1,2,\ldots\right\}$ are $iid$ with some symmetric distribution. The most popular choice of
such symmetric distribution is $N_d(0,\sigma^2I_d)$, the $d$-variate normal distribution with mean zero and
covariance matrix $\sigma^2I_d$, where $\sigma^2>0$ and $I_d$ is the $d$-dimensional identity matrix. The convergence
properties of the resulting RWM crucially depend on the chosen value of $\sigma^2$; too small values
leads to large acceptance rates but very little movement of the chain, and too large values lead to small acceptance
rates and only occasional movement of the chain, both of which slow down convergence, and hence, must be avoided. 
This so-called ``Goldilocks principle" is not a modern day observation; indeed, this has been recognized even by 
\ctn{Metropolis53}, who assumed the $U(-a, a)$ distribution of the $\epsilon_t$'s with $a>0$, and noted that too small 
or too large values of $a$ must be avoided. 

\subsubsection{The $iid$ target density set-up}
\label{subsubsec:iid_target}
Modern day research has of course attempted to make precise statements regarding the optimal value of $\sigma^2$, 
when $d$ is large enough. This study was initiated by \ctn{Roberts97} who considered a simple $iid$ product target density
of the form $\pi(x)=\prod_{i=1}^df(x_i)$ and a normal random proposal with $\sigma^2$ of the form $\frac{\ell^2}{d}$. 
In this situation, letting $U^d_t=X_{[dt],1}$ (where $[\cdot]$ denotes the integer part) be the sped up first
component of the $d$-dimensional Markov chain, which proposes $d$ jumps in every time unit, it can be shown
that under appropriate sufficient conditions, $U^d_t$ eventually becomes a continuous time diffusion process 
as $d\rightarrow\infty$, which has stationary distribution $f$ and speed measure 
$g(\ell)=2\ell^2\Phi\left(-\sqrt{\mathbb I}\ell/2\right)$, where 
$\mathbb I=E_f\left(\frac{f'(X)}{f(X)}\right)^2=\int_{-\infty}^{\infty}\left(\frac{f'(x)}{f(x)}\right)^2f(x)dx$.
The speed measure is related to the autocorrelation of the underlying Markov chain; in fact, high speed 
is equivalent to low autocorrelation (see \ctn{Roberts01}). Thus, it makes sense to maximize the speed measure
with respect to $\ell$. As such, the optimal value of $\ell$ is given by $\ell_{opt}=2.381/\sqrt{\mathbb I}$
and the optimal acceptance rate is given by $2\Phi\left(-\sqrt{\mathbb I}\ell_{opt}/2\right)\approx 0.234$.
This optimal acceptance rate need not be strictly enforced, however, as \ctn{Roberts01} demonstrate, using 
a measure of efficiency which is the reciprocal of integrated autocorrelation time, 
that the RWM proposal may be tuned to achieve an acceptance rate between $0.15$ to $0.5$, which would make 
the algorithm around 80\% efficient.

\subsubsection{The set-up where target density is the product of independent but non-identical densities}
\label{subsusbec:non_iid_target}

Although the aforementioned optimal scaling theory is built on the assumption of the simple (and unrealistic)
assumption of the product of $iid$ densities as the target, this has been extended to more realistic set-ups, such as
product of independent but non-identical densities with special forms. \ctn{Roberts01} considered the form $\pi(x)=\prod_{i=1}^dC_if(C_ix_i)$, where $C_1,\ldots,C_d$ are $iid$ realizations
from some distribution. In this case, the optimal scaling result for the $iid$ set-up continues to hold, albeit
the diffusion speed is reduced due to division by an ``inhomogeneity factor" given by 
$c=E\left(C^2_1\right)/E\left[\left(C_1\right)\right]^2$, which is greater than or equal to one. This factor
is responsible for slowing down the algorithm as the variability among $C_1,\ldots,C_d$ increases.

\ctn{Bedard2007}, \ctn{Bedard2008}, \ctn{BedRose} considered a similar framework, but different powers of $d$
for the co-ordinate wise target densities. Their main result is that if the individual components are dominated
by the sum of all the components, then the optimal acceptance rate remains $0.234$, but on the other hand,
if any component is comparable to the sum, then the optimal acceptance is reduced.

\subsubsection{The dependent set-up}
\label{subsubsec:dependent_target}
Although the aforementioned optimal scaling theories assume the target to be at most inhomogeneous
product of $d$ densities, as shown in \ctn{Rosenthal11} (see also \ctn{Roberts01}), the theory of \ctn{Roberts01} for
independent but non-identical target density can be adapted to the case of $d$-variate normal
target distributions. Indeed, following \ctn{Rosenthal11}, let us assume that the target is $N\left(0,\Sigma\right)$, 
where $\Sigma$ 
is a $d$-dimensional covariance
matrix, and the proposal is of the form $y_{t+1}=x_t+\epsilon_{t+1}$, where 
$\epsilon_t\stackrel{iid}{\sim}N\left(0,\tilde\Sigma\right)$,
where $\tilde\Sigma$ is the appropriate covariance matrix to be determined by the optimal scaling theory. It can be seen
that the problem can be equivalently formulated as considering the target to be $N\left(0,\Sigma\tilde\Sigma^{-1}\right)$
and the normal random walk covariance to be the $d$-dimensional identity matrix. Then, in the 
form $\pi(x)=\prod_{i=1}^dC_if(C_ix_i)$, $C_i=\sqrt{\lambda_i}$, where $\lambda_1,\ldots,\lambda_d$ are the eigenvalues
of $\Sigma\tilde\Sigma^{-1}$. As $d\rightarrow\infty$, this corresponds to the case where $C_1,\ldots,C_d$ are random
with $E(C_1)=\frac{1}{d}\sum_{i=1}^d\sqrt{\lambda_i}$ and $E(C^2_1)=\frac{1}{d}\sum_{i=1}^d\lambda_i$. In this case,
the inhomogeneity factor is approximately given by 
$c=d\left(\sum_{i=1}^d\lambda_i\right)/\left(\sum_{i=1}^d\sqrt{\lambda_i}\right)^2$.
It is thus clear that the diffusion speed is maximized when the above eigenvalues are all equal, which implies
that one must set $\tilde\Sigma\propto\Sigma$. Applying the optimal scaling theory for the $iid$ case one then obtains the
value of the proportionality constant to be $\left(2.38\right)^2/d$.

\ctn{Pillai2011} consider a more realistic and general dependent set-up where the joint target density 
is absolutely continuous with respect to a Gaussian measure, and even in their case, 
the optimal acceptance rate turned out to be $0.234$ for normal RWM proposals.

\subsection{Optimal scaling for Metropolis within Gibbs}
\label{subsec:mh_gibbs}

%The main notion of Gibbs sampling is to update one or multiple components of a 
%multidimensional random vector conditional on the remaining components. 
\ctn{NealRoberts} investigated optimal scaling in the Metropolis within Gibbs context, where in
any given iteration, only a fixed proportion $c_{d}$ of the $d$ coordinates are updated using RWM, leaving the remaining
co-ordinates unchanged. 
Here $c_{d}$ is a function of $d$ and it is assumed that as $d \rightarrow \infty$, 
$c_{d} \rightarrow c$, for some $0<c\leq 1$. 
To analytically represent the transitions, first let for $i=1,\ldots,d$, 
%we define an indicator function $\chi_{i}$ for $ i= 1,\ldots,d$. For fixed $d$, 
\begin{eqnarray}
\mathbb{\chi}_{i} &=& 1 \hspace{0.5 cm} \mbox{if transition takes place in the} \ \  i^{th} \ \ \mbox{coordinate} 
\nonumber \\
&=& 0 \hspace{0.5 cm} \mbox{if no transition takes place in the}\ \  i^{th} \ \ \mbox{coordinate}. 
\label{eq:chi_definition}
\end{eqnarray}
Then, 
\begin{equation}
P (\mathbb{\chi}_{i}= 1) = c_{d}; \ \ i =1,\ldots,d,
\label{eq:chi_distribution}
\end{equation}
and the transition is given by
\begin{equation}
(x_{1},\ldots,x_{d}) \rightarrow (x_{1}+\mathbb{\chi}_{1}\epsilon_1,\ldots,x_{d}+ \mathbb{\chi}_{d}\epsilon_d),
\label{eq:transition1}
\end{equation}
where, for $i=1,\ldots,d$, $\epsilon_i\stackrel{iid}{\sim}N\left(0,\frac{\ell}{d}\right)$.
Assuming the target density to be a product of $iid$ densities, \ctn{NealRoberts} obtained, in the RWM within
Gibbs set-up, the optimal acceptance rate $0.234$. It can be verified that the same optimal acceptance
rate is achieved even for the target densities that are products of independent but non-identical, and for
dependent target densities discussed above.

\ctn{Dey13} consider a similar set-up under the additive TMCMC within Gibbs premise. In their case,
the transition can be represented as
\begin{equation}
(x_{1},\ldots,x_{d}) \rightarrow (x_{1}+\mathbb{\chi}_{1}b_{1}\epsilon,\ldots,x_{d}+ \mathbb{\chi}_{d}b_{d}\epsilon),
\label{eq:transition2}
\end{equation}
where $\epsilon\equiv\frac{\ell}{\sqrt{d}}\epsilon^*$, with $\epsilon^*\sim N(0,1)I_{\{\epsilon^*>0\}}$.
\ctn{Dey13} show that in this case, the optimal acceptance rate is $0.439$ for all the aforementioned
forms of the target densities. In the simulation studies reported in \ctn{Dey13}, optimally scaled additive
TMCMC considerably outperformed optimally scaled RWM when all the variables are updated in every iteration
in terms of various measures of convergence and mixing, in particular, the Kolmogorov-Smirnov distance 
of the Markov chains from the target distributions.
Hence, one can expect far superior performance of TMCMC even if a proportion of the variables is updated in every iteration.

\subsection{Optimal scaling for the Metropolis-Adjusted Langevin Algorithm (MALA)}
\label{subsec:mala}

One way to simulate from the target density $\pi$ without resorting to the traditional MH method is
to simulate from the discretized version of some appropriate diffusion equation having stationary
distribution $f$. Such an idea owes its origin in \ctn{Grenander94} and \ctn{Philips96}. In particular,
the Langevin diffusion $dx_t=dB_t+\frac{1}{2}\nabla\log \pi(x_t)dt$, where $B_t$ is the standard Brownian motion.
\ctn{Roberts98} note that the Langevin equation is the only non-explosive diffusion which is reversible
with respect to $f$. Implementation of the Langevin equation proceed by discretization:
$x_{t+1}=x_t+\frac{\sigma^2}{2}\nabla\log \pi\left(x_t\right)+\sigma\epsilon_t$, where $\epsilon_t$
is generated from the $d$-dimensional normal with mean zero and identity covariance matrix. In the above,
$\sigma^2$ is associated with the size of discretization, which is to be appropriately chosen.

However, the discretized version does not necessarily mimic the behaviour of the original diffusion equation.
\ctn{Roberts96} note that the discretized chain may even be transient if 
$\underset{x\rightarrow -\infty}{\lim}\sigma^2\nabla\log f(x)|x|^{-1}$ and 
$\underset{x\rightarrow \infty}{\lim}\sigma^2\nabla\log f(x)|x|^{-1}$ exist and larger than 1 and smaller
than -1, respectively. A way to rectify this is to consider the discretized version as a proposal
distribution for the MH method in the usual way; this has been suggested by \ctn{Besag94}.
The MALA based MH algorithm is given as follows.
\begin{algo}\label{algo:mala_mh} \topline MALA \botline \normalfont \ttfamily
\begin{itemize}
 \item Assume that the current state is $x= (x_{1}, \ldots, x_{d})$. 
 \item Propose $y\sim N\left(x+\frac{\sigma^2}{2}\nabla\log \pi\left(x\right),\sigma^2I_d\right)$ as the proposed value.
 \item Accept $y$ with probability
\begin{equation}
\alpha=\min\left\{1,\frac{\pi(y)}{\pi(x)}\times
\frac{\exp\left\{-\frac{1}{2\sigma^2}\left(y-x-\frac{\sigma^2}{2}\nabla\log \pi\left(x\right)\right)^2\right\}}
{\exp\left\{-\frac{1}{2\sigma^2}\left(x-y-\frac{\sigma^2}{2}\nabla\log \pi\left(y\right)\right)^2\right\}}\right\}.
\label{eq:mala_mh_accept}
\end{equation}
\item Accept $x$ with the remaining probability.
\end{itemize}
\botline \rmfamily
\end{algo}
\ctn{Robert04} show that the discretized proposal can be naturally derived by considering a Laplace approximation
perspective. 

The optimal scaling of $\sigma$ has been derived by \ctn{Roberts98} by considering $\sigma^2=\ell^2/d^{1/3}$.
This scaling order originated in physics (\ctn{Kennedy91}) and turned out to be relevant for the optimal
scaling investigation. The optimal acceptance obtained by \ctn{Roberts98} in the $iid$ set-up is 0.574, 
which is much higher than than for RWM. Even for the independent but the non-identical set-up considered
by \ctn{Roberts01}, the optimal acceptance rate turned out to be 0.574. Perhaps not surprisingly, 
the acceptance rate remains the same in the general dependent set-up where the joint target density
is absolutely continuous with respect to a Gaussian measure; see \ctn{Pillai2012}. 

Thus, in all the cases considered so far, the MALA significantly outperforms in terms of acceptance rate.
However, MALA is not geometrically ergodic when $\nabla f(x)\rightarrow 0$ as $\|x\|\rightarrow\infty$ 
(\ctn{Roberts96}), although in this situation the MALA resembles the RWM, which is geometrically ergodic
under relevant sufficient conditions (see, for example, \ctn{Jarner00}). Thus, MALA need not always be superior 
to RWM in terms of performance. 

It is useful to note that a TMCMC version of the Langevin diffusion can also be considered as follows.   
Suppose that we are simulating from a $d$ dimensional space (usually $\mathbb{R}^{d}$). 
Let us define $d$ random variables $b_{1}, \ldots, b_{d}$ in the same way as (\ref{eq:b_tmcmc}). 
%such that, for $i=1,\ldots,d$, 
%\begin{equation}
%b_{i} =\left\{\begin{array}{ccc} +1 & \mbox{with probability} & p_i; \\
% -1 & \mbox{with probability} & 1-p_i.
%\end{array}\right.
%\label{eq:b}
%\end{equation}
Then TMCMC based on the discretized Langevin proposal, which we refer to as TMCMC-adjusted
Langevin algorithm (TALA) is given as follows:
\begin{algo}\label{algo:mala_tmcmc} \topline TALA \botline \normalfont \ttfamily
\begin{itemize}
 \item Assume that the current state is $x= (x_{1}, \ldots, x_{d})$ and let $b_1$ and 
 $\epsilon_1\sim q(\cdot)I_{\{\epsilon_1>0\}}$
 be associated with the current proposed value, where $q(\cdot)$ is any arbitrary univariate density.
 \item Propose $b_2$ and $\epsilon_2\sim q(\cdot)I_{\{\epsilon_2>0\}}$. Set
 $y=x+\frac{\sigma^2}{2}\nabla\log \pi\left(x\right)+\sigma b_2\epsilon_2$ as the proposed value.
 \item Accept $y$ with probability
\begin{equation}
\alpha=\min\left\{1,\frac{P(b_1)}{P(b_2)}\times\frac{\pi(y)}{\pi(x)}\times\frac{q(\epsilon_1)}{q(\epsilon_2)}\right\},
\label{eq:mala_tmcmc_accept}
\end{equation}
where for any $b$ of the form (\ref{eq:b_tmcmc}), $P(b)$ denotes the probability of $b$.
\item Accept $x$ with the remaining probability.
\end{itemize}
\botline \rmfamily
\end{algo}
Observe that unlike the original TMCMC principle, the acceptance ratio is not free of the proposal density. In fact,
the ratio $q(\epsilon_1)/q(\epsilon_2)$ is an adjustment for the issue that for TALA we do not use the inverse of the
forward transformation to move backward using the same $\epsilon$ used in the forward direction, unlike the
original TMCMC principle. The reason for not using inversion (and the same $\epsilon$) is that bijection associated with the
transformation in this case is not assured for general target densities. However, unlike MALA, the acceptance ratio of 
TALA provided in (\ref{eq:mala_tmcmc_accept}) does not require evaluation
of the gradient, resulting in computational simplicity. Note that in practice the gradient is usually
approximated numerically, and indeed for simulation purpose a small margin of error is permissible, 
but it is desirable to evaluate the acceptance rate without any error. Thus, from this perspective, eliminating
the gradient based calculations is important, which TALA achieves. Also note that if $p_i=1/2$ for all $i$ in 
(\ref{eq:b_tmcmc}), 
then the ratio $P(b_1)/P(b_2)$ cancels in the acceptance ratio, resulting in further simplification. 

Optimal scaling for TALA is an interesting challenge which we shall handle. We anticipate
that the optimal acceptance rate of TALA will be much higher than that of MALA
because of the drastic dimension reduction achieved by updating all the variables using a single random variable.

\subsection{Optimal scaling in hybrid Monte Carlo}
\label{subsec:optimal_hmc}

The hybrid Monte Carlo (HMC) methods, introduced by \ctn{Duane87}, is a method of MCMC simulation from 
the target distribution $\pi$ that considers as proposal a discretized version of the solution
of the deterministic Hamiltonian equations from physics and uses the MH acceptance probability to accept
the proposed value. Briefly, %if $\pi(x)$ is the target distribution, 
one may imagine a
dynamical system where $x(t)\in\mathbb R^d$ is likened to the $d$-dimensional position vector of a body of particles
at time $t$. Also, let $v(t)=\dot{x}(t)=\frac{d{x}}{dt}$ be the speed vector of the
particles, $\dot{v}(t)=\frac{d{v}}{dt}$ be the acceleration vector, and 
$\Vec{F}$ be the force exerted on the particles. Thanks to Newton's law of motion, 
$\Vec{F}=\boldm\dot{v}(t)=(m_1\dot{v_1},\ldots,m_d\dot{v_d})(t)$, where $m\in\mathbb R^d$
is a mass vector. From the simulation perspective, the momentum vector, $p=m v$ %often used in classical mechanics,
may be interpreted as a set of auxiliary variables that facilitates simulation from $\pi(x)$.

The kinetic energy of the system is defined as $W(p)=p'M^{-1}p$, where $M$ is the mass matrix.
In general, $M$ is usually chosen to be a diagonal matrix. %$M=diag\{m_1,\ldots,m_d\}$.
The potential energy field of the system is defined as $U(x)=-\log\pi(x)$, which now
connects our target density of interest to the dynamical system. The total energy (Hamiltonian function)
is given by $H(x,p)=U(x)+W(p)$, which is used to build a joint distribution over the phase-space $(x,p)$.
The joint distribution is of the form
\begin{equation}
f(x,p)\propto\exp\left\{-H(x,p)\right\}=\pi(x)\exp\left(-p'M^{-1}p/2\right),
\label{eq:phase_space_dist}
\end{equation}
so that simulating jointly from $f(x,p)$ by some appropriate MCMC mechanism and discarding the 
corresponding simulations of $p$ yields samples from $\pi$. 

The essence of HMC lies in the construction of a novel proposal strategy that hinges upon 
Newton's law of motion, derived from the law of conservation of energy. These admit representation in the 
form of the Hamiltonian equations, given by 
\begin{eqnarray}
\dot{x}(t)&=&\frac{\partial H(x,p)}{\partial p}=M^{-1}p,\nonumber\\
\dot{p}(t)&=&-\frac{\partial H(x,p)}{\partial x}=-\nabla U(x),\nonumber
\end{eqnarray}
where $\nabla U(x)=\frac{\partial U(x)}{\partial x}$.
The above equations form the crux for an efficient proposal mechanism, but for being usable, discretization
is required. Indeed, these can be approximated by the so-called leap-frog algorithm (\ctn{Hockney70}), 
given by
\begin{align}
x(t+\d t)&=x(t)+\d tM^{-1}\left\{p(t)-\frac{\d t}{2}\nabla U\left(x(t)\right)\right\}\label{eq:frog1}
\\[1ex]
p(t+\d t)&=p(t)-\frac{\d t}{2}\left\{\nabla U\left(x(t)\right)+\nabla U\left(x(t+\d t)\right)\right\}\label{eq:frog2}
\end{align}
As such, given choices of $M$, $\d t$, and $L$, the HMC is then the following algorithm:
\begin{algo}\label{algo:hmc} \topline HMC \botline \normalfont \ttfamily
\begin{itemize}
 \item Initialise $x$ and draw $p\sim N(0, M)$.
 \item Assuming the current state to be $(x,p)$, do the following:
\begin{enumerate}
 \item Generate $p_1\sim N\left(0,M\right)$;
 \item Letting $(x(0),p(0))=(x,p_1)$, run the leap-frog algorithm for 
 $L$ time steps, to yield $(x_2,p_2)=\left(x(t+L\d t),p(t+L\d t)\right)$;
\item Accept $(x_2,p_2)$ with probability 
\begin{equation}
\min\left\{1,\exp\left\{-H(x_2,p_2)+H(x,p_1)\right\}\right\},
\label{eq:hmc_accept}
\end{equation}
and accept $(x,p_1)$ with the remaining probability.
\end{enumerate}
\end{itemize}
\botline \rmfamily
\end{algo}
In the above algorithm, it is not required to store simulations of $p$.
Detailed balance can be easily seen to hold by observing that the leapfrog algorithm is volume preserving 
(``sympletic") and time reversible. The other ergodic properties also easily follow.

The non-local behaviour of the leap-frog algorithm allows the algorithm to explore the state space
more efficiently compared to RWM. However, the tuning parameters of HMC, namely, $L$, $M$ and $\d t$
must be chosen carefully. For each dynamic evolution, \ctn{Cheung09} suggest selecting $L$ from a 
discrete uniform distribution on $\{1,\ldots,L_{\max}\}$, for some pre-chosen $L_{\max}$. This strategy bypasses
the issue of getting into a somewhat rare, but undesirable resonance condition (\ctn{Mackenzie89}). 
\ctn{Cheung09} also suggest selecting $M$ to be the identity matrix if the components of $x$ are of comparable scale,
which can be ensured by appropriate normalization at the initial stage.

The most challenging issue seems to be properly tuning the step size $\d t$ of the leap-frog algorithm, which 
affects the acceptance rate and convergence of the HMC algorithm in ways similar to that of the scale parameters
of RWM and MALA, and optimal choice of this parameter is of much importance. \ctn{Cheung09} suggest choosing
$\d t$ such that the empirical acceptance rate is at least $0.1$. Using heuristic arguments 
and calculations \ctn{Neal11} obtained the optimal acceptance rate $0.65$ for HMC for $\d t=O\left(d^{-4}\right)$, 
so that $\d t$ can be tuned to achieve the acceptance rate. The results obtained by \ctn{Neal11} are further validated by \ctn{Beskos13} who establish, in the case of $iid$ product density as the target, a formal theory of optimal 
scaling for HMC, considering $\d t=\ell\times d^{-4}$.

\ctn{Dutta13} show that HMC is a special case of TMCMC, where the momentum vector plays the role
of the random variables using which the relevant forward and inverse transformations are taken; 
in the Appendix we briefly touch upon the issue.
However, since the main essence of TMCMC is to update all the variables using transformations of a scalar
random variable, it is worth updating the momentum vector $p$ using a single random variable.
%It is worth mentioning that a modification of HMC given by Algorithm \ref{algo:hmc} can be envisaged
%using TMCMC based ideas, in the lines of TALA (Algorithm \ref{algo:tala}). 
In this regard, we provide the TMCMC based version of HMC in Algorithm \ref{algo:hmc_tmcmc}, 
where, for simplicity we consider additive TMCMC, noting that any
valid transformation satisfying the conditions stated in \ctn{Dutta13} may be considered.
\begin{algo}\label{algo:hmc_tmcmc} \topline TMCMC based HMC \botline \normalfont \ttfamily
\begin{itemize}
 \item Let $(x_1,p_1)$ be the current value. 
 Also, let $b_1$ with probability $P(b_1)$ and $\epsilon_1\sim q(\cdot)I_{\{\epsilon_1>0\}}$ be associated
 with the current value $p_1$. %and draw $p\sim N(0, M)$.
 \item Do the following:
\begin{enumerate}
 %\item Generate $p'\sim N\left(0,M\right)$;
 \item Propose $b_2$ with probability $P(b_2)$ and $\epsilon_2\sim q(\cdot)I_{\{\epsilon_2>0\}}$. Set
 $\tilde p_1=p_1+\sigma b_2\epsilon_2$ as the proposed value.
 \item Letting $(x(0),p(0))=(x_1,\tilde p_1)$, run the leap-frog algorithm for 
 $L$ time steps, to yield $(x_2,p_2)=\left(x(t+L\d t),p(t+L\d t)\right)$;
\item Accept $x_2$ with probability 
\begin{equation}
\min\left\{1,\frac{P(b_1)}{P(b_2)}\times \exp\left\{-H(x_2,p_2)+H(x_1,p_1)\right\}
\times\frac{q(\epsilon_1)}{q(\epsilon_2)}\right\},
\label{eq:hmc_tmcmc_accept}
\end{equation}
and store $\tilde p_1$ as the current value for the next iteration.
\item Else accept $x_1$ with the remaining probability and store $p_1$ as the current value for the next iteration.
\end{enumerate}
\end{itemize}
\botline \rmfamily
\end{algo}
Given fixed scalings of the additive TMCMC above, due to drastic dimension reduction of the
momentum vector $p$, one may expect higher optimal acceptance rate for
the TMCMC based HMC algorithm compared to the original HMC algorithm %Algorithm \ref{algo:hmc_tmcmc} 
with respect to optimal scaling of $\d t$. Because of dimension reduction, the TMCMC-fed HMC method is 
also expected to have diffusion speed
that is far more robust compared to that of the original HMC procedure, as in the case of optimal
scaling of additive TMCMC relative to RWM. If optimal scaling of both $\d t$ and $\sigma$ is desired, 
then new issues open up, and merits detailed investigation. 

\subsection{Multiple-try MCMC}
\label{subsec:multiple_try}

By multiple-try MCMC we mean the
MCMC algorithm that selects the next proposal from a set of available, perhaps dependent,
proposals. For MH-adapted versions of such an idea, see, for example, \ctn{Liu00}, 
\ctn{Liang10}, \ctn{Martino13}. To briefly describe the main idea based on MH we consider 
$w(x,y)=\pi(x)q(x,y)\lambda(x,y)$, where $\pi$ is the target density, $q(x,y)$ is an arbitrary proposal 
satisfying $q(x,y)>0$ if and only if
$q(y,x)>0$ and $\lambda(x,y)$ is an arbitrary symmetric non-negative function such that $\lambda(x,y)>0$ whenever $q(x,y)>0$.
If the current state is $x^{(t)}=x$, then the basic multiple-try MH for the $(t+1)$-th iteration is given as follows:
\begin{algo}\label{algo:mtm} \topline Multiple-try MH \botline \normalfont \ttfamily
\begin{itemize}
\item Draw $k$ realizations, $y_1,\ldots,y_k$, from $q(x,\cdot)$.
\item Select $y$ from the set $\left\{y_1,\ldots,y_k\right\}$ with probability proportional to
$w(y_j,x)=\pi(x)q(x,y_j)\lambda(x,y_j)$; $j=1,\ldots,k$. 
\item Obtain the $(k-1)$ auxiliary variables $\tilde x_1,\ldots,\tilde x_{k-1}$ from $q(y,\cdot)$, and let $\tilde x_k=x$.
\item Accept $y$ with probability
\[\alpha=\min\left\{1,\frac{w(y_1,x)+\cdots +w(y_k,x)}{w(\tilde x_1,y)+\cdots +w(\tilde x_k,y)}\right\}.\]
\end{itemize}
\botline \rmfamily
\end{algo}
When $\lambda(x,y)=1/q(x,y)$, $w(x,y)=\pi(x)$, and in this case, the above algorithm boils down to
oriental bias Monte Carlo (\ctn{Frenkel02}) for molecular simulation. For various other versions of
multiple try MCMC, see, for example, \ctn{Liu00} and \ctn{Bedard12}. In fact,
\ctn{Bedard12} investigated scaling analysis of many variations of the above multiple-try MH method when the target $\pi$
is the product of $iid$ densities, $w(x,y)=\pi(x)$, and when the proposals
are generated from multivariate normal random walk proposals. As to be expected, the scaling constant, 
the diffusion speed, and the acceptance
rate are increasing with $k$, the number of trial proposals. As we primarily investigated, the same issue holds in the 
corresponding TMCMC case, and the optimal acceptance rate tends to 1 as $k\rightarrow\infty$, independently
of the scale of the random walk proposal. Thus, when $k$ is very large, it seems that one can achieve virtually any desired 
diffusion speed simply by choosing the scaling constant large enough. Indeed, since the algorithm is convergent, 
the close to one acceptance rate implies that one can achieve almost $iid$ samples from the target $\pi$ with large enough $k$,
where $k$ must increase at a rate faster than the scaling constant. But this of course comes at a very high
computational cost, and it is debatable whether such a multiple-try strategy is worth in practice. 
\ctn{Bedard12} also investigated optimal scaling with alternative choices of $w(x,y)$, but the weights proportional
to the target density yielded the best results.

\subsection{Delayed rejection MCMC}
\label{subsec:delayed_rejection}

The delayed rejection MCMC, which has been introduced by \ctn{Tierney99}, attempts, at any given iteration
of the algorithm, to successively improve the proposal by generating a sequence of trial values from
possibly different proposal distributions till ultimate acceptance of a trial value or till a given number, $k$, 
of trial values are generated. Further development of the method was provided by \ctn{Mira01} for fixed-dimensional
problems and by \ctn{Green01} for variable-dimensional problems. Applications of delayed rejection MH can be
found in \ctn{Harkness00}, \ctn{Ums04}, \ctn{Raggi05}, \ctn{Haario06}, \ctn{Trias09}, etc. and optimal scaling
of this method for random walk proposals when $k=2$ and the target is the product of $iid$ densities, 
has been undertaken by \ctn{Bedard14}. 
The two-step delayed rejection MH is given by the following algorithm when $x$ is the current state of the chain:
\begin{algo}\label{algo:delayed_rejection} \topline Delayed rejection MH \botline \normalfont \ttfamily
\begin{itemize}
\item Draw $y_1$, from proposal distribution $q_1(x;\cdot)$.
\item Accept $y_1$ with probability
\[\alpha_1(x;y_1)=\min\left\{1,\frac{\pi(y_1)q_1(y_1;x)}{\pi(x)q_1(x;y_1)}\right\}.\]
\item If $y_1$ is rejected, generate another trial value $y_2$ from possibly another proposal $q_2(x,y_1;\cdot)$.
\item Accept $y_2$ with probability 
\[\alpha_2(x,y_1;y_2)=\min\left\{1,\frac{\pi(y_2)q_1(y_2;y_1)[1-\alpha_1(y_2;y_1)]q_2(y_2,y_1;x)}
{\pi(x)q_1(x;y_1)[1-\alpha_1(x;y_1)]q_2(x,y_1;y_2)}\right\}.\]
\end{itemize}
\botline \rmfamily
\end{algo}
When the proposals are random walks, \ctn{Bedard12} suggest two different scalings: relatively large scale for the first
attempt, and a smaller scale for the second attempt if the first attempt leads to rejection. They also consider
two set-ups for the two proposal distributions; in one set-up they assume that $y_2$ is generated independently
of $y_1$ and in the other they consider generating $y_2$ conditionally on $y_1$ using a deterministic transformation
such that $y_2$ is generated from $q_2(x,\cdot)$. The optimal scaling results obtained by \ctn{Bedard14} are, however,
not encouraging. In the first set-up where $y_1$ and $y_2$ are generated independently, they obtained $0.234$ 
as the optimal acceptance rate for the first acceptance rate,
namely $\alpha_1$, while the second acceptance rate $\alpha_2$ converges to zero, showing that given the first proposal,
the second move is useless. For the second, dependent proposal set-up, the optimal acceptance rates for both the 
stages turned out to be $0.234$, showing that there is no improvement of the acceptance rate in the second attempt,
perhaps signifying inadequate learning from the first attempt. Since delayed rejection methods necessarily involves
much computational burden compared to the traditional RWM, the discouraging results of \ctn{Bedard12} seem to put a question mark
on the usefulness of such methods. As can be anticipated, for additive TMCMC adaptation of delayed rejection, the
corresponding acceptance rates in the two proposal set ups of \ctn{Bedard14} would be $0.439$, and would not amount
to any improvement over the usual additive TMCMC.

\subsection{Optimal scaling in adaptive MCMC methods}
\label{subsec:optimal_adaptive}

The adaptive MCMC methods are concerned with proposal distributions that are updated in every iteration based on
progressive learning with the iterations. Thus, the chain is not Markov but is so designed that asymptotically it 
becomes Markov and converges to the  target distribution. Thus, adaptive MCMC is about a family of 
Markov kernels  $\left \{ P_{\sigma} \right \}_{\lambda\in\Lambda}$, each having the same stationary
distribution $\pi$, where $\Lambda$ is an appropriate set of possible tuning parameters associated with
the possible Markov kernels. Letting $\lambda_t$ be associated with the Markov kernel at the $t$-th iteration
and $A$ be any relevant Borel set, 
we have $$P\left(X_{t+1}\in A|X_t=x,\lambda_t=\lambda,X_{t-1},\ldots,X_0,\lambda_{t-1},\ldots,\lambda_0\right)
=P_{\lambda}\left(x,A\right).$$
The choice of $\lambda_t$ is allowed to depend upon $X_{t-1},\ldots,X_0,\lambda_{t-1},\ldots,\lambda_0$, although
in practice, $\left\{(X_t,\lambda_t)\right\}_{t=0}^{\infty}$ is usually designed to be a Markov chain.
\ctn{Roberts07} prove convergence and ergodicity of the adaptive chain assuming the diminishing adaptation
condition
\begin{equation}
\underset{t\rightarrow\infty}{\lim}\underset{x}{\sup}\|P_{\lambda_{t+1}}(x,\cdot)-P_{\lambda_t}(x,\cdot)\|=0~~
\mbox{in probability}.
\label{eq:dim_adapt}
\end{equation}
and the bounded convergence condition
\begin{equation}
\left\{M_{\eta}\left(X_t,\lambda_t\right)\right\}_{t=0}^{\infty}~~\mbox{is bounded in probability},
\label{eq:bounded_convergence}
\end{equation}
with $M_{\eta}\left(X_t,\lambda_t\right)=\inf\left\{t\geq 1:P^t_{\lambda}(x,\cdot)-\pi(\cdot)\|\leq\eta\right\}$
being essentially the convergence time of $P_{\lambda}$ when started with the initial value $x$. As argued
in \ctn{Rosenthal11}, (\ref{eq:bounded_convergence}) is satisfied quite generally, except perhaps some pathological
examples, and thus the diminishing adaptation condition (\ref{eq:dim_adapt}) is more important and requires careful
designing of the adaptive scheme. 

A valid adaptive method that is very popular is to set $\lambda_t$ to be
the empirical average of $\lambda_0,\lambda_1,\ldots,\lambda_{t-1}$. Such a scheme has been used, for example, by
\ctn{Haario01} for adaptive optimal scaling with normal random walk, where at the $(t+1)$-th iteration 
the proposal $y$ is generated from $N\left(x_t,\frac{\ell^2_{opt}}{d}\Sigma_{t+1}\right)$, where
$\ell_{opt}=2.38$ is the optimal scale borrowed from the RWM based optimal scaling theory and 
$\Sigma_{t+1}$ is an estimate of the target covariance matrix, set as the empirical covariance matrix of $X_0,\ldots,X_t$. 
To prevent singularity of $\Sigma_{t+1}$, \ctn{Haario01} added a small positive quantity to its diagonal, for all
the iterations. Alternative ideas, such as a mixture distribution, may also be considered (see \ctn{Roberts09}).
Such optimal scaling based adaptive rules are expected to have an ultimate acceptance rate close to $0.234$.
There exist various modifications of the basic approach of \ctn{Haario01}; see, for example, \ctn{Haario05},
\ctn{Andrieu08}, \ctn{Craiu09}, \ctn{Roberts09}. 

\ctn{Dey13c} has constructed various adaptive versions of TMCMC, focussing particularly on additive TMCMC,
and aiming for the ultimate optimal acceptance rate $0.439$. Comparisons of adaptive additive TMCMC
with various RWM based adaptive algorithms in simulation studies led to the very interesting observation that
even for dimension as small as $d=10$, some of the RWM based adaptive algorithms failed to converge to the desired 
acceptance rate
$0.234$ even after $10^5$ iterations, while adaptive TMCMC reached its optimal acceptance rate $0.439$ 
much faster, for all the adaptive versions considered. For dimensions as high as $d=100$, the drop in efficiencies
of the RWM based algorithms in comparison to TMCMC became all the more pronounced. 
Among all the existing adaptive methods, the method of \ctn{Atchade05} based on stochastic approximation (\ctn{Robbins51})
performed the best, for both adaptive MH and adaptive TMCMC.

\subsection{Optimal scaling in Metropolis Coupled MCMC (MC$^3$)}
\label{subsec:mc3}
When the target distribution is multimodal, then the usual MCMC methods generally fail to adequately explore all
the modal regions. To combat this problem, Geyer proposed the following idea. 
Instead of generating a single MCMC
from the multimodal target density $\pi$, it is worth generating parallel chains with {\it tempered} 
target density $\pi^{\beta_j}$;
$j=0,1,\ldots,m$, where $0\leq\beta_n<\beta_{n-1}<\cdots<\beta_1<\beta_0=1$ are suitable {\it inverse temperatures} 
such that $\pi^{\beta_j}$ becomes
progressively smoother and tends to unimodality as $j$ increases. 
MC$^3$ proceeds by running one chain at each of the $m+1$ values of $\beta$. 
%If $\mathcal X$ was the original state space, then 
The current scenario with $m+1$ target densities can be thought
%of as has state space $\chi^{n+1}$, with unnormalised stationary density 
of as the product target density $\prod_{j = 0}^m \pi^{\beta_j}(\mathbf{x}_j)$, where $\mathbf{x}_j$ 
denotes the chain at a fixed inverse temperature 
$\beta_j$ with stationary density $\pi^{\beta_j}$. The MC$^3$ idea then suggest generating parallel MCMC from
the densities $\pi^{\beta_j}$ and occasionally swapping the values of the parallel chains. 
The swapping of the states help exchange information between different modal regions of the original target
and hence helps explore the target more efficiently compared to the usual MCMC algorithms.
The algorithm is given as follows.

\begin{algo}\label{algo:mc3} \topline The MC$^3$ algorithm \botline \normalfont \ttfamily
\begin{itemize}
\item Update in parallel the Markov chains for each of the tempered densities. 
using any convergent MCMC algorithm up to a certain number of iterations say $t_0$.
\item Then for each iteration $t$ ($t>t_0$),
\begin{enumerate}
\item Attempt \textit{within temperature move} by updating each $\mathbf{x}_j$ using the usual 
RWMH MCMC algorithm with stationary density $\pi^{\beta_j}$.
\item Attempt a \textit{temperature swap} by randomly choosing two different inverse temperatures, 
say $\beta_j$ and $\beta_k$, and then proposing to swap their respective state values 
%i.e. to interchange between $\mathbf{x}_j$ and $\mathbf{x}_k$. This proposal swap 
with probability 
$$\alpha = \min\left\{1,\frac{\pi^{\beta_j}(\mathbf{x}_k)\pi^{\beta_k}(\mathbf{x}_j)}
{\pi^{\beta_j}(\mathbf{x}_j)\pi^{\beta_k}(\mathbf{x}_k)}\right\}.$$ 
%which is the usual symmetric Metropolis algorithm probability. 
If the swap is rejected, the values of the states remain unchanged.
\end{enumerate}
\end{itemize}
\botline \rmfamily
\end{algo}
The spacing of the inverse temperatures $\beta_j$ has important consequences of the mixing of the algorithm. For instance,
if two close values of $\beta$ are swapped, then not much information is exchanged and so mixing is not expected to improve,
while the proposal to swap too far away values of $\beta$ would usually lead to rejection of the swap proposal.
Thus, optimal scaling of the spacings between the inverse temperatures is necessary. \ctn{Atchade10} propose the spacings
to be of length $\eta=\frac{\ell}{d}$, for a $d$-dimensional target density, where $\ell$ must be chosen optimally
chosen in some sense. Under the assumption that the original target density is a product of $iid$ densities, 
\ctn{Atchade10} maximize the stationary ESJD with respect
to $\ell$ to obtain the optimal spacing. For the optimal spacing, the corresponding swap acceptance rate turns out to
be $0.234$.

\ctn{Dey17} proposed to randomize the spacings such that $\eta=\frac{\ell}{\sqrt{d}}\epsilon$,
where $\epsilon\sim q(\cdot)I_{\{\epsilon>0\}}$, where $q$ is any arbitrary density. He referred to the
corresponding randomized algorithm as randomized Metropolis Coupled Markov Chain Monte Carlo (RMC$^3$).
When $q$ is the left truncated $N(0,1)$ density, \ctn{Dey17} proved that the optimal swap acceptance rate
of RMC$^3$, obtained via maximization of stationary ESJD, is $0.439$. In keeping with the much improved
swap acceptance rate, we observed much improved mixing of RMC$^3$ in comparison with MC$^3$ in simulation studies.
%, particularly for high-dimensional target distributions. 
We also propose to simulate the parallel Markov chains using TMCMC, rather than the traditional MCMC methods, for
much greater efficiency. The resulting methodology can be termed as 
randomized transformation-based Metropolis Coupled Markov Chain Monte Carlo (RTMC$^3$).

Recently \ctn{Khamaru16} created an appropriate randomized variable dimensional
swap based methodology for variable dimensional target distributions, where given some (perhaps, all) dimensions, the
target is multimodal. The parallel, variable-dimensional chains are simulated using Transdimensional Transformation
based Markov Chain Monte Carlo (TTMCMC) (\ctn{Das17}). The authors refer to this novel methodology as
randomized transdimensional transformation-based Metropolis Coupled Markov Chain Monte Carlo (RTTMC$^3$). Even for RTTMC$^3$, 
the optimal swap acceptance rate turned out to be $0.439$!

The rest of our paper is structured as follows.
In Section \ref{sec:tmcmc_thick_tailed} we discuss our diffusion based approach to optimal scaling 
of additive TMCMC with non-Gaussian, thick-tailed proposals, assuming that the regularity conditions of Theorem 4.1 of \ctn{Dey13}
are satisfied. Even though the proof of our result does not go through
with the Cauchy proposal (since the moments do not exist), our simulation studies indicate that at least the recipe for
obtaining optimal scaling and optimal acceptance rate remains valid even for the Cauchy proposal, which is what we
conjecture. We follow up our theoretical investigations with simulation studies and compare additive TMCMC and
RWM for Gaussian and Cauchy proposals, considering the target distributions to be a $t$ density with $5$ degrees of freedom, 
a density with exponential tails. As expected, TMCMC emerges the winner in all the cases; our simulation studies 
also demonstrate that the Gaussian proposal
is perhaps more efficient than the Cauchy proposal.
We consider another more realistic simulation study involving simulation from the posterior distribution associated with 
a mixture of Weibull distributions, and again TMCMC is seen to outperform RWM.
%, that are based on our general diffusion result (which we prove for proposals other than Cauchy),
%still remain valid, leading us to conjecture the validity of our optimal scaling and optimal acceptance rate 
%even for the Cauchy proposal.
In Section \ref{sec:diffusion_bounded_support} we consider target densities with bounded support, so that 
they are no longer continuous on $\mathbb R$. The indicator function associated with the bounded support condition
makes direct derivation of diffusion results difficult. To avoid such difficulty we consider the logistic transformation,
mapping the bounded random variables to $\mathbb R$, and obtain our diffusion result on the transformed space. We then make
use of the It\^{o} formula to obtain the diffusion result associated with the original bounded random variables, for
Gaussian/non-Gaussian proposal distributions. We show
that the notion and interpretation of diffusion speed remains intact even in the latter diffusion equation, so that
obtaining optimal scaling by maximizing the diffusion speed remains a valid approach. 
Explicit forms and values of the optimal scales and optimal acceptance rates for various proposal distributions 
are provided and discussed in Section \ref{sec:optimal}. We compare our diffusion based optimal scaling of additive TMCMC
with the ESJD based optimal scaling of RWM (\ctn{NealRoberts11}) in Section \ref{sec:comparison}, focussing
particularly on the Cauchy proposal. We show that our approach emphatically outperforms the ESJD method for the Cauchy based
RWM agorithm. 
In Section \ref{sec:slice} we compare additive TMCMC and RWM with the popular and usually effective slice sampling method in the
case of a $d$-dimensional target density with positive support, demonstrating that additive TMCMC significantly
outperforms both the competing methods for all the values of $d$ considered. 
Finally, we summarize our contributions and provide concluding remarks in Section \ref{sec:conclusion}.

\section{Diffusion based approach for additive TMCMC with non-Gaussian, thick-tailed
proposals}
\label{sec:tmcmc_thick_tailed}

The diffusion based approach for additive TMCMC, as considered by \ctn{Dey13} remains valid
in spite of non-Gaussian proposals.
To understand why this is the case, we first provide a brief overview of additive TMCMC.

\subsection{Additive TMCMC}
\label{subsec:overview_tmcmc}

%Note that our moves at each step were symmetric and the magnitude of the jump would depend on the choice of the proposal density $q$. 
%Again, $q$ must have its support as $\mathbb{R}^{+}$, and from now onwards, we shall assume that it is a $N(0,1)$ distribution truncated at $0$. 
%Note that at each step , we sample only one $\epsilon$ from this proposal distribution and updates all the co-ordinates at one go. 
%The notion of symmetrical transitions at each step can be expressed from a mathematical point of view as follows. 
As before, assume that we are simulating from a $d$ dimensional space (usually $\mathbb{R}^{d}$), 
and that we are currently at a point $x= (x_{1}, \ldots, x_{d})$.
Further, let us define $d$ random variables $b_{1}, \ldots, b_{d}$ as in (\ref{eq:b_tmcmc}).
%, such that, for $i=1,\ldots,d$, 
%\begin{equation}
%b_{i} =\left\{\begin{array}{ccc} +1 & \mbox{with probability} & p_i; \\
% -1 & \mbox{with probability} & 1-p_i.
%  \end{array}\right.
%  \end{equation}
  The additive TMCMC uses moves of the following type: 
  \begin{equation*}
  (x_{1}, \ldots, x_{d}) \rightarrow (x_{1}+ b_{1}\epsilon, \ldots, x_{d}+b_{d}\epsilon), 
  \end{equation*}
  where $\epsilon>0$ has any arbitrary distribution with support $\mathbb R_+$, the
  positive part of the real line. In this work,
  we shall assume that $p_i=1/2$ for $i=1,\ldots,d$ and that
  $\epsilon=\frac{\ell}{\sqrt{d}}\epsilon^*$, where $\epsilon^*\sim q(\cdot)I_{\{\epsilon^*>0\}}$,
  where $q(\cdot)$ is an arbitrary density with support $\mathbb R_+$. Here 
  for any set $A$, $I_{A}$ denotes the indicator function of $A$.
  %and $q(\epsilon)I_{\{\epsilon>0\}}\equiv N(0,\frac{\ell^2}{d})I_{\{\epsilon>0\}}$.

  Thus, a single $\epsilon$ is simulated from a distribution supported on $\mathbb R_+$, which is then either added to, 
  or subtracted
  from each of the $d$ co-ordinates of $x$ with probability $1/2$. Assuming that the target distribution is proportional to 
  $\pi$, the new move $x^*=(x_1+b_1\epsilon,\ldots,x_d+b_d\epsilon)$ is accepted
  with probability
  \begin{equation}
  \alpha=\min\left\{1,\frac{\pi(x^*)}{\pi(x)}\right\}.
  \label{eq:acc_tmcmc}
  \end{equation}

  %The RWM algorithm, on the other hand, proceeds by simulating $\epsilon_1,\ldots,\epsilon_d$ independently from
  %$N(0,\frac{\ell^2}{d})$, and then adding $\epsilon_i$ to the co-ordinate $x_i$, for each $i$. The new move is accepted
  %with probability having the same form as (\ref{eq:acc_tmcmc}).
  %Computationally, additive TMCMC seems to be an easier exercise than RWM because simulating 
  %a single $\epsilon$ from $N(0,\frac{\ell^2}{d})I_{\{\epsilon>0\}}$ is simple (it is the same as simulating only 
  %once from $N(0,\frac{\ell^2}{d})$ and
  %taking its absolute value), while simulating $b_i$ is equivalent to simple tosses of a fair coin which is a much easier exercise
  %compared to drawing a high-dimensional set of normal random variables required by RWM.

The main difference of additive TMCMC with the RWM algorithm is that, instead of simulating and utilizing a single $\epsilon$, 
the latter proceeds by simulating 
$\epsilon_1,\ldots,\epsilon_d$ independently from
some density supported on the entire real line, 
and then adding $\epsilon_i$ to 
the co-ordinate $x_i$, to form $x^*_i$, for each $i$. 
The new move is accepted with probability having the same form as (\ref{eq:acc_tmcmc}). 
The default, optimally scaled RWM proposal corresponds to
$\epsilon_i=\frac{\ell}{\sqrt{d}}\epsilon^*_i$, where $\epsilon^*_i\stackrel{iid}{\sim}N(0,1)$, 
for appropriate (optimal) choice of $\ell$. 

As discussed in \ctn{Dutta13}, in $d$ dimensions the number of $\epsilon_i$ allowed by TMCMC
ranges from 1 to $d$, so that RWM is a special case of additive TMCMC. In what follows, however,
we confine ourselves to a single $\epsilon$ for additive TMCMC.

\subsubsection{Computational gain of TMCMC over RWM}
\label{subsubsec:computational_gain}
Although TMCMC requires simulation of $d+1$ random variables in every iteration as opposed to simulation of 
$d$ random variates required by RWM, the computational complexity of the former algorithm is much less because
simulation of Bernoulli random variables is computationally a much simpler exercise compared to simulation
of normal deviates. The issue on computational gain of TMCMC is illustrated in \ctn{Dey13}; here we further remark that
RWM took about 43 minutes for completion of $10^6$ iterations for a $100$-dimensional target distribution composed
of products of standard normal densities truncated on $(-1,1)$ (see Section \ref{sec:comparison}), 
while additive TMCMC took just about 28 minutes for the
same number of iterations and the same target distribution, the codes been written in R and implemented on a single node desktop machine. 

%Our proof differs from the previous
%approaches associated with RWM particularly because, as already shown in Section \ref{sec:overview_tmcmc}, 
%in additive TMCMC the terms $b_i\epsilon$ are not
%jontly normally distributed unlike the RWM-based approaches (recall from Section \ref{subsec:non_robustness} 
%that RWM considered transitions of the form $(x_1,\ldots,x_d)\rightarrow 
%(x_1+\frac{\ell_1}{\sqrt{d}}\epsilon_1,\ldots,x_1+\frac{\ell_d}{\sqrt{d}}\epsilon_d)$,
%where, for $i=1,\ldots,d$, $\epsilon_i\stackrel{iid}{\sim} N(0,1)$).
%Thus, unlike the RWM-based approaches, in our case obtaining appropriate normal approximation to
%relevant quantities are not assured. To handle the difficulty, we

\subsection{Diffusion approach to additive TMCMC avoids technical difficulties associated with
non-Gaussian proposals using Lyapunov's central limit theorem conditional on $\epsilon$ and $b_1$}
\label{subsec:conditional_Taylor}

In order to prove diffusion based optimal scaling results for additive TMCMC, \ctn{Dey13} 
had to apply Lyapunov's central limit theorem on sums associated with the discrete random variables 
$\{b_i;i=2,\ldots,d\}$, {\it conditional on} $\epsilon$ (and $b_1$), and hence did not have to rely
on any Gaussian assumption. Indeed, as shown in \ctn{Dey13}, even if 
$q(\cdot)I_{\{\epsilon^*>0\}}\equiv N(0,\sigma^2)I_{\{\epsilon^*>0\}}$, so that for each $i$, 
$b_i\epsilon^*\sim N(0,\sigma^2)$, we still do not have joint normality of $(b_1\epsilon^*,\ldots,b_d\epsilon^*)$. 
%Note that even though $b_i\epsilon$ are pairwise uncorrelated
%($E(b_i\epsilon\times b_j\epsilon)=0$ for $i\neq j$), they are not independent since all of them involve the same $\epsilon$. 
In fact, $b_i\epsilon^*+b_j\epsilon^*=0$ with probability $1/2$ for $i\neq j$, showing that the 
linear combinations of $b_i\epsilon^*$ need not be normal. That is, the joint distribution of 
$(b_1\epsilon^*,\ldots,b_d\epsilon^*)$ is not normal, even though the marginal
distributions are normal and the components are pairwise uncorrelated ($E(b_i\epsilon^*\times b_j\epsilon^*)=0$ for $i\neq j$). 
This also shows that $b_i\epsilon^*$ are not independent, because independence would imply joint normality of the components. 
Note that $b_i\epsilon^*$ are dependent on the same $\epsilon^*$, hence they are not independent anyway.

\subsection{Formal diffusion result for non-Gaussian proposals for $iid$ product target densities}
\label{subsec:formal_result}

Let us consider target densities of the form
\begin{equation}
\pi_X(x) = \prod_{i=1}^{d}{f_X(x_{i})};\quad -\infty<x_i<\infty,\quad\forall~i=1,\ldots,d, 
\label{eq:iid_x}
\end{equation}

Let $X^d_t=(X_{t,1},\ldots,X_{t,d})$.
As in \ctn{Dey13} (see also the references therein), we define 
${U_{t}}^{d} = {X_{[dt],1}}$ ($[\cdot]$ denotes the integer part), 
the sped up first component of the actual additive TMCMC-induced Markov chain. 
Thus this process proposes a jump every $\frac{1}{d}$ time units. As $d \rightarrow \infty$, 
that is, as the dimension grows to $\infty$, the process essentially becomes a continuous time diffusion process.

Following \ctn{Dey13} let us assume that
\begin{align}
&E_{f_X}\left(\frac{f_X'(X)}{f_X(X)}\right)^4 <\infty,
\label{eq:assump1}\\
&E_{f_X}\left(\frac{f_X''(X)}{f_X(X)}\right)^4  <\infty,
\label{eq:assump2}\\
&E_{f_X}\left(\frac{f_X'''(X)}{f_X(X)}\right)^4  <\infty,
\label{eq:assump3}\\
&E_{f_X}\left|\frac{f_X''''(X)}{f_X(X)}\right|  <\infty.
\label{eq:assump4}
\end{align}
Following \ctn{Roberts97a} let us denote weak convergence of
processes in the Skorohod topology by ``$\Rightarrow$"; see also \ctn{Dey13}.
Then, the following theorem, which is essentially Theorem 4.1 of \ctn{Dey13}, holds: 
\begin{theorem}
\label{theorem:theorem1}
Assume that $f_X$ is positive with at least three continuous derivatives and that the fourth derivative exists almost everywhere. 
Also assume that $(\log f_X)'$ is Lipschitz continuous, 
and that (\ref{eq:assump1}) -- (\ref{eq:assump4}) hold. 
Let $X^d_0\sim\pi_X$, that is, the $d$-dimensional additive TMCMC chain is started at stationarity, and
let the transition be given by $(x_1,\ldots,x_d)\rightarrow (x_1+b_1\epsilon,\ldots,x_d+b_d\epsilon)$, where
for $i=1,\ldots,d$, $b_i=\pm 1$ with equal probability and 
$\epsilon\equiv\frac{\ell}{\sqrt{d}}\epsilon^*$, where $\epsilon^*\sim q(\cdot)I_{\{\epsilon^*>0\}}$. 
%N(0,1)I_{\{\epsilon^*>0\}}$.
We then have
\[
\{U^d_t;~t\geq 0\}\Rightarrow \{U_t;~t\geq 0\},
\]
where $U_0\sim f_X$ and $\{U_t;~t\geq 0\}$ satisfies the Langevin stochastic differential equation (SDE)
\begin{equation}
dU_t=g(\ell)^{1/2}dB_t+\frac{1}{2}g(\ell)\left(\log f_X(U_t)\right)'dt,
\label{eq:sde1}
\end{equation}
with $B_t$ denoting standard Brownian motion at time $t$,
\begin{equation}
%g(\ell)=4\ell^2\int_0^{\infty} u^2\Phi\left (- \frac{u\ell\sqrt{\mathbb{I}}}{2}\right)\phi(u)du;
g(\ell)=4\ell^2\int_0^{\infty} u^2\Phi\left (- \frac{u\ell\sqrt{\mathbb{I_X}}}{2}\right)q(u)du;
\label{eq:diff_speed_iid}
\end{equation}
$\Phi(\cdot)$ being the standard normal cumulative distribution function (cdf), and
\begin{equation}
\mathbb I_X=E_{f_X}\left(\frac{f_X'(X)}{f_X(X)}\right)^2.
\label{eq:information_x}
\end{equation}

%\begin{align}
%\mathbb I_X&=E_{f_X}\left(\frac{f_X'(X)}{f_X(X)}\right)^2\notag\\
%&=E_{f_Y}\left[1-\frac{2e^Y}{1+e^Y}+\frac{f'_X\left(\frac{a+be^Y}{1+e^Y}\right)}{f_X\left(\frac{a+be^Y}{1+e^Y}\right)}
%\times\frac{(b-a)e^{Y}}{(1+e^Y)^2}\right]^2\notag\\
%&=\bigint_{-\infty}^{\infty}
%\left[1-\frac{2e^y}{1+e^y}+\frac{f'_X\left(\frac{a+be^y}{1+e^y}\right)}{f_X\left(\frac{a+be^y}{1+e^y}\right)}
%\times\frac{(b-a)e^{y}}{(1+e^y)^2}\right]^2
%\frac{(b-a)e^{y}}{(1+e^y)^2}f_X\left(\frac{a+be^y}{1+e^y}\right)dy.
%\label{eq:information}
%\end{align}
\end{theorem}
The main difference of this theorem with Theorem 4.1 of \ctn{Dey13} is that here we allow
$\epsilon^*$ in $\epsilon\equiv\frac{\ell}{\sqrt{d}}\epsilon^*$ to have arbitrary distribution
$q(\cdot)I_{\{\epsilon^*>0\}}$, supported on the positive part of the real line, whereas
\ctn{Dey13} considered $q(\cdot)$ to be $N(0,1)$. The proof of the theorem only requires
%$\epsilon^*\sim q(\cdot)I_{\{\epsilon^*>0\}}$ 
$b_i\epsilon^*$ to have finite moments, and with this assumption, 
exactly the same proof of \ctn{Dey13} goes through for non-Gaussian choices of $q(\cdot)$.

\subsection{Conjecture for proposals where the moments of $b_i\epsilon^*$ do not exist}
\label{subsec:conjecture}

As indicated above, the proof of Theorem \ref{theorem:theorem1}, analogous to the proof of Theorem 4.1
of \ctn{Dey13}, does not carry over for proposal distributions for which the moments of $b_i\epsilon^*$ do not exist, 
which happens when $b_i\epsilon^*$ is distributed as Cauchy, for instance. 
The reason is that some requisite Taylor's series expansions associated with $b_1\epsilon^*$ will not be 
valid as the higher order terms do not converge in probability to zero as $d\rightarrow\infty$.
However, all our simulation studies demonstrated that 
%the form of the final
%diffusion equation (\ref{eq:sde1}) to which the our additive TMCMC algorithm converges, remains valid, 
%even for the Cauchy proposal. Indeed, 
our additive TMCMC algorithms with the Cauchy proposal 
and the scale $\ell/\sqrt{d}$, have empirical acceptance rate extremely close to that associated with
the theoretical acceptance rate associated with (\ref{eq:diff_speed_iid}), even for $d$ as small as 10, and
results of simulations with high dimensions $d=50$ and $d=100$
lend further support to this observation (see Sections \ref{subsec:comparison_unbounded} and \ref{sec:comparison}).
We thus {\it conjecture} that %Theorem \ref{theorem:theorem1} 
at least the method of obtaining optimal scaling and optimal acceptance rate, as discussed in
Section \ref{sec:optimal}, %associated with Theorem \ref{theorem:theorem1} 
remains valid even for the Cauchy proposal. We use the result as a 
``rule of thumb" even in situations where valid proofs are yet pending.

\subsection{Simulation experiments to compare performances of optimal TMCMC and RWM with respect to Gaussian and Cauchy
proposals}
\label{subsec:comparison_unbounded}
In this section we consider two target densities of the following forms, also considered by \ctn{NealRoberts11}:
\begin{equation}
f_X(x)=\frac{8}{3\sqrt{5}\pi}\left(1+\frac{x^2}{5}\right)^{-3};~x\in\mathbb R,
\label{eq:t_density}
\end{equation}
which is the $t$-distribution with $5$ degrees of freedom, and
\begin{equation}
f_X(x)=\left\{\begin{array}{cc} \frac{1}{4} & \mbox{if}~|x|<1;\\
\frac{1}{4}\exp\left(1-|x|\right)   & \mbox{if}~|x|\geq 1,\end{array}\right.
\label{eq:exptail_density}
\end{equation}
which is a distribution with exponential tails.

We use both Gaussian and Cauchy proposals for the competing additive TMCMC and RWM algorithms to simulate
from the above target distributions considering dimensions $d=10$, $50$ and $100$, and compare
the performances of the algorithms, with respect to both the proposal distributions, for both the target distributions. 
For the purpose of comparison we use the Kolmogorov-Smirnov (KS) 
distance between the empirical distribution function associated
with the MCMC simulations and the true, target distribution functions, both associated with the first co-ordinate
of the $d$-dimensional distributions. We also consider the autocorrelations of the underlying Markov chains.

Using equations (\ref{eq:diff_speed_iid}) and (\ref{eq:information_x}) we find that for both the target distributions
(\ref{eq:t_density}) and (\ref{eq:exptail_density}), the optimal acceptance rate of additive TMCMC is $0.439$ for
the Gaussian proposal and $0.380$ for the Cauchy proposal. 
As shown in Table \ref{table:table1}, even for target densities with bounded support, the optimal acceptance rate for
additive TMCMC with the Cauchy proposal is $0.380$; indeed, as argued in Section \ref{sec:optimal}, the optimal acceptance
rate depends only on the choice of the proposal distribution.
When the target density is (\ref{eq:t_density}), the
optimal scales for the Gaussian and Cauchy proposals are given by 
$\ell_{opt,Gaussian}=2.802$ and $\ell_{opt,Cauchy}=2.239$, respectively, and for target density (\ref{eq:exptail_density}),
these are given by  $\ell_{opt,Gaussian}=3.431$ and $\ell_{opt,Cauchy}=2.741$.
It is worth recalling that for both the target distributions and for both the proposal distributions we consider the scale 
of the form $\ell/\sqrt{d}$. 

On the other hand, although for both the target densities the Gaussian proposal based RWM has scale 
of the form $\ell/\sqrt{d}$, the ESJD-based approach of \ctn{NealRoberts11} requires the scale to be of the form
$\ell/d$ for the RWM based Cauchy proposal. It is worth noting that for target distributions with bounded supports
\ctn{NealRoberts11} consider the scale $\ell/(d\log d)$ for the RWM based Cauchy proposal, and obtained
the optimal acceptance rate $0.368$.

In the current context, using their ESJD approach, \ctn{NealRoberts11} obtained the optimal acceptance rate for either 
of the target distribution to be $0.234$, for both the proposal distributions. 
For our simulations we choose the scales appropriately in each case such that for RWM
the empirical acceptance rate obtained from the MCMC simulations is as close to $0.234$ as possible.
In all our simulations, the optimal scales of TMCMC led to empirical acceptance rates that are very close to the actual
optimal acceptance rates.

With the above set-up, we simulated $10^5$ MCMC realizations from each target distribution, with both Gaussian
and Cauchy proposals with respect to both additive TMCMC and RWM, for dimensions $d=10,50,100$. The KS distances
for each such simulation, are provided in Table \ref{table:ks_comparison}. 
As is observed from the table, in all the cases considered, TMCMC outperforms RWM significantly in terms of
the KS distance, even though in most cases the RWM based autocorrelations decrease somewhat faster than the
TMCMC based autocorrelations (figures not shown for brevity). Since the maximum diffusion speed is higher for RWM
when the Gaussian proposal is considered (see \ctn{Dey13}), and since the optimal scale for the RWM based Cauchy proposal
is chosen by maximizing ESJD, both of which are directly related to autocorrelations, it is not unexpected 
that the autocorrelations of RWM would generally decrease faster; the same phenomenon has been observed in \ctn{Dey13}.
However, neither the maximum diffusion speed nor ESJD guarantees that the KS distance would be smaller for RWM, and
as such, our results concur with those obtained in \ctn{Dey13}, that the TMCMC significantly
outperforms RWM in terms of the KS distance. Since smaller KS distance is far more desirable than smaller autocorrelations,
it is reasonable to conclude, as in our previous works related to TMCMC, that additive TMCMC is a much superior methodology
compared to RWM. The reason for the superior performance of TMCMC in terms of the KS distance can perhaps be attributed to
its much higher acceptance rate in comparison to the somewhat slow rate of decrease of the autocorrelations. 
To elaborate, while the mixing peroperties of TMCMC and RWM in terms of their respective autocorrelations
do not differ drastically, the acceptance rate of TMCMC is of course emphatically larger than that of RWM. The latter
cancels the slight advantage of RWM in terms of autocorrelations, and tilts the comparison in favor of TMCMC
in terms of the KS distance. 

In this context, let us note that for the RWM based Cauchy proposal, the scale being of the order $O\left(d^{-1}\right)$, 
even though smaller compared to the TMCMC scale of the  order $O\left(d^{-1/2}\right)$, has a slight edge over TMCMC
in terms of autocorrelaion decay. However, for target distributions with bounded supports, the RWM scale is of the
order $O\left(\left(d\log d\right)^{-1}\right)$, while that of TMCMC remains of the  order $O\left(d^{-1/2}\right)$.
The simulation experiments detailed in Section \ref{sec:comparison} demonstrate that the further incorporation 
of the $\log d$ factor
in the RWM scale washes out the autocorrelation-related advantage of RWM over TMCMC for bounded target distributions,
and in those cases, TMCMC emphatically outperforms RWM in terms of KS distance, as well as in terms of autocorrelation decay.

Finally, Table \ref{table:ks_comparison} demonstrates that the Gaussian proposal seems to have a slight 
edge over the Cauchy proposal, for both TMCMC and RWM. This is consistent with the more emphatic conclusion of 
\ctn{NealRoberts11} that the Gaussian proposal always outperforms the Cauchy proposal, at least in terms of ESJD. 
Even our atocorrelation plots revealed that for the Gaussian proposal the autocorrelations decays faster than
that of the Cauchy proposal, for both TMCMC and RWM, for both the target densities, and for $d=10,50,100$.
In this sense, our results are consistent with those of \ctn{NealRoberts11}.
\begin{table}
%\begin{sidewaystable}[h]
\centering
\caption{KS distances between MCMC-based and target distribution functions under TMCMC and RWM with 
Gaussian and Cauchy proposals.}
%\vspace{2cm}
\begin{tabular}{|c|c|c|c||c|c|c|}
\hline
\multirow{2}{*}{\backslashbox{Proposal}{Target}} 
& \multicolumn{3}{|c||}{$f_X(x)=\frac{8}{3\sqrt{5}\pi}\left(1+\frac{x^2}{5}\right)^{-3};~x\in\mathbb R$} & 
\multicolumn{3}{|c|}{$f_X(x)=\left\{\begin{array}{cc} \frac{1}{4} & \mbox{if}~|x|<1;\\\frac{1}{4}\exp\left(1-|x|\right)   & \mbox{if}~|x|\geq 1.\end{array}\right.$}\\  
\cline{2-7}
& $d=10$ & $d=50$ & $d=100$ & $d=10$ & $d=50$ & $d=100$ \\ 
\hline
TMCMC (Gaussian) & 0.006 & 0.011 & 0.029 & 0.009 & 0.011 & 0.016\\
RWM (Gaussian) & 0.013 & 0.018 & 0.043 & 0.017 & 0.021 & 0.021\\
\hline
TMCMC (Cauchy) & 0.007 & 0.017 & 0.016 & 0.009 & 0.014 & 0.016\\ 
RWM (Cauchy) & 0.013 & 0.028 & 0.026 & 0.022 & 0.026 & 0.021\\
\hline
\end{tabular}
\label{table:ks_comparison}
%\end{sidewaystable}
\end{table}

\subsection{Simulation study for comparing TMCMC and RWM in a more realistic setting}
\label{subsec:weibull}
We now consider a simulation study in the context of the following hierarchical Bayesian model
based on a mixture of two Weibull distributions, as suggested by a referee:
$$
y_1,\ldots,y_n\stackrel{iid}{\sim}\frac{1}{2}Weibull(\alpha_1,\beta_1)+\frac{1}{2}Weibull(\alpha_2,\beta_1),
$$
where $\alpha_1,\alpha_2$ are shape parameters and $\beta_1,\beta_2$ are scale parameters. 
We assume that {\it a priori}, 
for $j=,1,2$, $\alpha_j\sim Gamma(a_j,b_j)$, where $a_j$ and $b_j$ are shape and rate parameters respectively,
so that the mean and the variance of $\alpha_j$ are $a_j/b_j$ and $a_j/b^2_j$, respectively.
Specifically, we set $a_1=a_2=b_1=b_2=0.1$. We assume for simplicity that $\beta_1=\beta_2=1$. 

The goal of this study is to evaluate the performances of additive TMCMC and RWM in generating MCMC
samples from the posterior $\pi(\alpha_1,\alpha_2|y_1,\ldots,y_n)$, for various choices of $n$.
Observe that this posterior does not satisfy the conditions necessary for the optimal scaling theories.
For instance, the target posterior is only two-dimensional, and neither are the two co-ordinates $iid$ with respect
to the posterior. But here we wish to verify the importance of the optimal scaling theory in more realistic problems;
we also wish to compare the performances of additive TMCMC and RWM in this set-up, and the performances of 
non-Gaussian and Gaussian proposals with respect to both the algorithms.

Table \ref{table:ks_comparison} demonstrates that the Gaussian proposal has an edge over the Cauchy proposal. 
Thus, in order to outperform the Gaussian proposal it is of importance to consider non-Gaussian proposals 
that are somewhat close to the Gaussian proposal. The $t$ distribution with a reasonable degree of freedom 
may thus be appropriate. Table \ref{table:table1} shows that the $t$ distribution with $5$ degrees of freedom provides an
optimal acceptance rate that is quite close to the Gaussian proposal. Note that although the table considers target
distributions with bounded supports, it has been argued in Section \ref{sec:optimal} that the optimal acceptance 
rate is independent of the target distribution or its support, and depends only on the proposal distribution. Hence, 
it is appropriate in our current situation to consider the $t$ distribution with $5$ degrees of freedom
as a suitable non-Gaussian proposal.

To set the scales of $\alpha_1$ and $\alpha_2$, we first note that, since both have the same priors
and since the likelihood gives equal weight to both, their
posteriors are likely to be similar. Hence, we use the same scaling form $\ell/\sqrt{d}$ for both $\alpha_1$ and $\alpha_2$,
with respect to both additive TMCMC and RWM. In particular, with the Gaussian proposal based additive TMCMC, we tune 
$\ell$ so that the empirical acceptance rate is close to $0.439$ and for 
the $t$ distribution with $5$ degrees of freedom, we tune $\ell$ so that additive TMCMC has an 
empirical acceptance rate is close to $0.431$. For RWM, we tune $\ell$ such that the empirical acceptance rate for
both Gaussian and $t$ proposals is close to $0.234$.

We simulate $10$ data sets from our hierarchical Bayesian of sizes $n=10$, $20$, $30$, $40$, $50$, $60$, $70$, $80$, 
$90$, $100$, each consisting of $iid$ observations. For each vlue of $n$, we then draw from the posterior distribution 
$\pi(\alpha_1,\alpha_2|y_1,\ldots,y_n)$ using Gaussian and $t$ based additive TMCMC and RWM, with the aforementioned scalings.
We discard the first $15000$ iterations as burn-in and store the next $10^5$ iterations for evaluation of the methods.
Since the true marginal distribution functions of $\alpha_1$ and $\alpha_2$ are not analytically tractable for computation 
of the KS distances, we divide the 
$10^5$ iterations after the burn-in period into two parts; one part consists of the first $50000$ realizations
(after the burn-in) and the other part contains the next $50000$ iterations. We then consider the empirical KS distance
between these two parts; smaller values would indicate better convergence. Ideally, one should consider the joint 
empirical distribution function associated with the samples drawn from the joint posterior of $(\alpha_1,\alpha_2)$,
but certainly the marginal empirical distribution functions are much easier to deal with, which is why we do not
consider the joint empirical distribution functions.

Panel (a) of Figure \ref{fig:tmcmc_rwm_t_alpha} shows the KS distances for $\alpha_1$ associated with TMCMC and RWM, for
all the $10$ data sets of sizes $n=10$, $20$, $30$, $40$, $50$, $60$, $70$, $80$, $90$ and $100$, when the proposal
distribution is $t$ with $5$ degrees of freedom.
Similarly, panel (b) of Figure \ref{fig:tmcmc_rwm_t_alpha} shows the KS distances for $\alpha_2$ associated with TMCMC 
and RWM for the $t$ based proposal. Although for $\alpha_1$ TMCMC outperforms RWM only 50\% times in terms
of KS distances, in the case of $\alpha_2$, TMCMC beats RWM 80\% times. 
With the Gaussian based proposals, as Figure \ref{fig:tmcmc_rwm_gaussian_alpha} shows, TMCMC beats RWM
in 50\% cases with respect to $\alpha_1$ but outperforms RWM in 60\% cases with respect to $\alpha_2$.
Thus, overall, TMCMC is clearly seen to have an edge over RWM even where no optimal scaling theory holds.

Figures \ref{fig:tmcmc_t_gaussian_alpha} and \ref{fig:rwm_t_gaussian_alpha} compare the performances
of the $t$ and Gaussian proposals for TMCMC and RWM respectively. Figure \ref{fig:tmcmc_t_gaussian_alpha}
shows that for both $\alpha_1$ and $\alpha_2$, TMCMC with the $t$ proposal outperforms that with the Gaussian proposal
60\% times, demonstrating that for TMCMC, the $t$ proposal with $5$ degrees of freedom may be more appropriate
than Gaussian. On the other hand, Figure \ref{fig:rwm_t_gaussian_alpha} shows that RWM based on the $t$ proposal
beats that based on the Gaussian proposal 50\% times, for both $\alpha_1$ and $\alpha_2$, suggesting that both
the proposals may be equally preferred for RWM when the optimal scaling theory does not hold.

\begin{figure}%[htp]
\subfigure [TMCMC vs RWM: KS plots for $\alpha_1$.]{ \label{fig:tmcmc_rwm_t_alpha1}
\includegraphics[width=7cm]{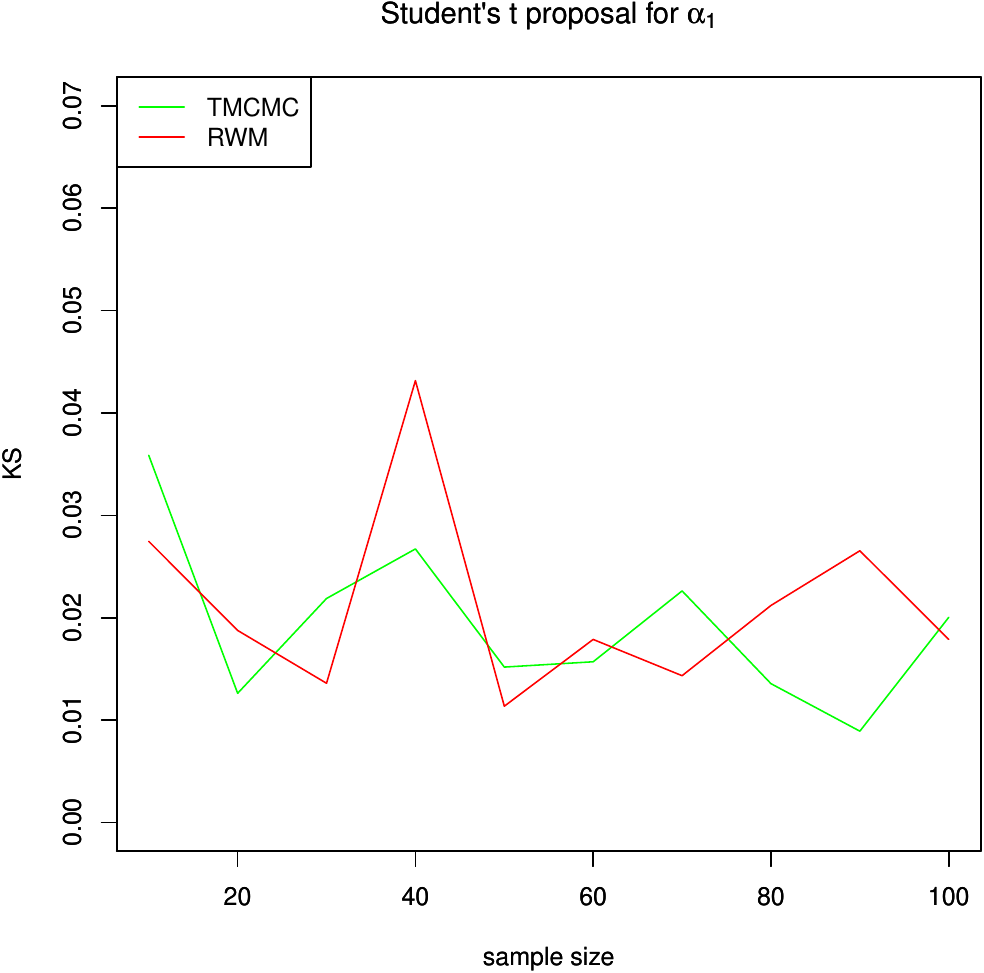}}
\hspace{2mm}
\subfigure [TMCMC vs RWM: KS plots for $\alpha_2$.]{ \label{fig:tmcmc_rwm_t_alpha2}
\includegraphics[width=7cm]{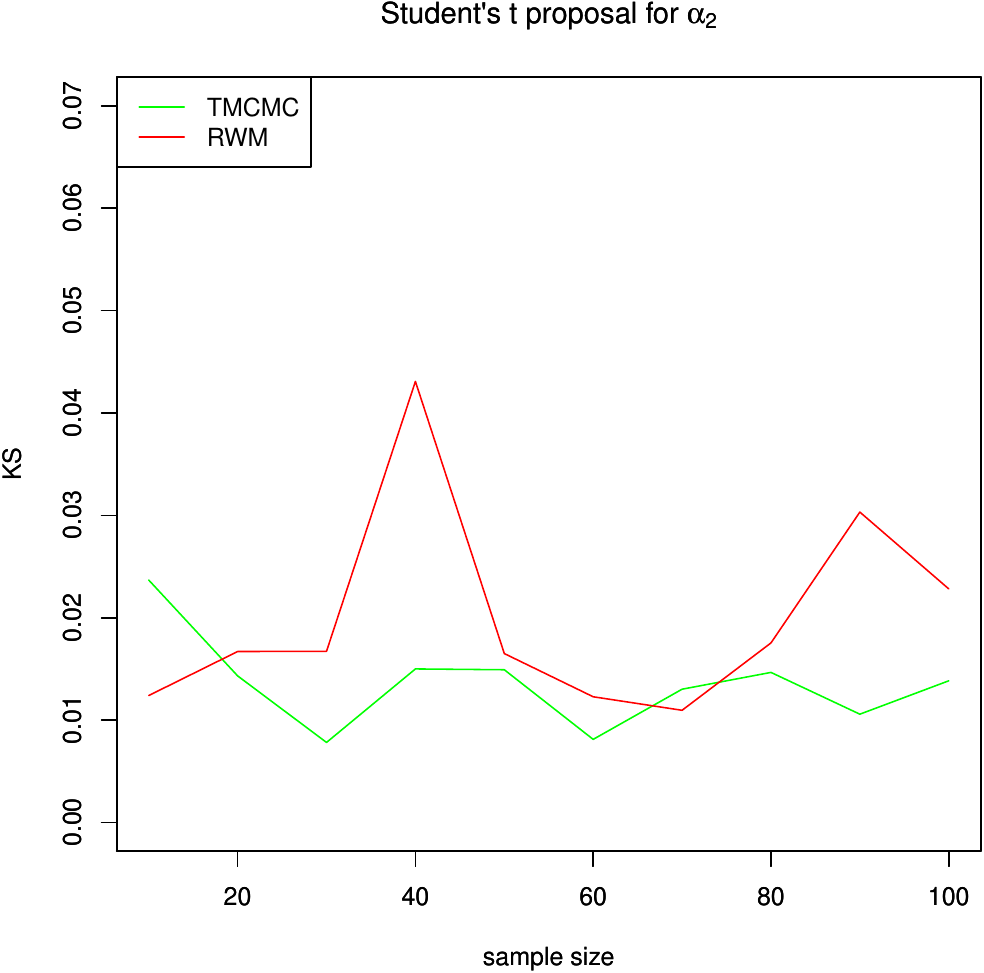}}
\caption{Plots of the KS distances of $\alpha_1$ and $\alpha_2$ 
associated with TMCMC and RWM for $10$ data sets when the proposal distribution is $t$ with $5$ degrees of freedom.
}
\label{fig:tmcmc_rwm_t_alpha}
\end{figure}

\begin{figure}%[htp]
\subfigure [TMCMC vs RWM: KS plots for $\alpha_1$.]{ \label{fig:tmcmc_rwm_gaussian_alpha1}
\includegraphics[width=7cm]{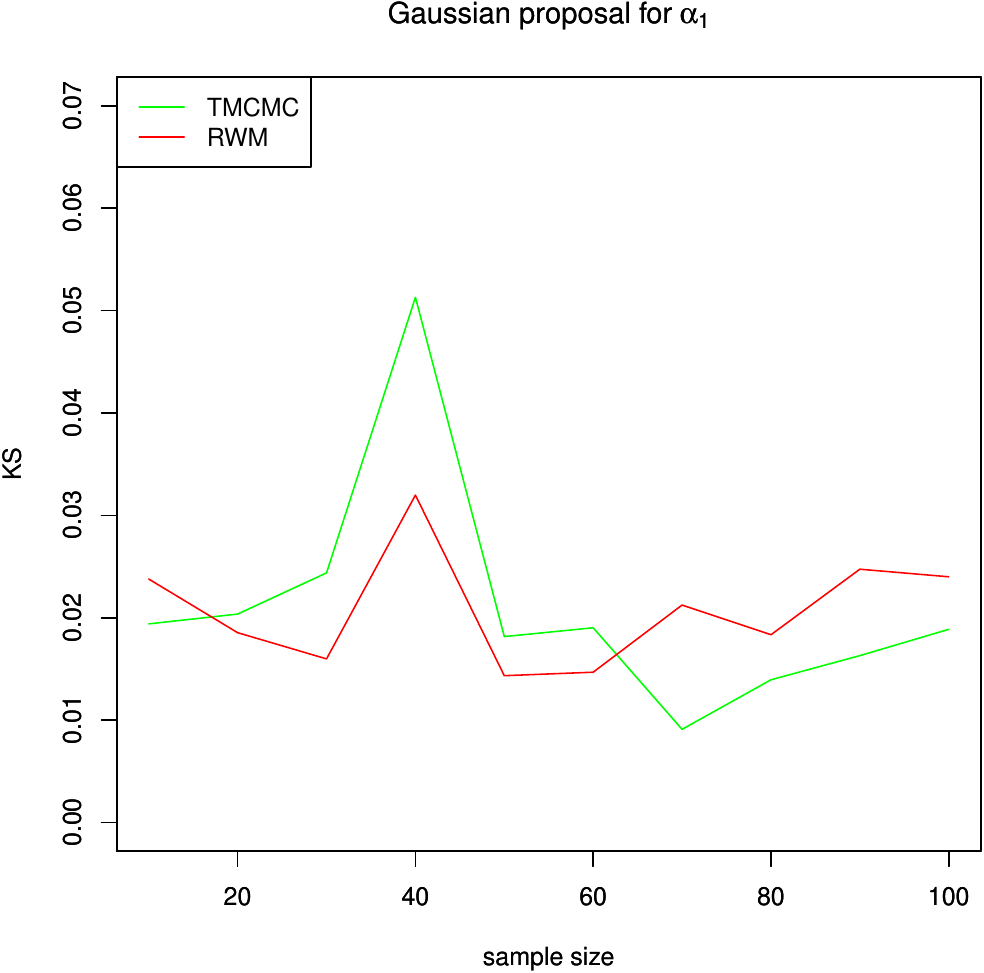}}
\hspace{2mm}
\subfigure [TMCMC vs RWM: KS plots for $\alpha_2$.]{ \label{fig:tmcmc_rwm_gaussian_alpha2}
\includegraphics[width=7cm]{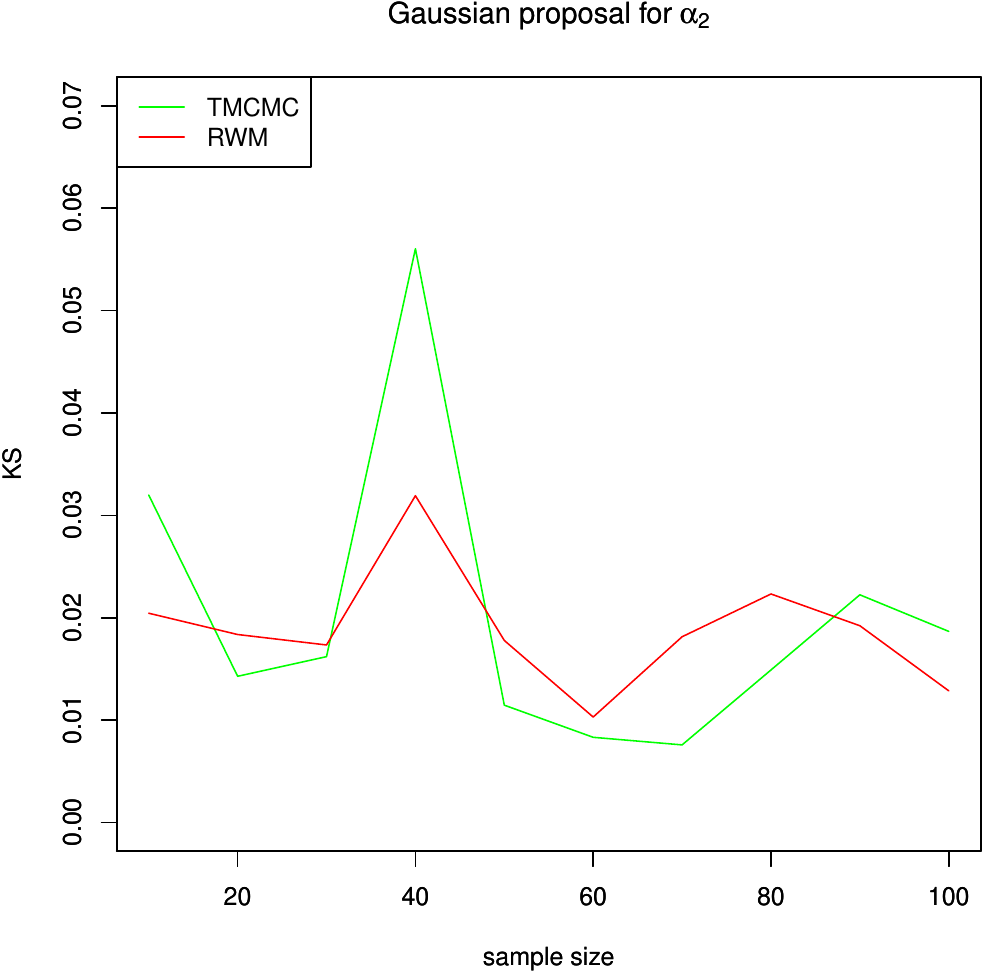}}
\caption{Plots of the KS distances of $\alpha_1$ and $\alpha_2$ 
associated with TMCMC and RWM for $10$ data sets when the proposal distribution is Gaussian.
}
\label{fig:tmcmc_rwm_gaussian_alpha}
\end{figure}

\begin{figure}%[htp]
\subfigure [TMCMC for $\alpha_1$: $t$ vs Gaussian.]{ \label{fig:tmcmc_t_gaussian_alpha1}
\includegraphics[width=7cm]{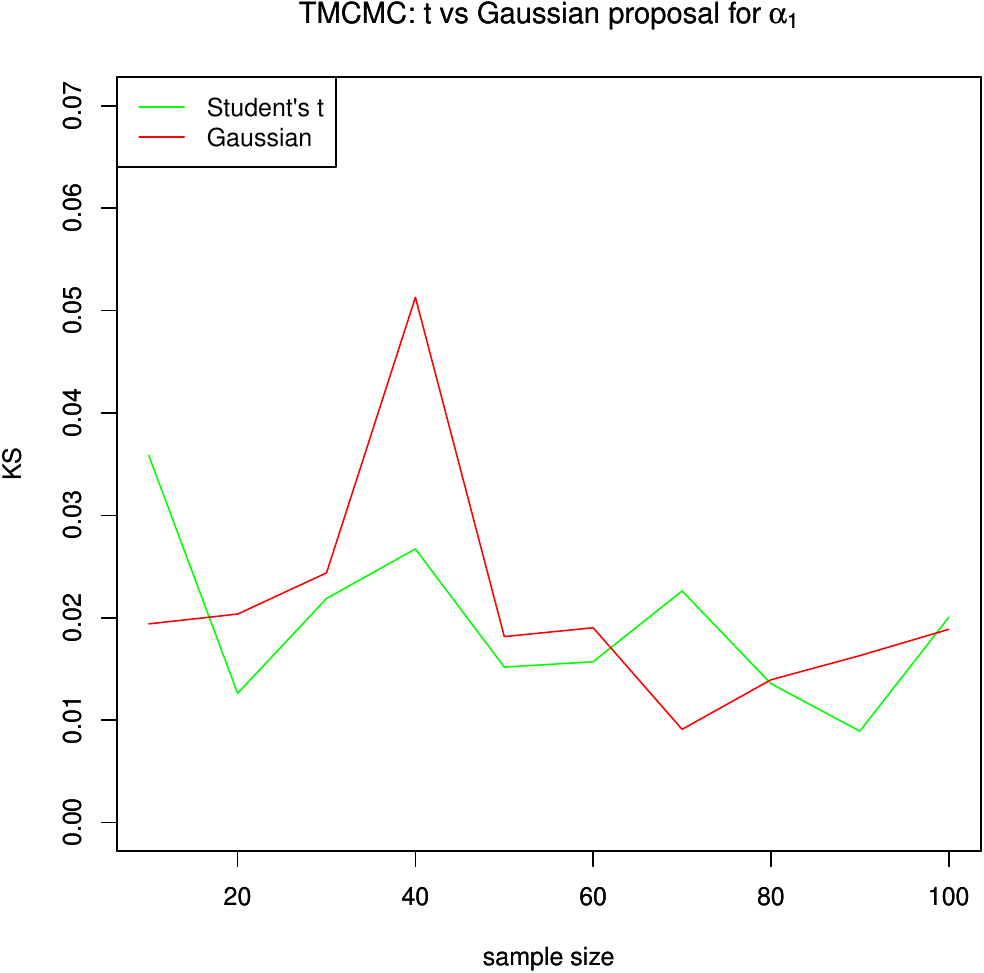}}
\hspace{2mm}
\subfigure [TMCMC for $\alpha_2$: $t$ vs Gaussian.]{ \label{fig:tmcmc_t_gaussian_alpha2}
\includegraphics[width=7cm]{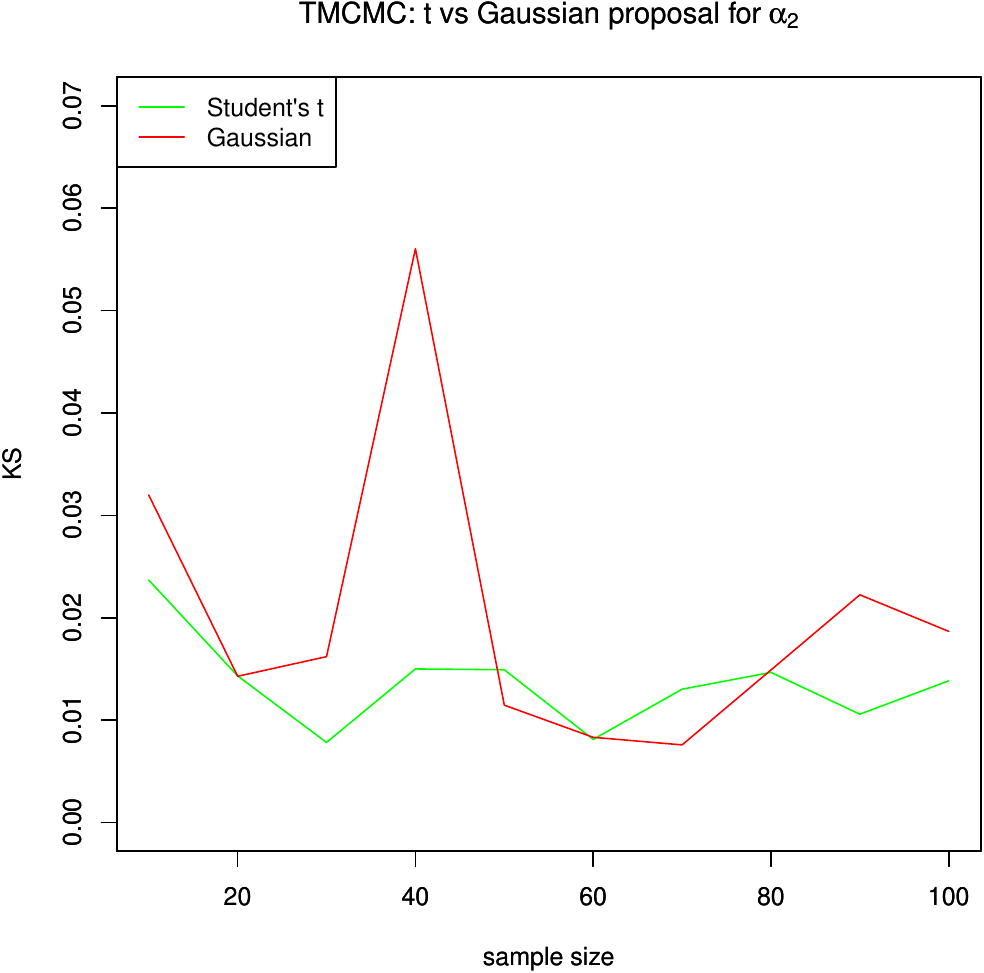}}
\caption{Plots of the KS distances comparing $t$ and Gaussian proposals 
associated with TMCMC.
}
\label{fig:tmcmc_t_gaussian_alpha}
\end{figure}

\begin{figure}%[htp]
\subfigure [RWM for $\alpha_1$: $t$ vs Gaussian.]{ \label{fig:rwm_t_gaussian_alpha1}
\includegraphics[width=7cm]{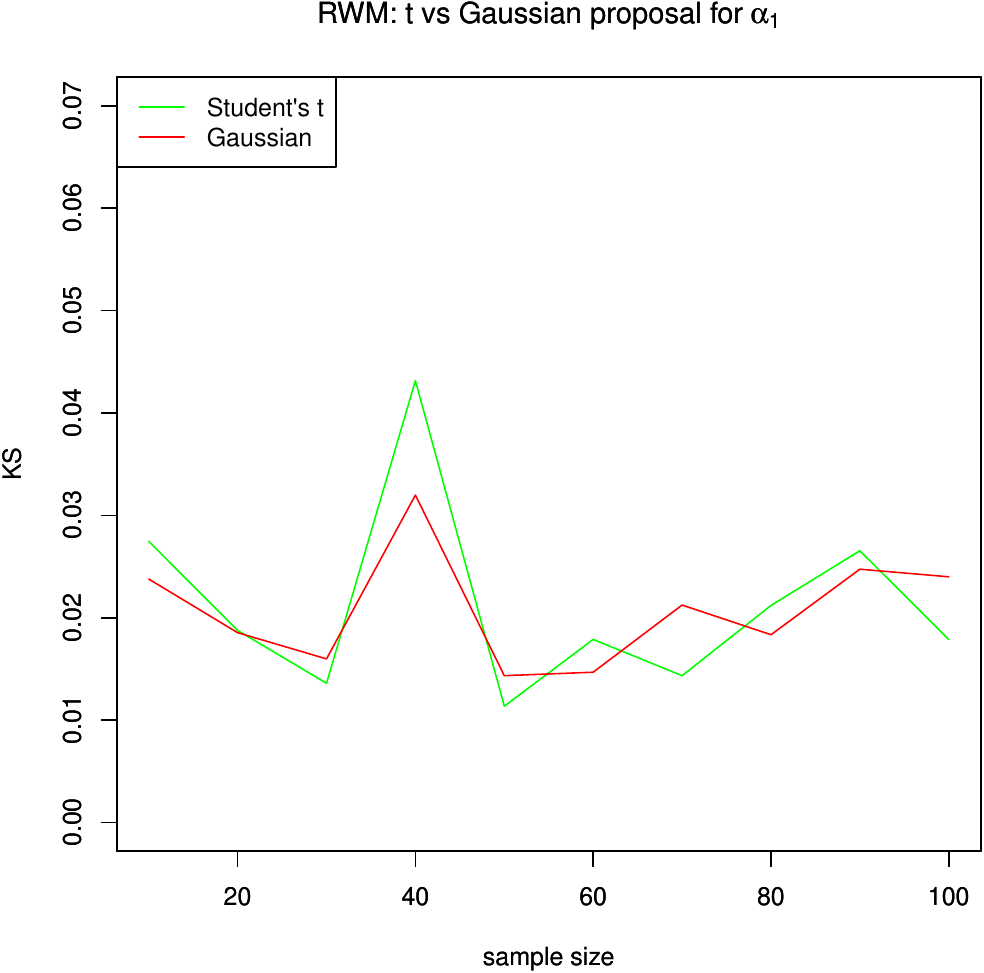}}
\hspace{2mm}
\subfigure [RWM for $\alpha_2$: $t$ vs Gaussian.]{ \label{fig:rwm_t_gaussian_alpha2}
\includegraphics[width=7cm]{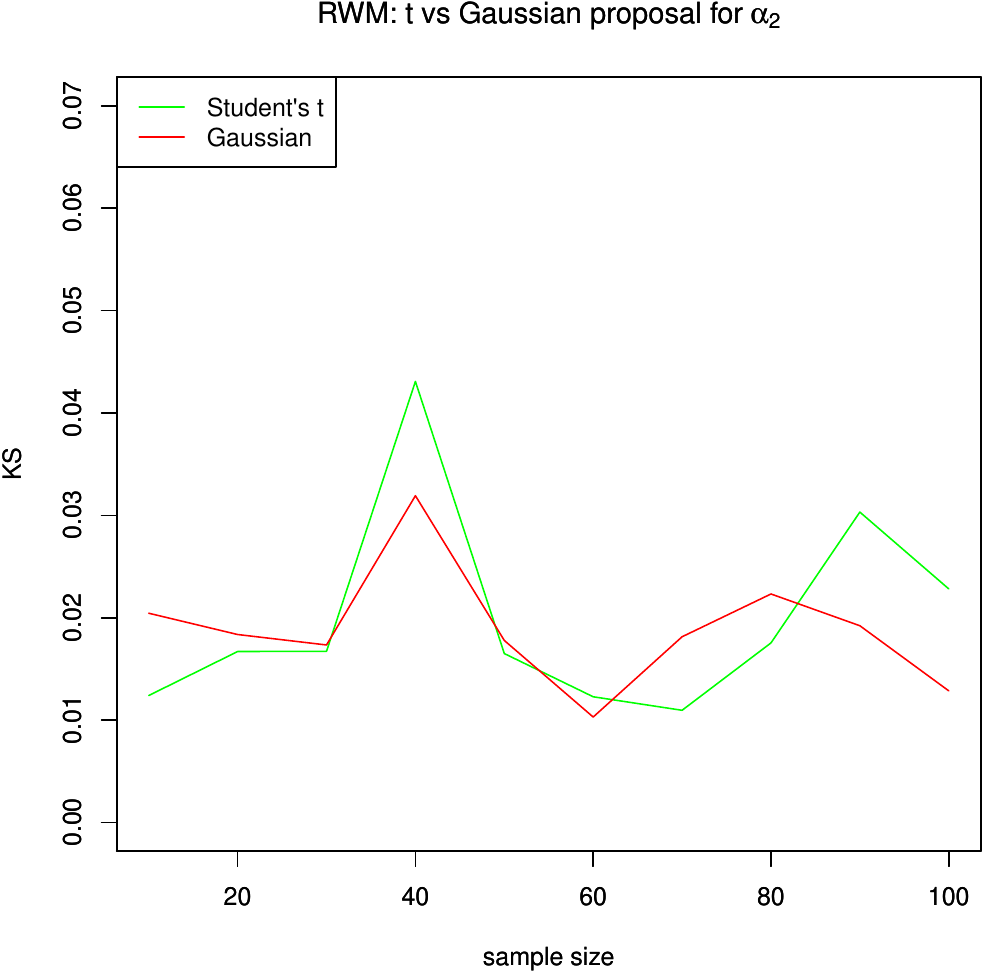}}
\caption{Plots of the KS distances comparing $t$ and Gaussian proposals 
associated with RWM.
}
\label{fig:rwm_t_gaussian_alpha}
\end{figure}

\section{Diffusion based optimal scaling for target densities with bounded supports}
\label{sec:diffusion_bounded_support}

Although the diffusion based approach of \ctn{Dey13} remains valid for additive TMCMC for any
proposal distribution such that $b_i\epsilon^*$ has finite moments, the approach needs to be slightly 
modified to accommodate target densities
with bounded supports, so that they are discontinuous in $\mathbb R^d$, say. 
Otherwise the mathematics becomes unwieldy due to the presence of the indicator functions
indicating the bounded support of the target density.
Moreover, for target densities
uniform on some bounded region, Fisher's information, which is an important ingredient in diffusion
based optimal scaling theory, is not well-defined.

In particular, let us consider target densities of the form
\begin{equation}
\pi_X(x) = \prod_{i=1}^{d}{f_X(x_{i})};\quad a<x_i<b,\quad\forall~i=1,\ldots,d, 
\label{eq:iid}
\end{equation}
for fixed real values $a<b$.

To handle such situations we provide a bijective (one-to-one and onto) transformation to each $x_i$ so that the transformed
random variables take values on the entire real line. In this paper, we will consider
the well-known logit transformation, given by
\begin{equation}
y_i=\log\left(\frac{x_i-a}{b-x_i}\right);\quad\forall~i=1,\ldots,d.
\label{eq:logit}
\end{equation}
Clearly, for each $i$, $y_i$ takes values on $\mathbb R$, and the resulting joint distribution
of $y=(y_1,\ldots,y_d)$ is given by
\begin{equation}
\pi_Y(y) = \prod_{i=1}^{d}{f_Y(y_{i})};\quad -\infty<y_i<\infty,\quad\forall~i=1,\ldots,d, 
\label{eq:iid2}
\end{equation}
where
\begin{equation}
f_Y(y_i)=(b-a)\times\frac{e^{y_i}}{\left(1+e^{y_i}\right)^2}\times f_X\left(\frac{a+be^{y_i}}{1+e^{y_i}}\right).
\label{eq:f_Y}
\end{equation}

%The transformed density $f_Y$ satisfies the necessary regularity conditions (\ref{eq:assump1}), (\ref{eq:assump2})
%and (\ref{eq:assump3}), with $Y$ replacing $X$ in the conditions. 
If $f_X$ satisfies the regularity conditions on $(a,b)$, then the transformed density $f_Y$ satisfies 
the corresponding regularity conditions on the real line $\mathbb R$.
Formally, we have the following lemma:

\begin{lemma}
\label{lemma:lemma1}
Regularity conditions on $f_X$ on $(a,b)$ carry over to regularity conditions on $f_Y$ on $\mathbb R$ in the
following ways:
\begin{itemize}
\item[(a)] Assume that $f_X$ is positive with at least three continuous derivatives and that the
fourth derivative exists almost everywhere on $(a,b)$. Then the same holds for the transformed density 
$f_Y$ on $\mathbb R$.

\item[(b)] If $f_X$ satisfies the moment conditions (\ref{eq:assump1}) -- (\ref{eq:assump4}),
then the transformed density $f_Y$ satisfies the same moment conditions with $Y$ replacing $X$.

\item[(c)] If $\left(\log f_X\right)'$ is 
Lipschitz continuous on $(a,b)$, then $\left(\log f_Y\right)'$ is Lipschitz continuous on $\mathbb R$.
\end{itemize}
\end{lemma}

\begin{proof}
Part (a) is trivial. Part (b) is also straightforward to see by taking derivatives and then making the
transformation $z=(a+be^y)/(1+e^y)$ in the integration associated with the expectation $E_{f_Y}$.

To establish part (c), we prove the equivalent condition of Lipschitz continuity of $\left(\log f_Y(y)\right)'$, that is, 
the absolute value of the second derivative of
\[\psi(y)=\log f_Y(y)=\log (b-a)+y-2\log\left(1+e^y\right)+\log f\left(\frac{a+be^y}{1+e^y}\right)\]
is bounded. 

Note that
\begin{align}
\psi''(y) &= -\frac{2e^y}{\left(1+e^y\right)^2}+\left(\log f_X(z)\right)''\left[\frac{(b-a)e^y}{\left(1+e^y\right)^2}\right]^2
+\left(\log f_X(z)\right)'(b-a)\frac{e^y(1-e^y)}{(1+e^y)^3},
\label{eq:Lip1}
\end{align}
with $z=\left(a+be^y\right)/(1+e^y)$. 
Hence, noting that $\frac{e^y\left|(1-e^y)\right|}{(1+e^y)^3}\leq \frac{e^y(1+e^y)}{(1+e^y)^3}=\frac{e^y}{(1+e^y)^2}$,
we have
\begin{align}
\left|\psi''(y)\right| 
&\leq \frac{2e^y}{\left(1+e^y\right)^2}+
\left|\left(\log f_X(z)\right)''\right|\left[\frac{(b-a)e^y}{\left(1+e^y\right)^2}\right]^2
+\left|\left(\log f_X(z)\right)'\right|(b-a)\frac{e^y}{(1+e^y)^2}\notag\\
&\leq 2 +(b-a)^2\left|\left(\log f_X(z)\right)''\right|
+(b-a)\left|\left(\log f_X(z)\right)'\right|.
%\quad\quad (\mbox{since}\ \ \frac{e^y\left|(1-e^y)\right|}{(1+e^y)^3}
%\leq \frac{e^y(1+e^y)}{(1+e^y)^3}=\frac{e^y}{(1+e^y)^2}.)
\label{eq:Lip2}
\end{align}
Since $\left(\log f_X(z)\right)'$ is Lipschitz continuous on $(a,b)$, this is clearly bounded on $(a,b)$, and 
by the equivalent characterization of Lipschitz continuity, $\left(\log f_X(z)\right)''$ is bounded on $(a,b)$. 
Hence, the right hand side of (\ref{eq:Lip2}) is bounded above, proving that 
$\left(\log f_Y\right)'$ is Lipschitz continuous on $\mathbb R$.

\end{proof}

Using Lemma \ref{lemma:lemma1}, we then have the following theorem, which is analogous to Theorem \ref{theorem:theorem1},
but deals with the transformed target density $f_Y$ instead of the original target $f_X$, which is supported on $(a,b)$.
\begin{theorem}
\label{theorem:theorem2}
Assume that $f_X$ is positive with at least three continuous derivatives and that the fourth derivative exists almost everywhere
on $(a,b)$. 
Also assume that $(\log f_X)'$ is Lipschitz continuous on $(a,b)$, 
and that (\ref{eq:assump1}) -- (\ref{eq:assump4}) hold. 
Let $Y^d_t=(Y_{t,1},\ldots,Y_{t,d})$, where $Y_{t,i}=\log\left(\frac{X_{t,i}-a}{b-X_{t,i}}\right);$ $i=1,\ldots,d$.
As before, we define 
${U_{t}}^{d} = {Y_{[dt],1}}$ ($[\cdot]$ denotes the integer part), 
the sped up first component of the actual additive TMCMC-induced Markov chain, associated with the logistic 
transformation of the original random variable $X_{[dt],1}$ supported on $(a,b)$. 
Let $Y^d_0\sim\pi_Y$, that is, the $d$-dimensional additive TMCMC chain is started at stationarity
(equivalently, $X^d_0\sim\pi_X$), and
let the transition be given by $(y_1,\ldots,y_d)\rightarrow (y_1+b_1\epsilon,\ldots,y_d+b_d\epsilon)$, where
for $i=1,\ldots,d$, $b_i=\pm 1$ with equal probability and 
$\epsilon\equiv\frac{\ell}{\sqrt{d}}\epsilon^*$, where $\epsilon^*\sim q(\cdot)I_{\{\epsilon^*>0\}}$. 
%N(0,1)I_{\{\epsilon^*>0\}}$.
We then have
\[
\{U^d_t;~t\geq 0\}\Rightarrow \{U_t;~t\geq 0\},
\]
where $U_0\sim f_Y$ and $\{U_t;~t\geq 0\}$ satisfies the Langevin stochastic differential equation (SDE)
\begin{equation}
dU_t=g(\ell)^{1/2}dB_t+\frac{1}{2}g(\ell)\left(\log f_Y(U_t)\right)'dt,
\label{eq:sde2}
\end{equation}
with $B_t$ denoting standard Brownian motion at time $t$,
\begin{equation}
%g(\ell)=4\ell^2\int_0^{\infty} u^2\Phi\left (- \frac{u\ell\sqrt{\mathbb{I}}}{2}\right)\phi(u)du;
g(\ell)=4\ell^2\int_0^{\infty} u^2\Phi\left (- \frac{u\ell\sqrt{\mathbb{I_Y}}}{2}\right)q(u)du;
\label{eq:diff_speed_iid2}
\end{equation}
$\Phi(\cdot)$ being the standard normal cumulative distribution function (cdf), and
\begin{align}
\mathbb I_Y&=E_{f_Y}\left(\frac{f_Y'(Y)}{f_Y(Y)}\right)^2\notag\\
&=E_{f_Y}\left[1-\frac{2e^Y}{1+e^Y}+\frac{f'_X\left(\frac{a+be^Y}{1+e^Y}\right)}{f_X\left(\frac{a+be^Y}{1+e^Y}\right)}
\times\frac{(b-a)e^{Y}}{(1+e^Y)^2}\right]^2\notag\\
&=\bigint_{-\infty}^{\infty}
\left[1-\frac{2e^y}{1+e^y}+\frac{f'_X\left(\frac{a+be^y}{1+e^y}\right)}{f_X\left(\frac{a+be^y}{1+e^y}\right)}
\times\frac{(b-a)e^{y}}{(1+e^y)^2}\right]^2
\frac{(b-a)e^{y}}{(1+e^y)^2}f_X\left(\frac{a+be^y}{1+e^y}\right)dy.
\label{eq:information}
\end{align}
\end{theorem}

\subsection{SDE associated with the original bounded random variables $X$}
\label{subsec:sde_x}
Theorem \ref{theorem:theorem2} gives the SDE and the diffusion speed $g(\ell)$ associated with $U^d_t=Y_{[dt],1}$.
However, we are interested in the SDE and the diffusion speed associated with 
\begin{equation}
V^d_t=X_{[dt],1}=\frac{a+b e^{U^d_t}}{1+e^{U^d_t}}.
\label{eq:v}
\end{equation}

In this regard, we have the following theorem:

\begin{theorem}
\label{theorem:theorem3}
Under the assumptions of Theorem \ref{theorem:theorem2} it holds that
\[
\{V^d_t;~t\geq 0\}\Rightarrow \{V_t;~t\geq 0\},
\]
where $V_0\sim f_X$ and $\{V_t;~t\geq 0\}$ satisfies the SDE
\begin{equation}
\frac{(b-a)dV_t}{(V_t-a)(b-V_t)}=g(\ell)^{1/2}dB_t+
\frac{1}{2}g(\ell)\left\{\left(\log f_Y\left(\log \left(\frac{V_t-a}{b-V_t}\right)\right)\right)'
+\left(\frac{b+a-2V_t}{b-a}\right)\right\}dt.
\label{eq:sde2_V}
\end{equation}
\end{theorem}

\begin{proof}
Since $\{U^d_t;~t\geq 0\}\Rightarrow \{U_t;~t\geq 0\}$, it follows from (\ref{eq:v}) that 
$\{V^d_t;~t\geq 0\}\Rightarrow \{V_t;~t\geq 0\}$.
SDE (\ref{eq:sde2_V}) follows from (\ref{eq:sde2}) by using transformation (\ref{eq:v}) 
and applying the It\^{o} formula. %This SDE is associated with the original bounded random variables $X$.

\end{proof}

\subsection{Notion of diffusion speed associated with the original bounded random variables $X$}
\label{subsec:diffusion_speed_x}

Since the SDE (\ref{eq:sde2_V}) is not of the same form as (\ref{eq:sde2}) where a measure of
diffusion speed, $g(\ell)$, is well-defined, one may enquire if such notion of diffusion speed
at all exists in the case of (\ref{eq:sde2_V}). Intuitively, SDE (\ref{eq:sde2_V}) must have exactly
the same diffusion speed as (\ref{eq:sde2}), because of the bijection (\ref{eq:v}). It follows
from Theorem \ref{theorem:theorem4} below that
this is indeed the case.

\begin{theorem}
\label{theorem:theorem4}
Assume that $\{Z_t;~t\geq 0\}$ satisfies the SDE
\begin{equation}
\frac{(b-a)dZ_t}{(Z_t-a)(b-Z_t)}=dB_t+
\frac{1}{2}\left\{\left(\log f_Y\left(\log \left(\frac{Z_t-a}{b-Z_t}\right)\right)\right)'
+\left(\frac{b+a-2Z_t}{b-a}\right)\right\}dt.
\label{eq:sde_Z}
\end{equation}
Then $\{V_t;~t\geq 0\}=\{Z_{g(\ell)t};~t\geq 0\}$ satisfies SDE (\ref{eq:sde2_V}).

\end{theorem}

\begin{proof}
The proof is analogous to the arguments of \ctn{Bedard06} who clarify the 
notion of diffusion speed in the case of Langevin SDE.

Let $s=g(\ell)t$, so that $ds=g(\ell)dt$.
Hence,
\begin{align}
dZ_s
&=\frac{(Z_s-a)(b-Z_s)}{(b-a)}\left\{dB_s+\frac{1}{2}\left\{\left(\log f_Y\left(\log \left(\frac{Z_s-a}{b-Z_s}\right)\right)\right)'
+\left(\frac{b+a-2Z_s}{b-a}\right)\right\}ds\right\}\notag\\
&=\frac{(Z_{g(\ell)t}-a)(b-Z_{g(\ell)t})}{(b-a)}\notag\\
&\quad\quad\times\left\{\sqrt{g(\ell)dt}
+\frac{1}{2}\left\{\left(\log f_Y\left(\log \left(\frac{Z_{g(\ell)t}-a}{b-Z_{g(\ell)t}}\right)\right)\right)'
+\left(\frac{b+a-2Z_{g(\ell)t}}{b-a}\right)\right\}g(\ell)dt\right\}\notag\\
&=\frac{(Z_{g(\ell)t}-a)(b-Z_{g(\ell)t})}{(b-a)}\notag\\
&\quad\quad\times\left\{\sqrt{g(\ell)}dB_t+
\frac{1}{2}\left\{\left(\log f_Y\left(\log \left(\frac{Z_{g(\ell)t}-a}{b-Z_{g(\ell)t}}\right)\right)\right)'
+\left(\frac{b+a-2Z_{g(\ell)t}}{b-a}\right)\right\}g(\ell)dt\right\}\notag\\
&=\frac{(V_t-a)(b-V_t)}{(b-a)}\notag\\
&\quad\quad\times\left\{\sqrt{g(\ell)}dB_t+\frac{1}{2}\left\{\left(\log f_Y\left(\log \left(\frac{V_t-a}{b-V_t}\right)\right)\right)'
+\left(\frac{b+a-2V_t}{b-a}\right)\right\}g(\ell)dt\right\}\notag\\
&=dV_t.\notag
\end{align}
\end{proof}

Theorem \ref{theorem:theorem4} shows that if $Z_t$ is interpreted as a process with unit speed measure, then the limiting process 
$V_t$ is a ``sped-up" version of $Z_t$ by the quantity $g(\ell)$. 
Hence, $g(\ell)$ can be interpreted as a measure of the diffusion speed of SDE (\ref{eq:sde2_V}).
Thus, it makes sense to maximize $g(\ell)$ with respect to $\ell$ to obtain optimal scaling even when
the original random variables $X$ are bounded.

It is clear that exactly the same ideas carry over to situations where the target is a product of 
independent but non-identical densities (assuming that the individual densities
have the same support), and for TMCMC within Gibbs algorithms, as considered in \ctn{Dey13}. We omit details for brevity.

\section{Optimal scalings and acceptance rates with respect to different proposal distributions and target densities
in our SDE based approach}
\label{sec:optimal}

From Theorem \ref{theorem:theorem2} the optimal scales and the optimal acceptance rates under different proposal distributions
can be obtained as follows. 
Let $\ell^*$ be the maximizer of %(\ref{eq:diff_speed_iid2}) when $\mathbb I_Y$ is deliberately set equal to 1. 
\begin{equation}
%g(\ell)=4\ell^2\int_0^{\infty} u^2\Phi\left (- \frac{u\ell\sqrt{\mathbb{I}}}{2}\right)\phi(u)du;
g^*(\ell)=4\ell^2\int_0^{\infty} u^2\Phi\left (- \frac{u\ell}{2}\right)q(u)du.
\label{eq:diff_speed_iid3}
\end{equation}
Then the optimal scale is given by
\begin{equation}
\ell_{opt}=\frac{\ell^*}{\sqrt{\mathbb I_Y}},
\label{eq:optimal_scale}
\end{equation}
and the corresponding optimal acceptance rate is given by 
\begin{eqnarray}
\alpha_{opt}&=&4\int_0^{\infty} \Phi\left (- \frac{u\ell_{opt}\sqrt{\mathbb I_Y}}{2}\right)q(u)du\nonumber\\
&=&4\int_0^{\infty} \Phi\left (- \frac{u\ell^*}{2}\right)q(u)du.
\label{eq:optimal_acceptance_rate}
\end{eqnarray}

Thus, $\ell^*$ depends only upon the proposal density $q(\cdot)$, the optimal scale $\ell_{opt}$ depends upon
$q(\cdot)$ as well as Fisher's information $\mathbb I_Y$, and the optimal acceptance rate depends upon
$q(\cdot)$ only. 
Note that the optimal scale depends upon the chosen logit transformation $y_i=\log\left(\frac{x_i-a}{b-x_i}\right)$
only through $\mathbb I_Y$. Since the optimal acceptance rate is independent of $\mathbb I_Y$, it is clearly
independent of any bijective transformation used for mapping $x_i$ to $y_i$. 
As is also clear, the optimal acceptance rate does not depend upon the target density or its support.

Table \ref{table:table1} displays the optimal scales and optimal acceptance rates with respect to different
choices of the proposal density $q(\cdot)$ and target densities associated with truncated normal and
uniform distributions. As the degrees of freedom of the Student's $t$ proposal density increases from 1  to 5,
that is, as the proposal distribution approaches the $N(0,1)$ density beginning with the $Cauchy(0,1)$ density,
it is seen that optimal scales and optimal acceptance rates increase and approach those associated
with the $N(0,1)$ proposal in the TMCMC case; recall, in particular, that the optimal acceptance rate 
of additive TMCMC for the $N(0,1)$ proposal is 0.439. 

This increase in the optimal scales and the optimal acceptance rates are to be expected since
the successive proposal distributions for increasing degrees of freedom have progressively thinner tails
resulting in greater acceptance rates -- the optimal scales increase to compensate for the thin tails
so that the acceptance rates do not increase too fast.

Note that when the proposal distribution $q(\cdot)$ is $U(0,1)$, the optimal scale 
is much higher than those associated with the $t$-distributions. This is again
to be expected since unlike for $t$-distribution based proposals, here the proposed $\epsilon^*\sim U(0,1)$
must lie within $(0,1)$ with probability one, so that the resultant proposed values 
$x_i+b_i\frac{\ell}{\sqrt{d}}\epsilon^*$ are quite close to $x_i$, resulting in too high acceptance rate
unless the scale $\ell$ is quite large. It is also noteworthy that in this example this case of $U(0,1)$ proposal corresponds
to target distribution with bounded support as well as proposal with bounded support.

\begin{table}
%\begin{sidewaystable}[h]
\centering
\caption{Optimal scales ($\ell_{opt}$) and optimal acceptance rates ($\alpha_{opt}$)
under different proposal distributions when the target densities are $iid$ products of
$N(0,1)$ truncated on $(-1,1)$ and $U(-1,1)$, respectively.}
%\vspace{2cm}

\begin{tabular}{|c|c|c|c|c|}
\hline
\multirow{2}{*}{\backslashbox{Proposal}{Target}} 
& \multicolumn{2}{|c|}{$f_X(x)=N(x;0,1)I_{(-1,1)}(x)$} & \multicolumn{2}{|c|}{$f_X(x)=U(x;(-1,1))$}\\  
\cline{2-5}
& $\ell_{opt}$ & $\alpha_{opt}$ & $\ell_{opt}$ & $\alpha_{opt}$\\ 
\hline
$q(\cdot)= t_1(0,1)~(Cauchy(0,1))$ & 2.934 & 0.380 & 3.358 & 0.380\\
$q(\cdot)= t_2(0,1)$ & 3.196 & 0.413 & 3.658 & 0.413\\
$q(\cdot)= t_3(0,1)$ & 3.319 & 0.423 & 3.799 & 0.423\\
$q(\cdot)= t_4(0,1)$ & 3.391 & 0.428 & 3.882 & 0.428\\
$q(\cdot)= t_5(0,1)$ & 3.439 & 0.431 & 3.936 & 0.431\\
$q(\cdot)= U([0,1])$ & 5.572 & 0.420  & 6.377 & 0.420\\
\hline
\end{tabular}
\label{table:table1}
%\end{sidewaystable}
\end{table}

\section{Comparison with the ESJD approach associated with RWM}
\label{sec:comparison}

\ctn{NealRoberts11} consider $X=(X_{1},X_{2},\ldots,X_{d})$ to be a random vector with $0<X_{i}<1$ for each $i$ 
and that the density $\pi$ for $X$ has the following form:
\begin{equation}
\pi(x)= \prod_{i=1}^{d} f_X(x_{i}) = \prod_{i=1}^{d} \exp(h(x_{i}));  
\hspace{0.5 cm} 0<x_{i}<1; \hspace{0.5 cm} \forall~ i=1,2,\ldots,d,
\label{eq:density}
\end{equation}
where $h$ is continuously differentiable on $[0,1]$.

Theorem 4.1 of \ctn{NealRoberts11} provides ESJD based optimal scaling of RWM with the $Cauchy(0,1)$ proposal
when the target distribution is of the form (\ref{eq:density}). The scaling they consider is
$\frac{\ell}{d\log d}$. In other words, \ctn{NealRoberts11} consider RWM of the form
$x_i+\frac{\ell}{d\log d}\epsilon^*_i$, where $\epsilon^*_i\stackrel{iid}{\sim}Cauchy(0,1)$.
The optimal acceptance rate in this case, provided in Theorem 4.1 of \ctn{NealRoberts11}, is given 
by 0.368.

Our result in this regard (which is actually a conjecture; see Section \ref{subsec:conjecture}) 
is quite significantly different from that of Theorem 4.1 of \ctn{NealRoberts11}.
Indeed, our optimal acceptance rate with $Cauchy(0,1)$ proposal for $\epsilon^*$ is 0.380 (see also the first row
of Table \ref{table:table1}), which is higher
than that obtained by \ctn{NealRoberts11}. But more significantly, while 
the scaling in the case of additive TMCMC is of the form $\ell/\sqrt{d}$, that of RWM based on ESJD is of the form $\ell/(d\log d)$.
Consequently, with $Cauchy(0,1)$ proposal, the former is expected to explore the target distribution in much less number of iterations compared to the latter. 
%the complexity of our diffusion based 
%additive TMCMC approach is only $O(d)$, that of RWM based on ESJD is $O\left(\left(d\log d\right)^2\right)$.
%That is, additive TMCMC with $Cauchy(0,1)$ proposal will take of the order of $d$ steps to explore the 
%target distribution, while RWM with $Cauchy(0,1)$ proposal will take of the order of $\left(d\log d\right)^2$ steps.
This seems to be a very significant advantage of our TMCMC approach compared with RWM. 

In order to assess the performance of additive TMCMC and RWM for Cauchy proposal, we conduct simulation studies,
assuming the target density to be a product of $N(0,1)$ densities truncated on $(-1,1)$. The additive TMCMC
considers moves of the type 
\[
(x_1,\ldots,x_d)\rightarrow \left(x_1+\frac{\ell_{TMCMC,opt}}{\sqrt{d}}b_1\epsilon^*,\ldots,
x_d+\frac{\ell_{TMCMC,opt}}{\sqrt{d}}b_d\epsilon^*\right),
\]
where $\epsilon^*\sim Cauchy(0,1)$ such that $\epsilon^*>0$, and $b_i=\pm 1$ with probability each, for $i=1,\ldots,d$. 
On the other hand, RWM considers moves of the type
\[
(x_1,\ldots,x_d)\rightarrow \left(x_1+\frac{\ell_{RWM,opt}}{d\log d}\epsilon^*_d,\ldots,
x_d+\frac{\ell_{RWM,opt}}{d\log d}\epsilon^*_d\right),
\]
with $\epsilon^*_i\stackrel{iid}{\sim}Cauchy(0,1)$, for $i=1,\ldots,d$.  

We conduct three experiments, with $d=10,50,100$, comparing the autocorrelations of TMCMC and RWM chains in each case.
In all the cases, we ran the two algorithms for $10^6$ iterations, starting with a draw from the target distribution.
For TMCMC, we set $\ell_{TMCMC,opt}=2.934$, as provided in Table \ref{table:table1}. The empirical acceptance rates, correct up to
three decimal places, turned out to be 0.381, 0.379 and 0.380, respectively, for dimensions $d=10,50$ and $100$.
Thus, the empirical acceptance rates turned out to be very accurate, even for dimension as small as $d=10$. These empirical results
also serve to strengthen our belief regarding the conjecture made in Section \ref{subsec:conjecture}.

For RWM we tuned 
$\ell_{RWM,opt}$ such that the empirical acceptance rate is approximately 0.368. For dimension $d=10,50,100$, we obtain 
$\ell_{RWM,opt}=1.6,2.06,2.26$, which yielded empirical acceptance rates 0.365, 0.374 and 0.368, respectively, correct up to
three decimal places.

As already mentioned in Section \ref{subsubsec:computational_gain}, RWM took around 43 minutes to perform $10^6$ iterations
for 100 dimensions,
while TMCMC required only around 28 minutes to perform the same number of iterations.

Figure \ref{fig:comparison2} compares the autocorrelations associated with TMCMC (thick, green vertical lines) and RWM (red vertical lines) chains for dimensions 10, 50 and 100.
In every case, the autocorrelations corresponding to TMCMC are uniformly lower than those based on RWM. This clearly appears to
be the consequence of lesser complexity of additive TMCMC with scaling $\ell/\sqrt{d}$ as opposed to that of RWM with scaling $\ell/(d\log d)$.

Apart from the autocorrelations, we have also compared TMCMC with RWM with respect to the KS distance. For
$d=10$, the TMCMC and RWM based KS distances, up to three decimal places, are 0.006 and 0.008, respectively; 
for $d=50$, the respective distances are 0.013 and 0.035, and for $d=100$, the TMCMC based KS distance is 0.014, 
while that based on RWM is 0.041. In other words, TMCMC significantly outperforms RWM with respect to the
$Cauchy(0,1)$ proposal in terms of the KS distance.
%d = 10; TMCMC KS = 0.00615; RWM KS = 0.00793
%d = 50; TMCMC KS = 0.01307; RWM KS = 0.03483
%d = 100; TMCMC KS = 0.01376; RWM KS = 0.04121

%d = 10; 10000 iterations; TMCMC KS = 0.0462; RWM KS = 0.0771

%Since we are particularly interested in how accurate the asymptotic diffusion results are in relatively low dimensions,
%we set $d=10$.
%We ran both the algorithms for 10,000 iterations, beginning with the initial value drawn exactly from the target distribution,
%that is, we started at the stationary distribution, as demanded by Theorem \ref{theorem:theorem2}.
%For TMCMC, we set $\ell_{TMCMC,opt}=2.934$, as provided in Table \ref{table:table1}. The empirical acceptance rate turned out
%to be 0.376, which is very accurate, even for dimension $d=10$. For RWM we tuned 
%$\ell_{RWM,opt}$ such that the empirical acceptance rate is approximately 0.368. For $\ell_{RWM,opt}=1.6$, we obtained
%the empirical acceptance rate 0.371.

%Figure \ref{fig:comparison} shows the results of our simulation study. For the first 10,000 iterations TMCMC
%explored the target density more adequately than RWM, the traceplots indicate faster mixing of TMCMC compared to RWM, and
%the autocorrelation of TMCMC decayed much faster than that of RWM. All these are to be expected because
%of much lesser complexity of TMCMC compared to RWM, at least as far as the Cauchy proposal is concerned.

%

\begin{figure}%[htp]
\subfigure [$d=10$.]{ \label{fig:comp1}
\includegraphics[width=5.1cm]{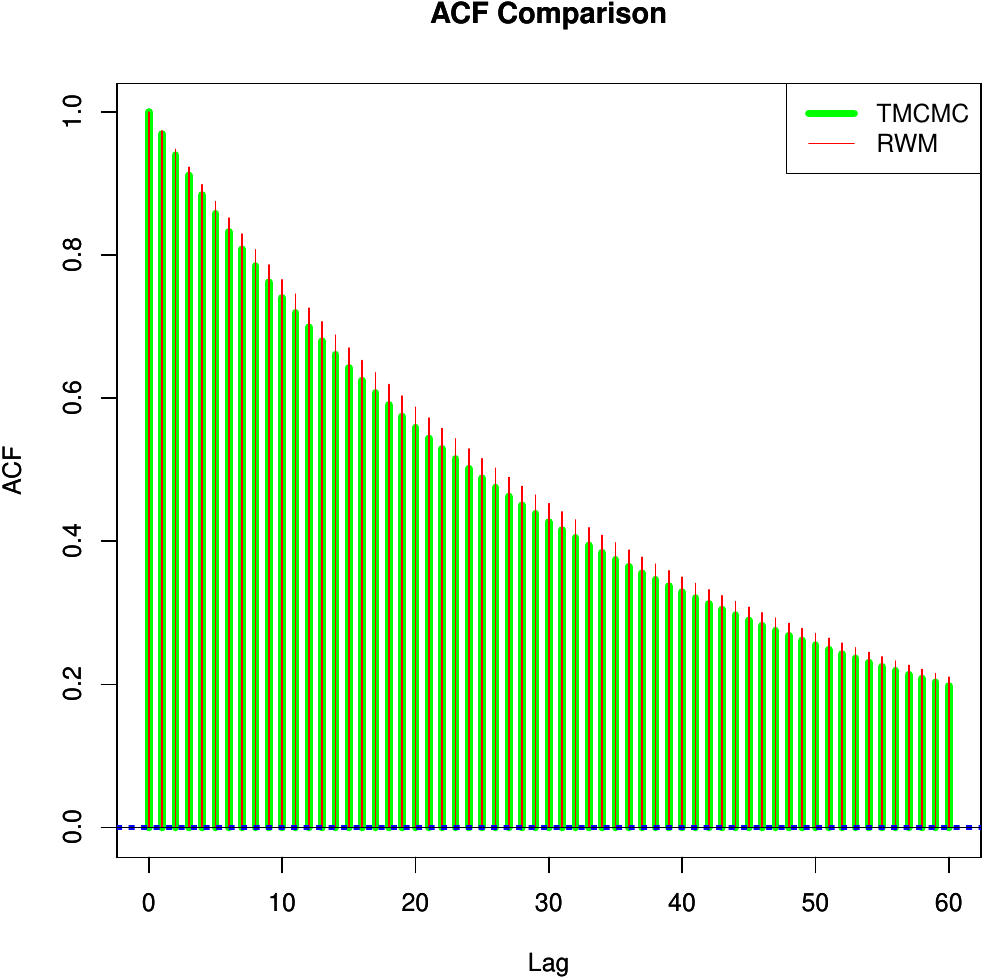}}
\hspace{-2mm}
\subfigure [$d=50$.]{ \label{fig:comp2}
\includegraphics[width=5.1cm]{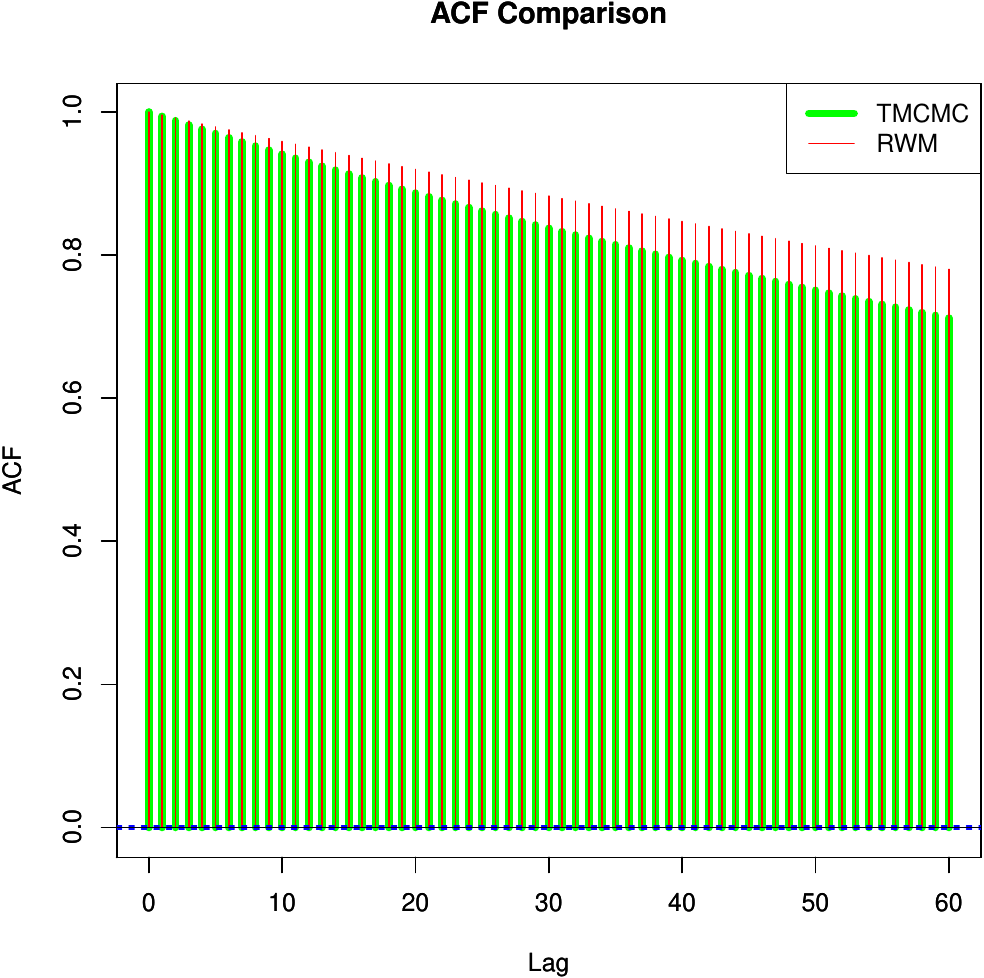}}
\hspace{-2mm}
\subfigure [$d=100$.]{ \label{fig:comp3}
\includegraphics[width=5.1cm]{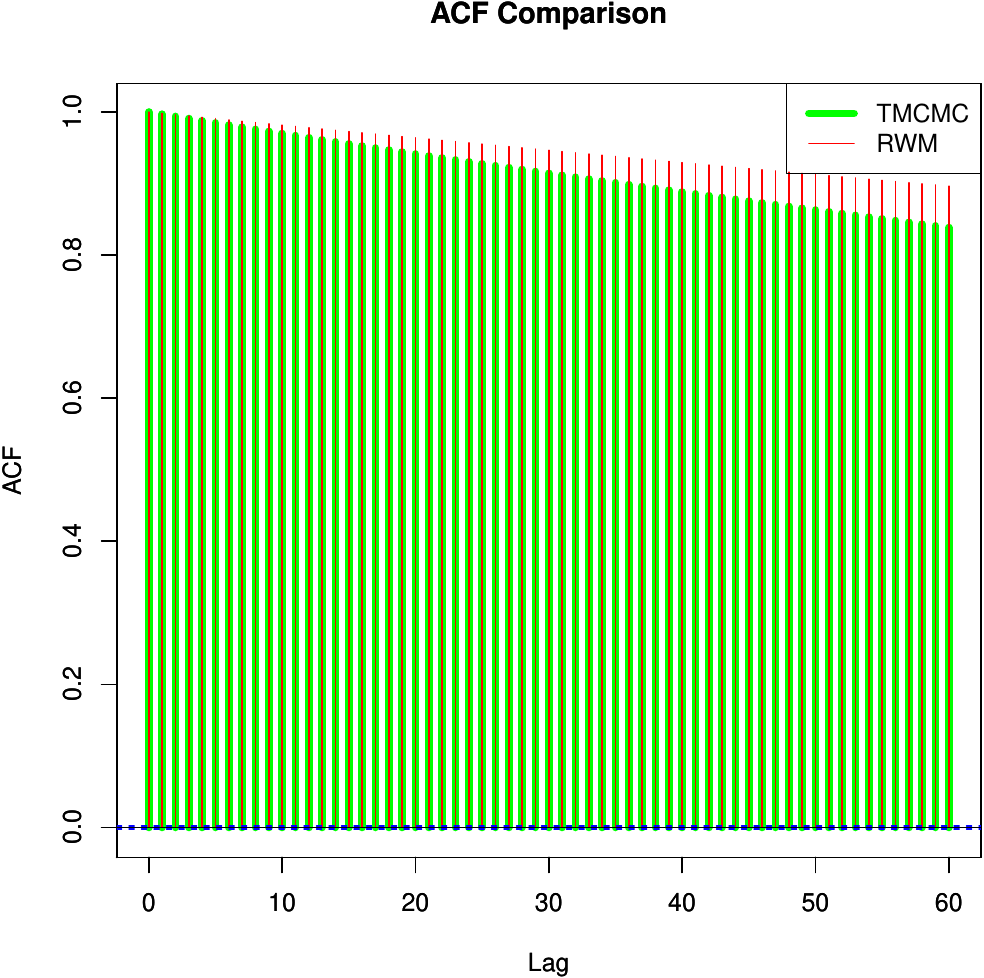}}
\caption{ Panels (a), (b) and (c) compare the autocorrelations based on $10^6$ iterations of additive TMCMC and RWM when 
the true target density is the product of $N(0,1)$ truncated
on $(-1,1)$, with dimensions $d=10,50,100$, respectively. 
%For $d=10$, with scale $1.6/(d\log d)$, based on the first 10,000 samples; the empirical acceptance rate is 0.371.
%The lower panels (d), (e) and (f) display the TMCMC based plots of the same 
%for $d=10$, with scale $2.934/\sqrt{d}$; the empirical acceptance rate is 0.376.
}
\label{fig:comparison2}
\end{figure}

Figure \ref{fig:comparison} magnifies the issue related to the speed of exploration of the target density 
by additive TMCMC and RWM, by comparing the two algorithms for the first 10,000 iterations when $d=10$.
As seen in the figure, in the first 10,000 iterations 
TMCMC explored the target density more adequately than RWM, the traceplots indicate faster mixing of TMCMC compared to RWM, and
the autocorrelation of TMCMC decayed much faster than that of RWM. 
In this case, the TMCMC based KS distance is $0.046$ while that based on RWM is $0.077$, confirming the visual insight
offered by Figure \ref{fig:comparison}.

\begin{figure}%[htp]
\subfigure [RWM histogram and true density.]{ \label{fig:rwm1}
\includegraphics[width=5cm]{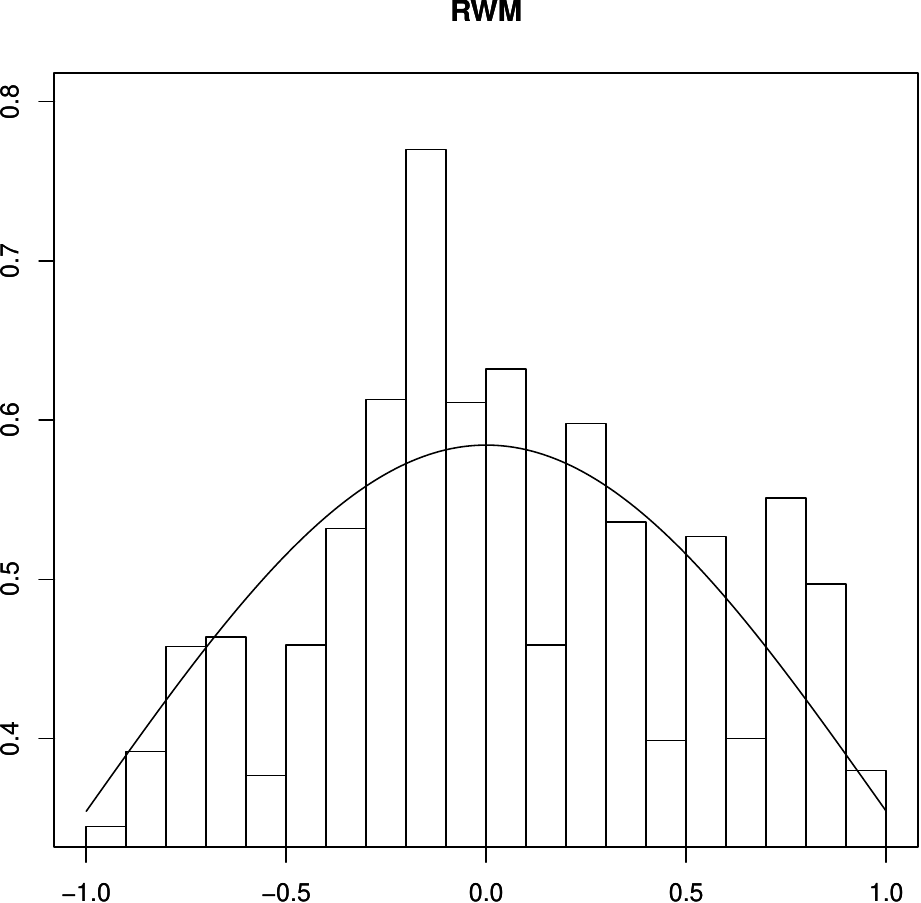}}
\hspace{-2mm}
\subfigure [RWM traceplot.]{ \label{fig:rwm2}
\includegraphics[width=5cm]{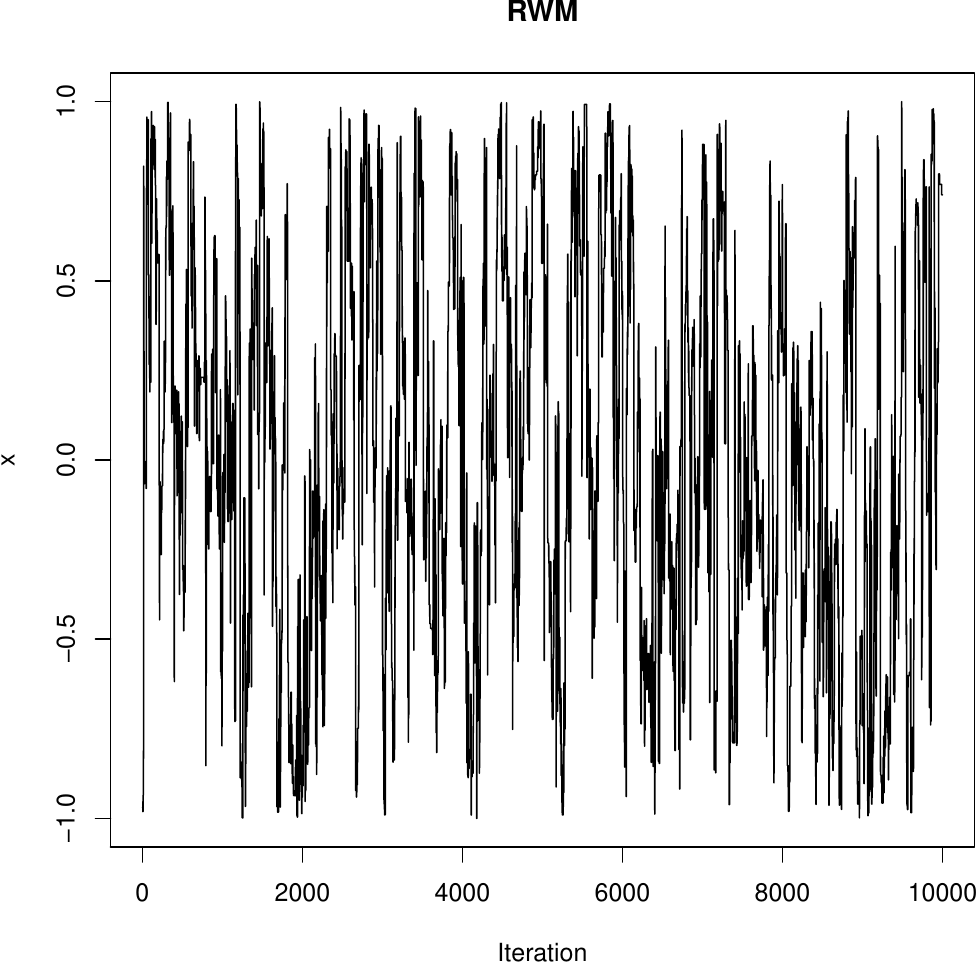}}
\hspace{-2mm}
\subfigure [RWM ACF.]{ \label{fig:rwm3}
\includegraphics[width=5cm]{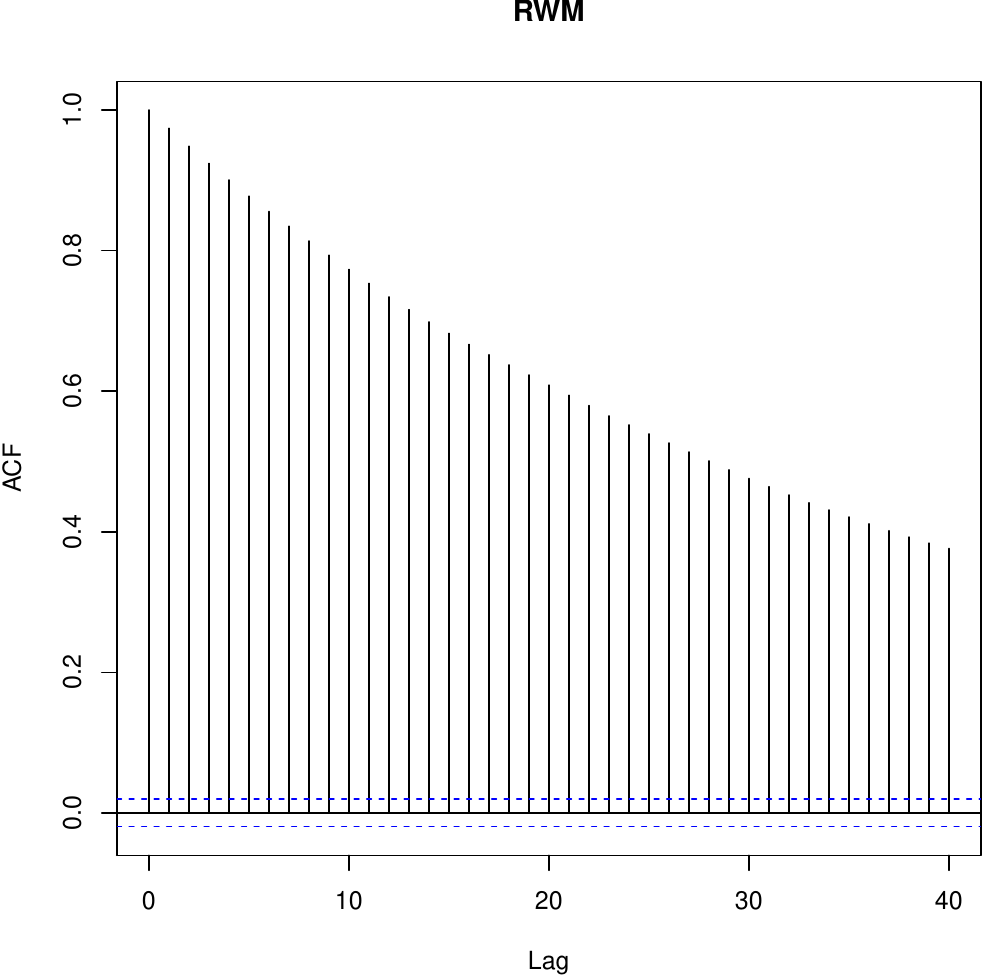}}\\
\subfigure [TMCMC histogram and true density.]{ \label{fig:tmcmc1}
\includegraphics[width=5cm]{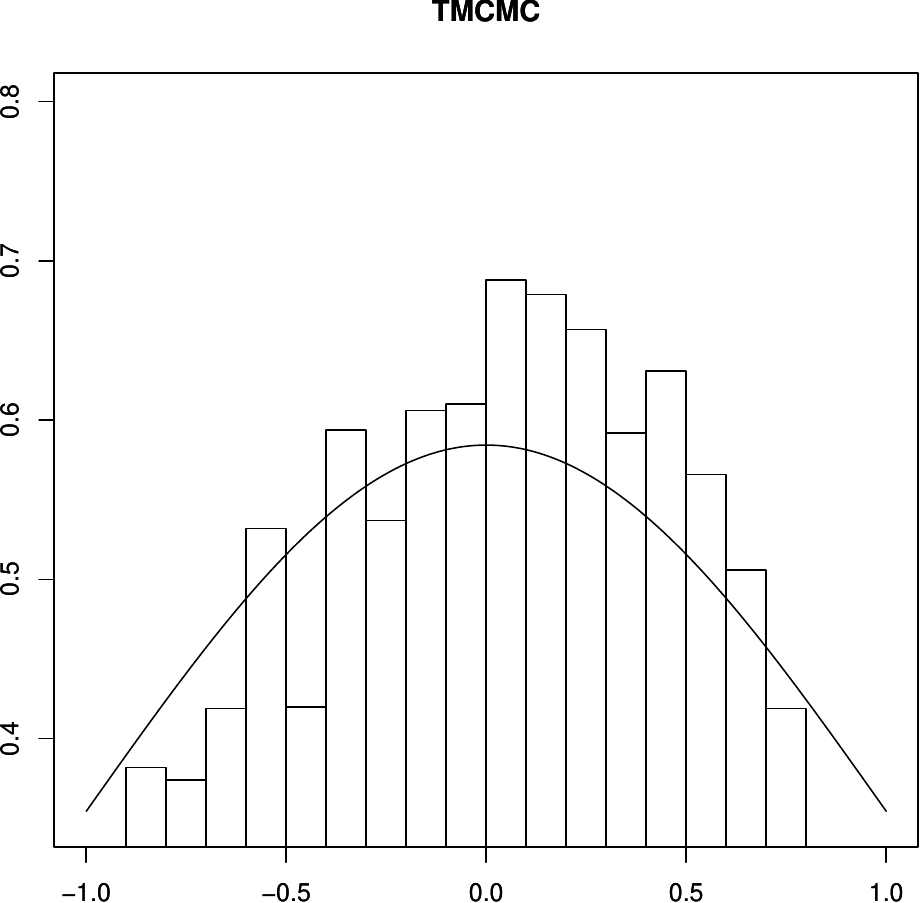}}
\subfigure [TMCMC traceplot.]{ \label{fig:tmcmc2}
\includegraphics[width=5cm]{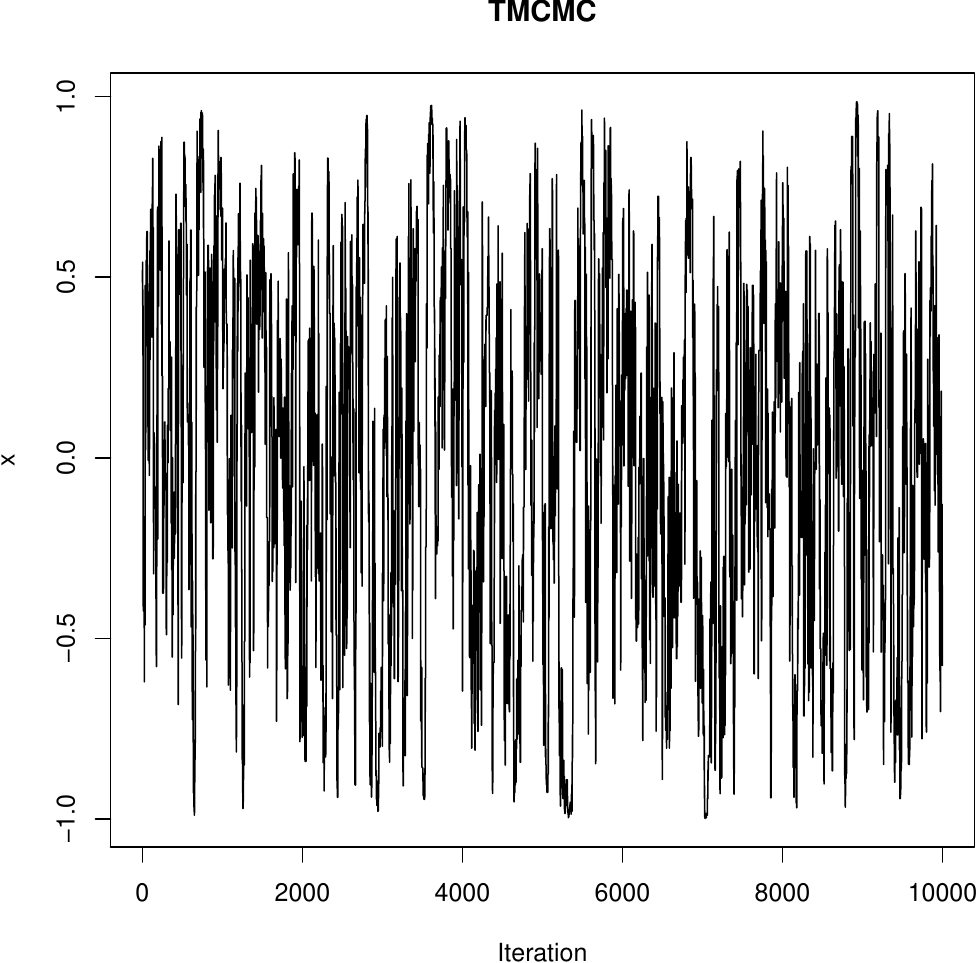}}
\subfigure [TMCMC ACF.]{ \label{fig:tmcmc3}
\includegraphics[width=5cm]{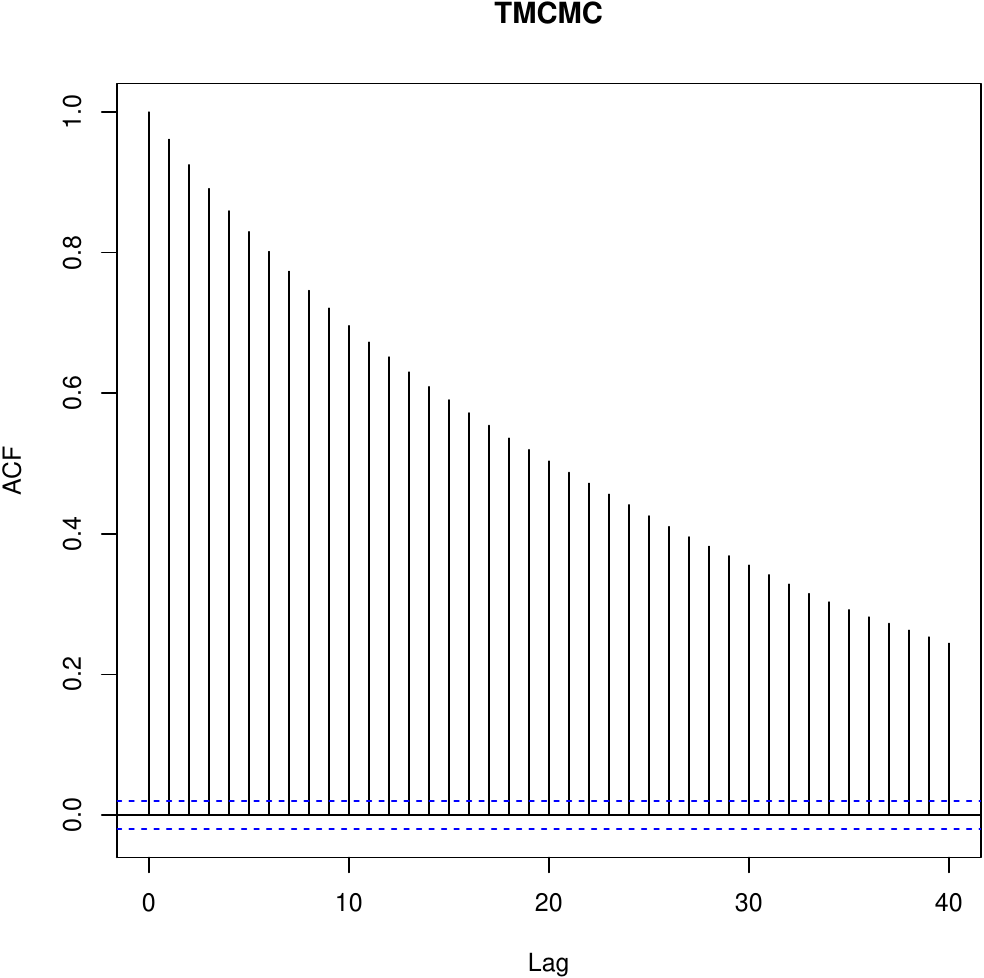}}
\caption{ The upper panels (a), (b) and (c) show the RWM based histogram and the true target density $N(0,1)$ truncated
on $(-1,1)$, traceplot and the autocorrelation
functions respectively for $d=10$, with scale $1.6/(d\log d)$, based on the first 10,000 samples.%; the empirical acceptance rate is 0.371.
The lower panels (d), (e) and (f) display the TMCMC based plots of the same 
for $d=10$, with scale $2.934/\sqrt{d}$.
%; the empirical acceptance rate is 0.376.
}
\label{fig:comparison}
\end{figure}

\section{Comparison of our optimal scaling theory with slice sampling}
\label{sec:slice}
Slice sampling is a well-known methodology of introducing auxiliary variables that aid in Gibbs sampling.
The general algorithm is associated with the factorization of the density $f(x)$ as $f(x)\propto\prod_{i=1}^mf_i(x)$,
where $f_i(x)$ are positive functions that need not be densities. Since 
$f_i(x)=\int I_{\left\{0\leq z_i\leq f_i(x)\right\}}dz_i$, it follows that one may introduce the auxiliary variables
$z_1,\ldots,z_m$ such that the joint distribution of $(x,z_1,\ldots,z_m)$ is proportional to
$\prod_{i=1}^mI_{\left\{0\leq z_i\leq f_i(x)\right\}}$, so that the marginal distribution of $x$ is $f$.
For $i=1,\ldots,m$, the full conditonal distribution of $z_i$ given $x$ is the uniform distribution on $[0,f_i(x)]$
and that of $x$ given $z_1,\ldots,z_m$ is the uniform distribution on the slice 
$\left\{y:f_i(y)\geq z_i,~i=1,\ldots,m\right\}$. Thus, a Gibbs sampling strategy can be envisaged for
sampling from the joint distribution of $(x,z_1,\ldots,z_m)$, and then discarding the samples of $z_1,\ldots,z_m$
to store the samples of $x\sim f$. This is the so-called slice sampling strategy, which often
induces good mixing properties for distributions with truncated support. For details, see
\ctn{Neal03}, \ctn{Robert04} and the references therein.
It is thus important to compare TMCMC and RWM based methods with slice sampler.

It is however, to be borne in mind, that it is not in general straightfoward to sample from
the full conditional of $x$ given $z_1,\ldots,z_m$, particularly when $m$ is large.
\ctn{Neal03} attempts to create proposals to deal with this problem but those are very specialized
proposals and are not expected to handle general situations (\ctn{Robert04}). 
Furthermore, \ctn{Roberts03} (see also \ctn{Robert04}) provide an example of a distribution for which
slice sampling performs poorly. 
Indeed, letting $\pi(z)\propto\exp\left(-\|z\|\right)$, where $z\in\mathbb R^d$
and $\|z\|=\sqrt{\sum_{i=1}^dz^2_i}$, note that $x=\|z\|$ is itself a Markov chain and in fact, a slice sampler
Markov chain for the distribution $\pi_d(x)\propto x^{d-1}\exp\left(-x\right)$; $x>0$. Here the factorization
is given by $f_1(x)=x^{d-1}$ and $f_2(x)=\exp\left(-x\right)$.
This is an example where the performance of the slice sampler deteriorates as $d$ increases.
Indeed, as demonstrated in \ctn{Robert04} by simulations, for $d=1$ and $5$, the slice sampler mixes reasonably well
with fast decreasing autocorrelatons but for $d=10$ and particularly for $d=50$, the performance of the slice sampler sharply
deteriorates.

We compare the performances of Gaussian proposal based additive TMCMC and RWM with slice sampler in 
the case of $\pi_d(x)$. For comparability with the results reported in \ctn{Robert04}, in each case we 
consider a sample of size $1000$ for TMCMC and RWM; 
we consider a burn-in of size $1000$ in each case. We tune additive TMCMC and RWM with scales of the form
$\ell/\sqrt{d}$ such that the acceptance rates are approximately $0.439$ and $0.234$ respectively, for $d=1,5,10,50$.
Figures \ref{fig:slice_trace} and \ref{fig:slice_autocorrelation} shows the trace plots and the autocorrelation plots
associated with TMCMC and RWM. Observe that compared to Figure 8.5 of \ctn{Robert04}, the trace plots and the autocorrelation
plots with respect to both TMCMC and RWM indicate much superior performance compared to slice sampler, for each dimension
$d=1,5,10,50$. Moreover, the plots shown in  Figures \ref{fig:slice_trace} and \ref{fig:slice_autocorrelation}
show that, unlike the slice sampler, the performances of TMCMC and RWM do not deteriorate with increasing dimensionality.
We also take this opportunity to compare additive TMCMC and RWM in this example. As shown in Figure 
\ref{fig:slice_autocorrelation}, the autocorrelations based on additive TMCMC decrease faster than those of RWM, for
all the values of $d$ considered; this is in keeping with the visual information offered by the trace plots of
Figure \ref{fig:slice_trace}. We also consider the KS distances between the empirical distribution functions
associated with the first $500$ and the last $500$ iterations after the burn-in period for comparing additive TMCMC and RWM.
Table \ref{table:slice} shows that the KS distances associated with TMCMC are smaller than those of RWM for all the values
of $d$ considered. Thus, RWM is again outperformed by TMCMC, while slice sampling performs the worst in this example.
The numerical results, in conjunction with the difficulty of implementation of slice samplers in complex problems, 
certainly leads us to recommend TMCMC for superior performances in general situations.

\begin{figure}%[htp]
\subfigure [TMCMC traceplot for $d=1$.]{ \label{fig:slice_tmcmc1}
\includegraphics[height=4cm,width=7cm]{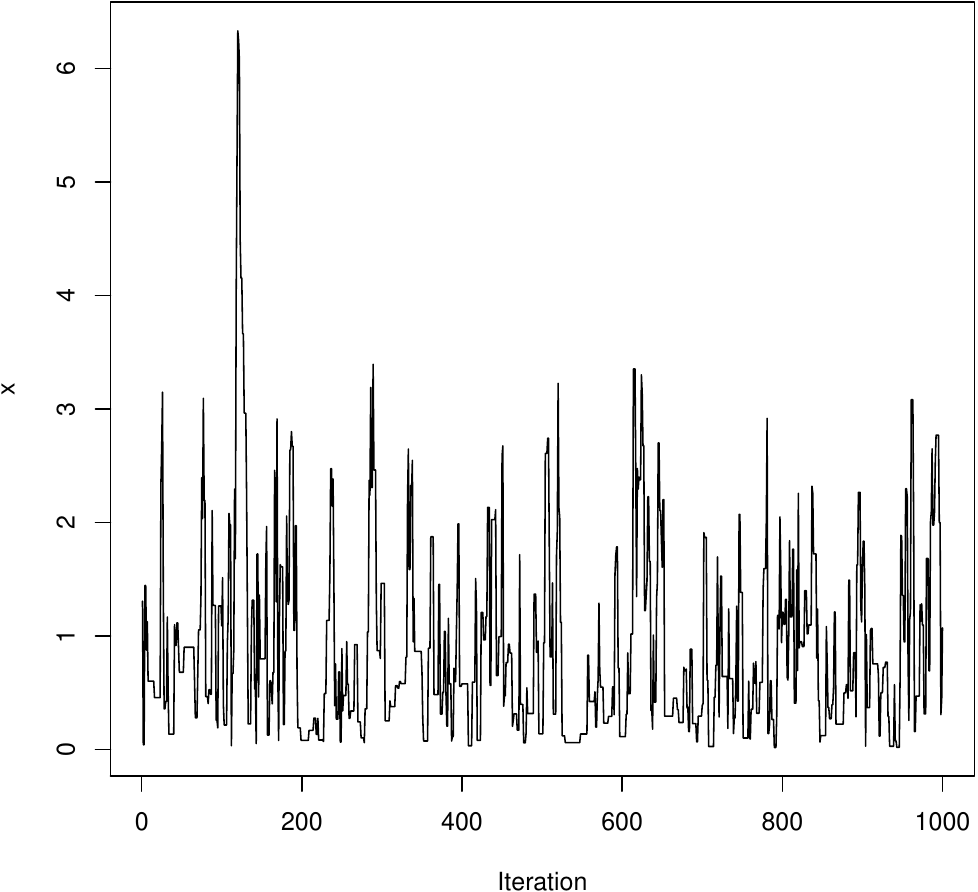}}
\hspace{2mm}
\subfigure [RWM traceplot for $d=1$.]{ \label{fig:slice_rwm1}
\includegraphics[height=4cm,width=7cm]{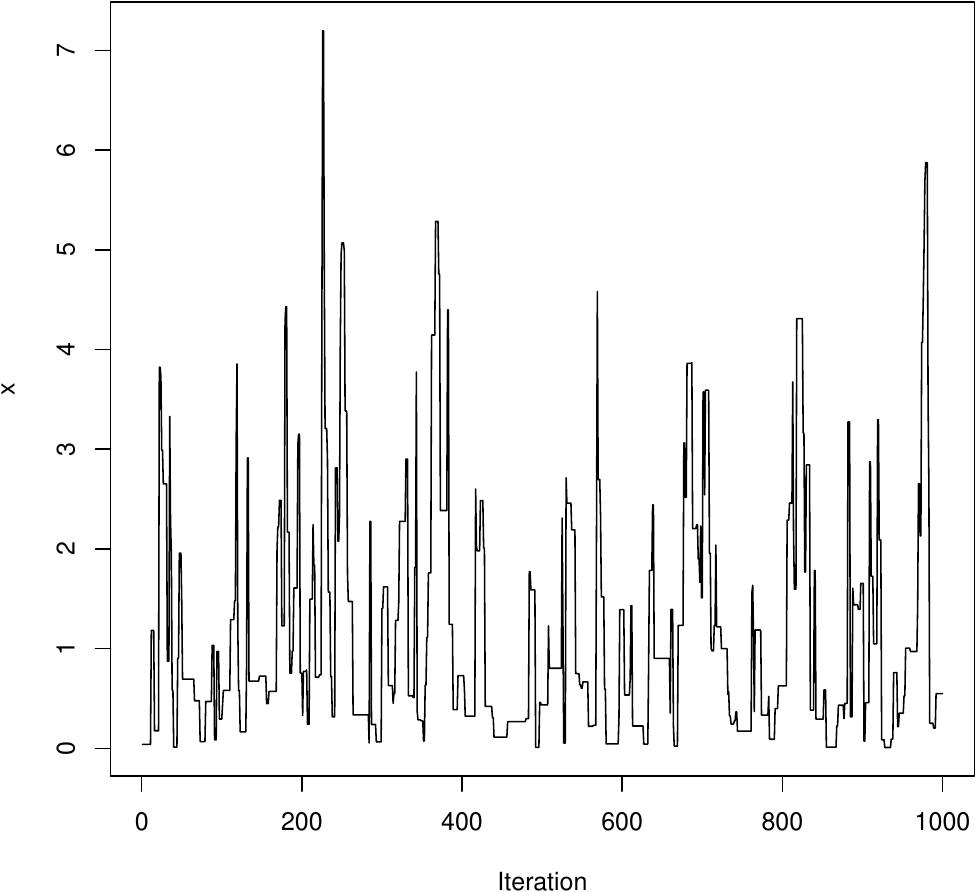}}\\
\vspace{2mm}
\subfigure [TMCMC traceplot for $d=5$.]{ \label{fig:slice_tmcmc5}
\includegraphics[height=4cm,width=7cm]{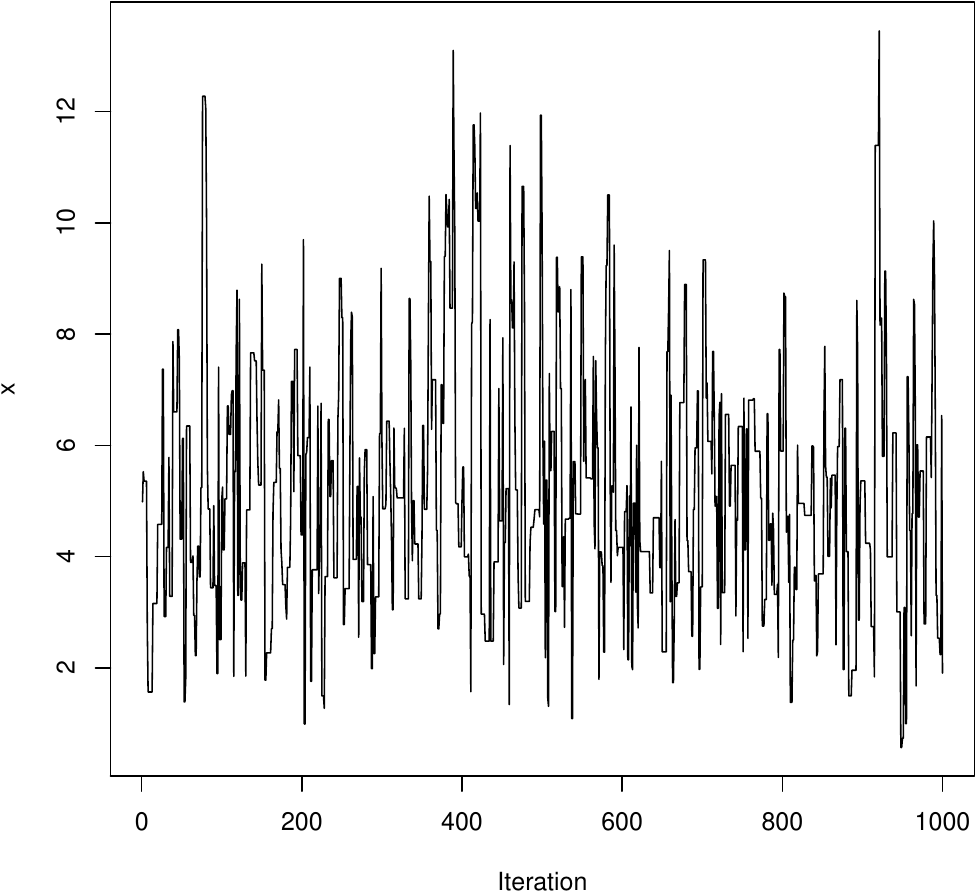}}
\hspace{2mm}
\subfigure [RWM traceplot for $d=5$.]{ \label{fig:slice_rwm5}
\includegraphics[height=4cm,width=7cm]{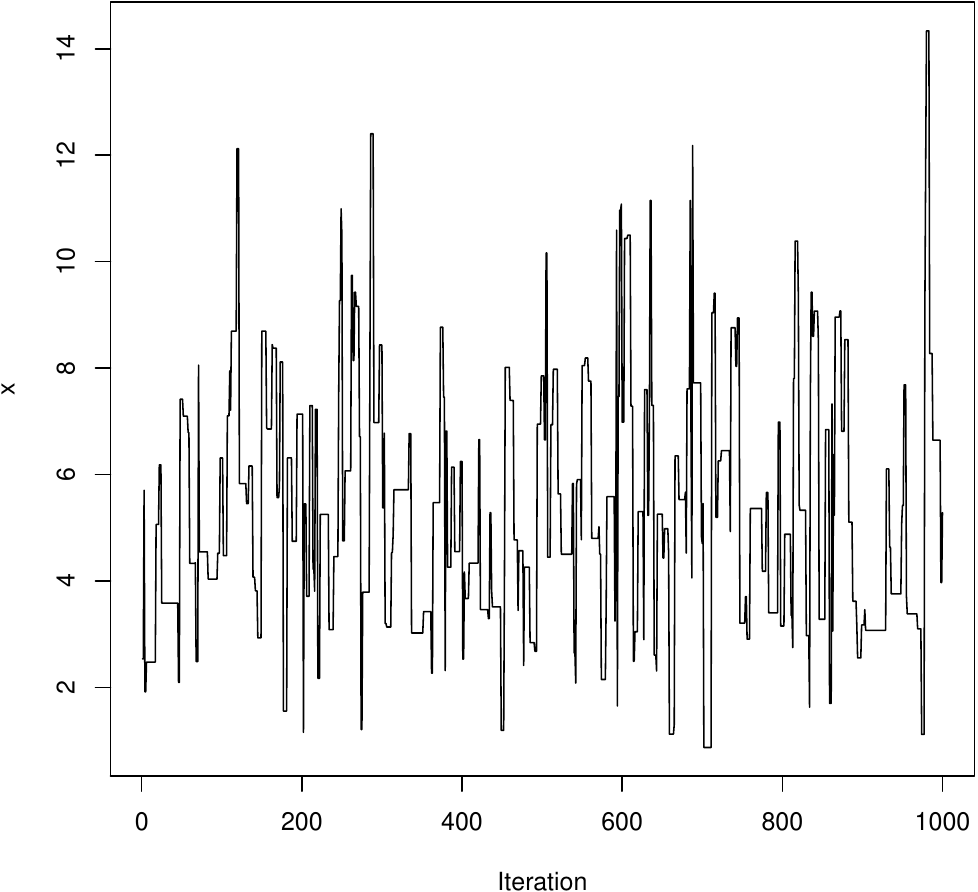}}\\
\vspace{2mm}
\subfigure [TMCMC traceplot for $d=10$.]{ \label{fig:slice_tmcmc10}
\includegraphics[height=4cm,width=7cm]{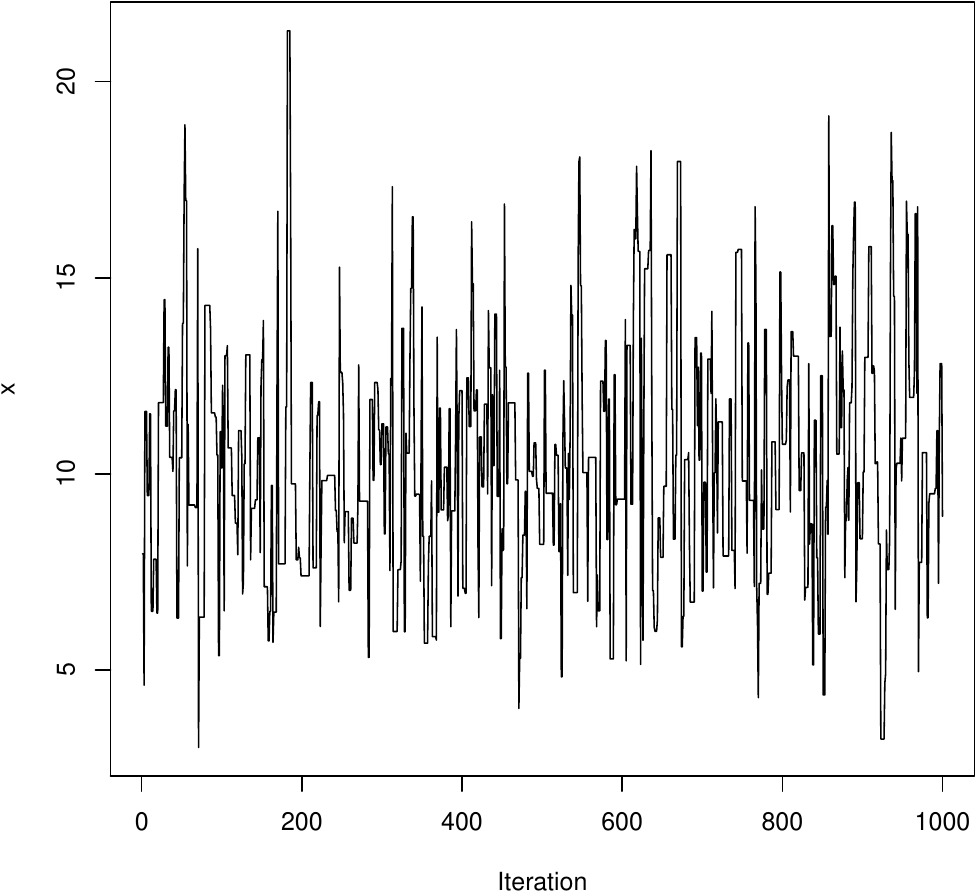}}
\hspace{2mm}
\subfigure [RWM traceplot for $d=10$.]{ \label{fig:slice_rwm10}
\includegraphics[height=4cm,width=7cm]{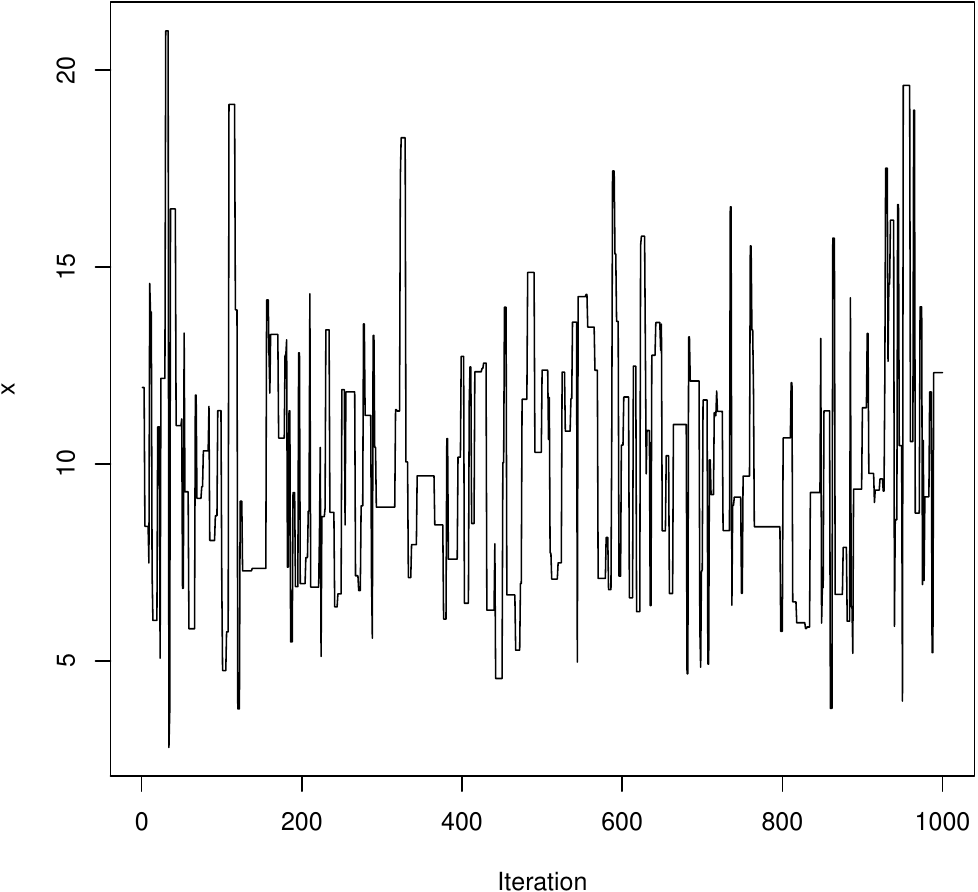}}\\
\vspace{2mm}
\subfigure [TMCMC traceplot for $d=50$.]{ \label{fig:slice_tmcmc50}
\includegraphics[height=4cm,width=7cm]{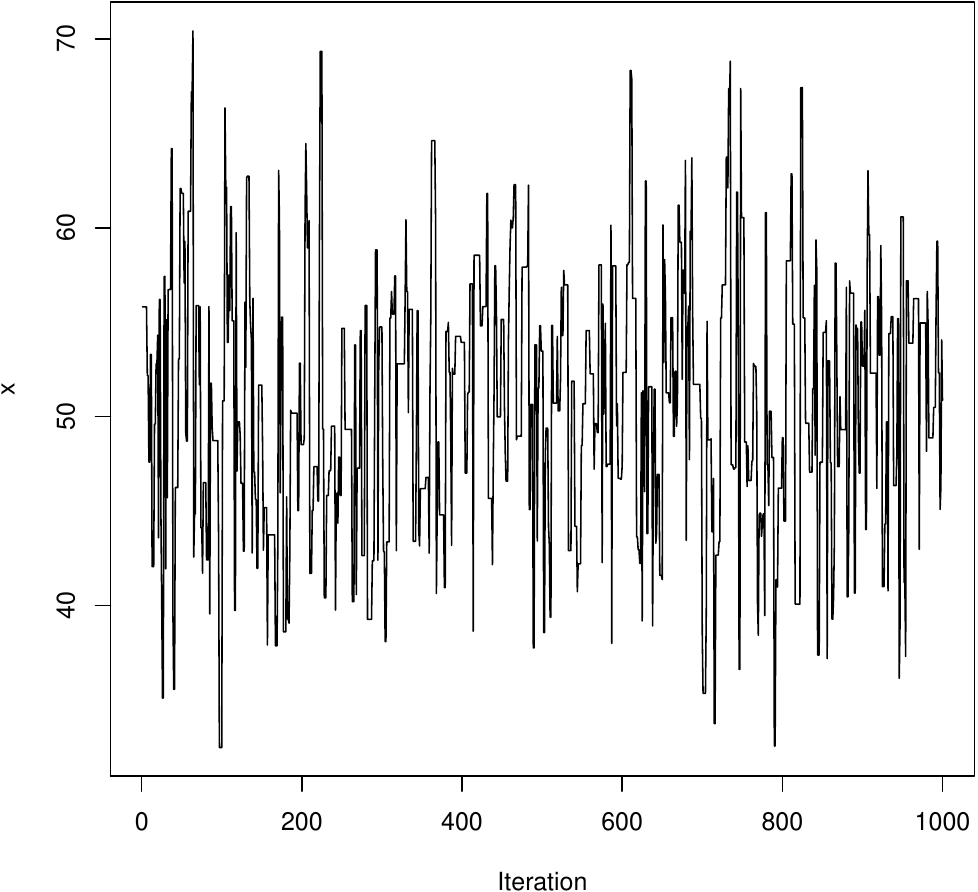}}
\hspace{2mm}
\subfigure [RWM traceplot for $d=50$.]{ \label{fig:slice_rwm50}
\includegraphics[height=4cm,width=7cm]{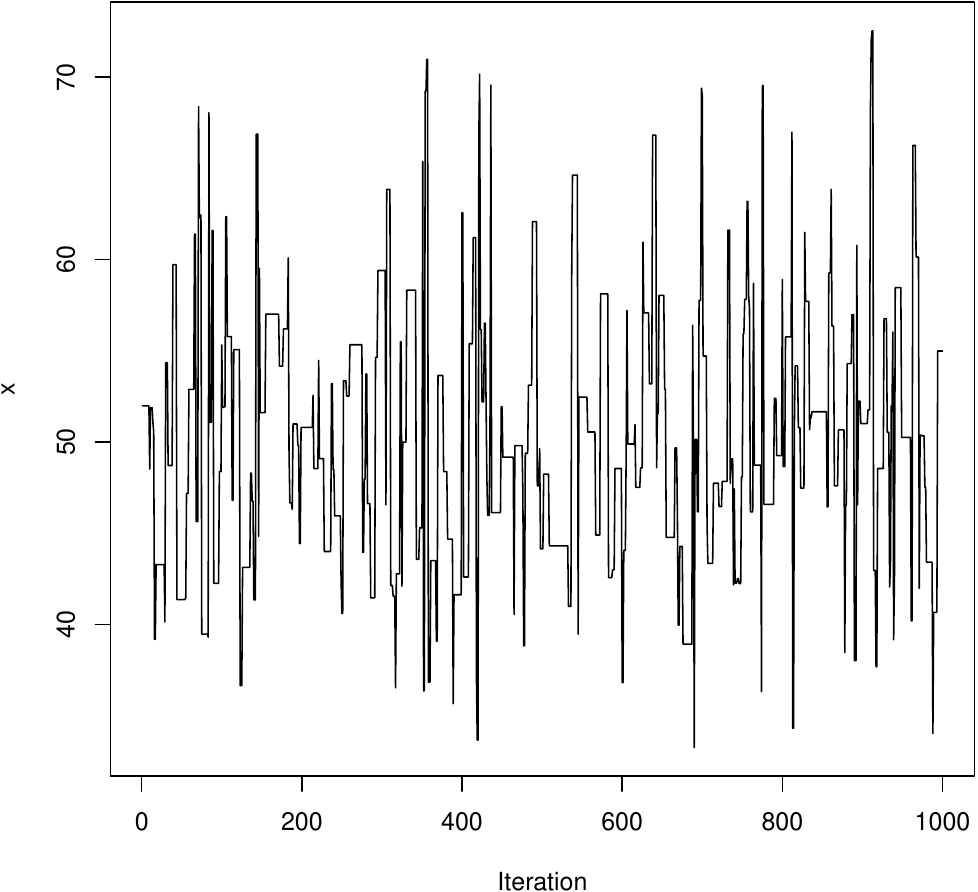}}\\
\caption{TMCMC and RWM based plots for $\pi_d(x)\propto x^{d-1}\exp\left(-x\right)$; $x>0$, 
with scales of the form $\ell/\sqrt{d}$, for $d=1,5,10,50$.
}
\label{fig:slice_trace}
\end{figure}

\begin{figure}%[htp]
\subfigure [TMCMC vs RWM autocorrelations for $d=1$.]{ \label{fig:slice_autocorr_tmcmc1}
\includegraphics[height=6cm,width=7cm]{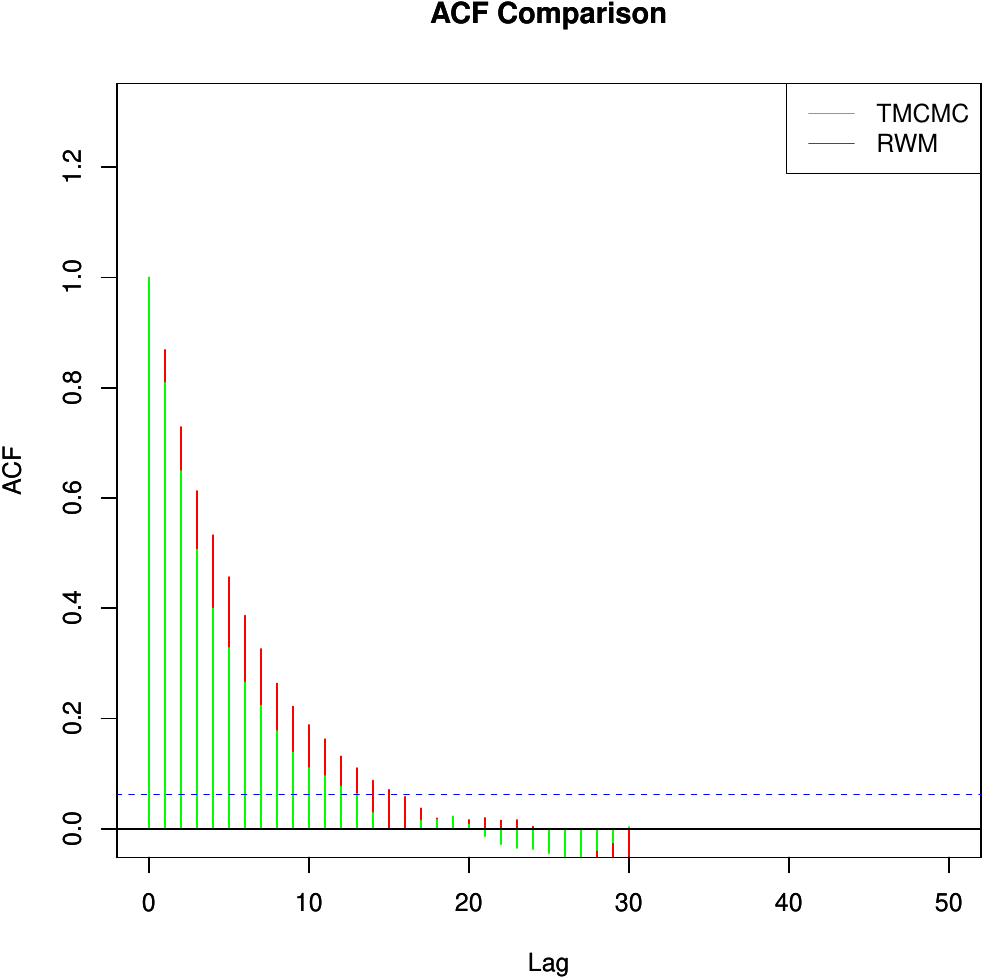}}
\hspace{2mm}
\subfigure [TMCMC vs RWM autocorrelations for $d=5$.]{ \label{fig:slice_autocorr_tmcmc5}
\includegraphics[height=6cm,width=7cm]{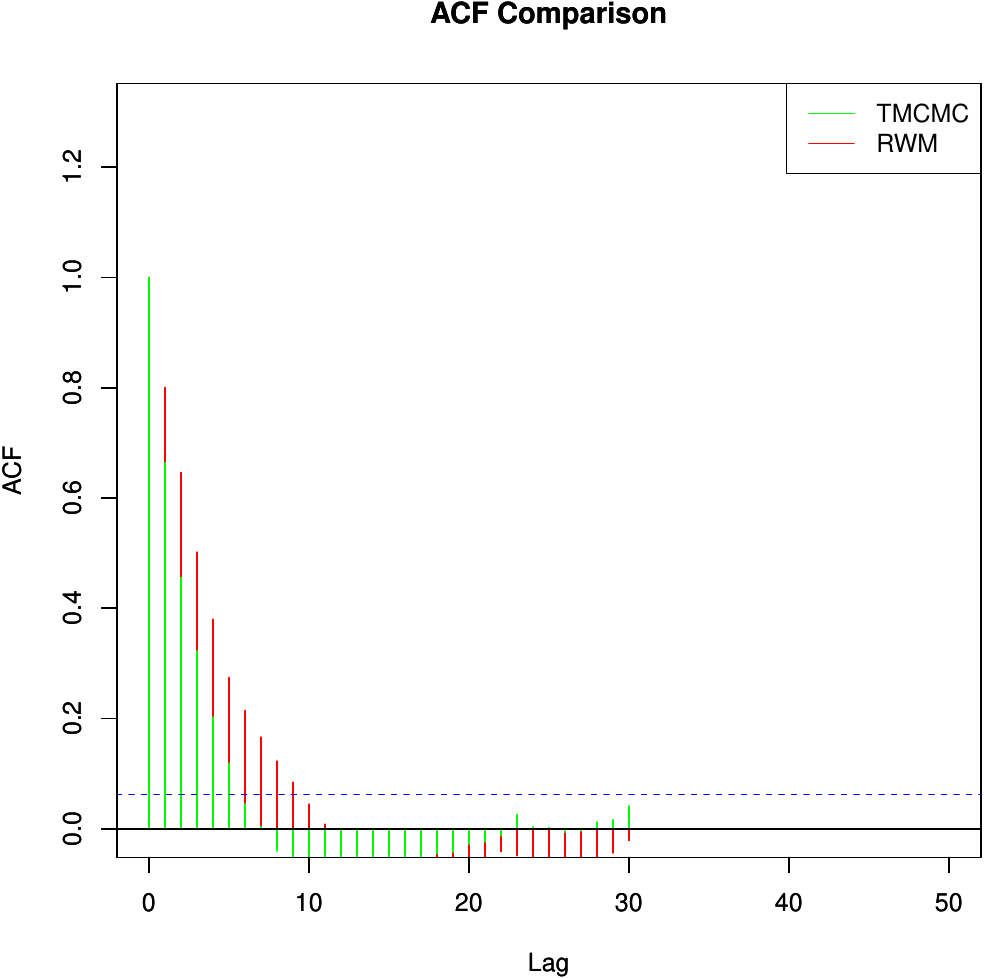}}\\
\vspace{2mm}
\subfigure [TMCMC vs RWM autocorrelations for $d=10$.]{ \label{fig:slice_autocorr_tmcmc10}
\includegraphics[height=6cm,width=7cm]{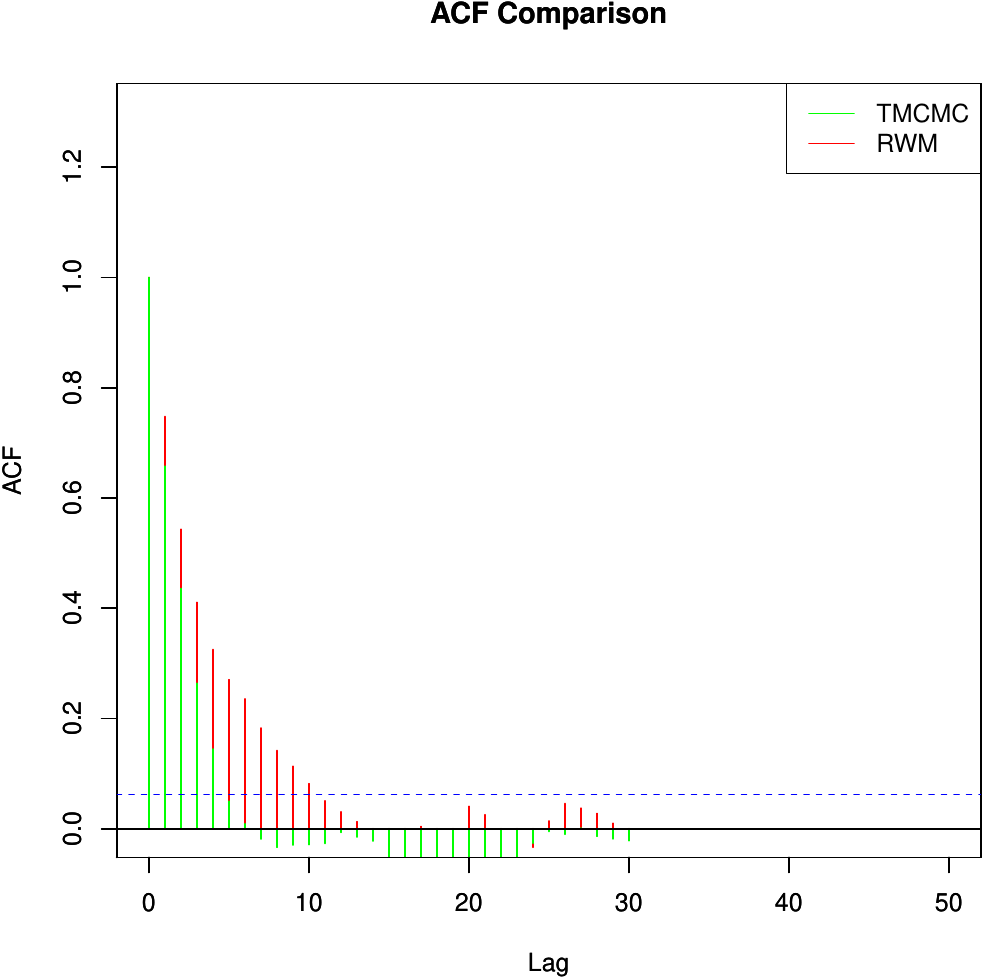}}
\hspace{2mm}
\subfigure [TMCMC vs RWM autocorrelations for $d=50$.]{ \label{fig:slice_autocorr_tmcmc50}
\includegraphics[height=6cm,width=7cm]{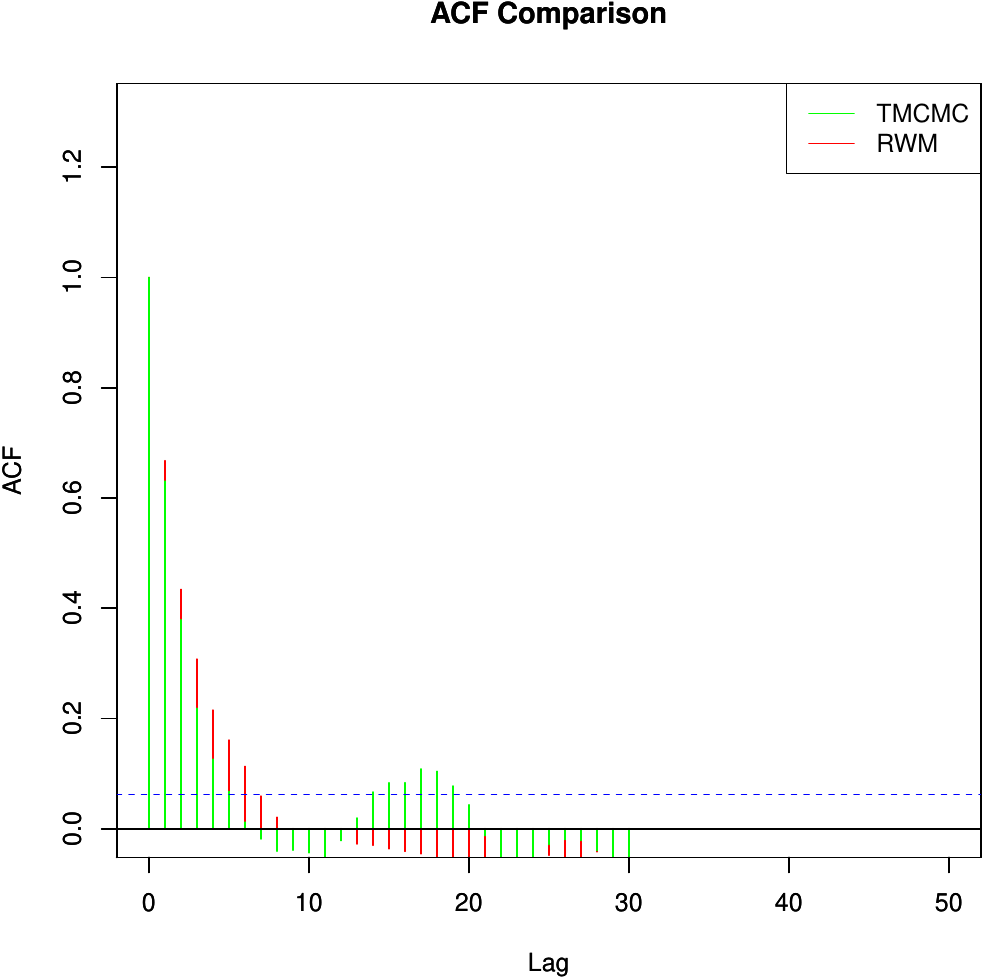}}\\
\caption{TMCMC and RWM based autocorrelation plots for $\pi_d(x)\propto x^{d-1}\exp\left(-x\right)$; $x>0$, 
with scales of the form $\ell/\sqrt{d}$, for $d=1,5,10,50$.
}
\label{fig:slice_autocorrelation}
\end{figure}

\begin{table}
%\begin{sidewaystable}[h]
\centering
\caption{KS distances: additive TMCMC vs RWM.} 
%with 
%Gaussian proposal when the target is $\pi_d(x)\propto x^{d-1}\exp\left(-x\right)$; $x>0$.}
%\vspace{2cm}
\begin{tabular}{|c|c|c|c|c|}
\hline
\multirow{2}{*}{\backslashbox{Proposal}{Target}} 
& \multicolumn{4}{|c|}{$\pi_d(x)\propto x^{d-1}\exp\left(-x\right);~x>0$}\\ 
\cline{2-5}
& $d=1$ & $d=5$ & $d=10$ & $d=50$ \\ 
\hline
TMCMC (Gaussian) & 0.102 & 0.078 & 0.142 & 0.086\\
RWM (Gaussian)  & 0.104 & 0.128 & 0.166 & 0.108\\
\hline
\end{tabular}
\label{table:slice}
%\end{sidewaystable}
\end{table}

\section{Summary and conclusion}
\label{sec:conclusion}

In this article, our contribution is two-fold. First, we have attempted to provide a comprehensive
review and discussion of the optimal scaling literature for various approaches of MCMC and contrasted
them with the corresponding versions of TMCMC. Second, and our main contribution, is a novel diffusion based approach 
to optimal scaling of additive TMCMC in non-regular cases, in contrast with the ESJD approach of \ctn{NealRoberts11}developed for RWM. 

Among the non-regular examples,  we have considered non-Gaussian proposal distributions and discontinuous target densities
with bounded support, and have proposed simple extensions of the results of \ctn{Dey13} for non-Gaussian proposals in 
conjunction with the logistic transformation of the random variables with bounded support to map them on the real line and apply
our diffusion results. We then used the It\^{o} formula to revert back to SDE associated with the original bounded
random variables, showing subsequently that the optimal scaling approach based on maximizing diffusion speed remains valid.
For the Cauchy proposal, even though we are still unable to prove the results explicitly, our simulation results
led us to conjecture that optimal scaling and optimal acceptance rate with the Cauchy proposal can be obtained using 
the same recipe discussed in Section \ref{sec:optimal}. 
%associated with SDE based results proved for non-Cauchy proposals, remain intact even for the Cauchy case. 
Comparison with the ESJD
approach of \ctn{NealRoberts11} for RWM showed that the complexity of RWM with the Cauchy proposal is much higher 
than that of additive TMCMC. %$O\left((d\log d)^2\right)$,
%while that additive TMCMC with the Cauchy proposal is only $O\left(d\right)$. 
The effect of much lesser
complexity of additive TMCMC is reflected in our simulation based comparison between
RWM and additive TMCMC with respect to the Cauchy proposal in the case of truncated normal target, where TMCMC 
outperforms RWM.
Our other simulation studies with target distributions taken to be a $t$ distribution with $5$ degrees of freedom,
a distribution with exponential tails, the posterior distribution associated with mixture of Weibull distributions,
all demonstrate additive TMCMC to be a far superior algorithm compared to RWM. Comparison of additive TMCMC and RWM 
with a slice sampler in the case of a $d$-dimensional density not only demonstrated that the former two are much more
effective compared to the popular slice sampling method, but also re-established the superiority of additive TMCMC
over RWM.

Although our results are with respect to target distributions that are products of {\it iid} densities, we are hopeful that the ideas and the results will go through even in the case of target densities that are products of independent but non-identical densities, as considered in \ctn{Dey13} and \ctn{Bedard2007}, as long as the individual densities have the same support. The same ideas are also expected to carry over to TMCMC within Gibbs algorithms, as considered in \ctn{Dey13}.

\section*{Acknowledgment}
We are sincerely grateful to the two reviewers whose constructive comments have led to a much improved version
of our manuscript.

\section*{Appendix}
\subsection*{HMC is a special case of TMCMC}

Let us denote the $L$-step leap-frog transformation in the HMC algorithm \ref{algo:hmc} 
associated with $(x_2,p_2)$ be denoted by $T_L$. Then $(x_2,p_2)=T_L(x,p_1)$, and in the TMCMC notion,
is the forward transformation, given $p_1\sim N(0,M)$. For convenience, we further consider the step 
$(x_2,p_2)\rightarrow (x_2,-p_2)$.
Thus, slightly abusing notation, we define $T_L$ to be the $L$-step leap-frog transformation applied to $(x,p_1)$
yielding $(x_2,p_2)$; then negating $p_2$ to finally yield $(x_2,-p_2)$. In practice, this negation is unnecessary
due to symmetry of $N(0,M)$ (see, for example, \ctn{Neal11}), which is why we did not mention this step
in Algorithm \ref{algo:hmc}.
To reach $(x,p_1)$ from $(x_2,-p_2)$, we draw $-p_2\sim N(0,M)$, and then apply $T_L$ to $(x_2,-p_2)$ to first
obtain $(x,-p_1)$ by running $(x_2,-p_2)$ forward for $L$ leap-frog steps (see \ctn{Liu01}, \ctn{Neal11}), and
then negating the resulting momentum to get back $(x,p_1)$. The Jacobian of the transformation is 1, thanks
to its volume-preserving property (see \ctn{Liu01}, \ctn{Neal11}). It is easy to see that detailed balance holds for
this algorithm, and that irreducibility and aperiodicity also hold.

%For moving from $(x,p_1)$ to $(x_2,p_2)$ where $p_1\sim N(0,M)$, the Markov transition kernel satisfies
%\begin{equation}
%\pi(x,p_1)\min\left\{1,\frac{\pi(x_2,p_2)}{\pi(x,p_1)}\right\}=\min\left\{\pi(x,p_1),\pi(x_2,p_2)\right\}
%\end{equation}

The above arguments show that only the forward move is necessary to move back and forth in the state space.
In fact, the forward move $T_L$ itself acts as the backward move given $-p_2\sim N(0,M)$. Moreover, $T_L$
acts simultaneously on the entire set of state variables, as both the forward and backward move. 
Recall that TMCMC makes use of random indicator variables that associate the forward transformation with $+1$ and
the backward transformation with $-1$.
However, since the backward move is also the forward move
here, such indicator is unnecessary for HMC.
Also note that the momentum variable acts as the vector $\be=(\epsilon_1,\ldots,\epsilon_d)'$ associated
with TMCMC. Note that the momentum variable can not be a singleton unlike general TMCMC algorithms and must be
of the same dimensionality as $x$, but this is certainly allowed by the general TMCMC theory; see \ctn{Dutta13}. 

Thus, the leap-frog based transformation $T_L$ simplifies several issues of the general TMCMC methodology
while subscribing to its basic philosophy. Hence HMC can be viewed as a special case of TMCMC.

%The inverse transformation 
%of $(x,p_1)$, which we denote by $T^{-1}_L$, is equivalent to first negating $p_1$, then taking the forward transformation,
%and then finally again negating $p_1$ (see \ctn{Neal11}). In other words, the inverse transformation
%works as follows: $(x,p_1)\rightarrow T_L(x,-p_1)=(x_2,-p_2)\rightarrow (x_2,p_2)$. That is, $T^{-1}_L(x,p_1)=(x_2,p_2)$. 
%Note that this is the backward transformation in the TMCMC notion, given $p_1\sim N(0,M)$.
%Interestingly, both the forward and backward transformations of $(x,p_1)$ lead to $(x_2,p_2)$.

\bibliography{irmcmc}

\end{document}